\documentclass[12pt,a4paper]{article}
\newcommand{\titleinfo}{Conditional Method Confidence Set} 
\title{\titleinfo}
\def\authora{Lukas Bauer}
\def\affa{University of Freiburg}

\def\emaila{\href{mailto:lukas.bauer@vwl.uni-freiburg.de}{lukas.bauer@vwl.uni-freiburg.de}}

\def\authorb{Ekaterina Kazak}
\def\affb{University of Birmingham} 

\def\emailb{\href{mailto:e.kazak@bham.ac.uk}{e.kazak@bham.ac.uk}}

\newcommand{\authorinfo}{\authora, \authorb}

\usepackage[english]{babel}   			   
\usepackage[OT1]{fontenc} 	               
\usepackage[utf8]{inputenc}              

\usepackage[left=3cm, right=3cm, top=3cm, bottom=3cm]{geometry} 
\usepackage{setspace} 
\usepackage{microtype}
\usepackage[medium]{titlesec}

\usepackage{amsmath,amsfonts,amsthm,amssymb,bbm}
\usepackage{marvosym,nicefrac}
\usepackage{bm} 
\usepackage{listings} 
\usepackage{stmaryrd}

\usepackage{color,xcolor,graphicx,epsfig,epstopdf}
\definecolor{DarkBlue}{rgb}{0.1,0,0.55}
\definecolor{DarkGreen}{RGB}{24,126,35}
\definecolor{lgrey}{RGB}{245,245,245}     
\definecolor{colKeys}{RGB}{0,0,255}       
\definecolor{colIdentifier}{RGB}{0,0,0}	  
\definecolor{colComments}{RGB}{34,139,34} 
\definecolor{colString}{RGB}{160,32,240}  

\usepackage[hyphens]{url}
\usepackage[pdfborder={0 0 0},bookmarks=true,breaklinks=true]{hyperref}
\hypersetup{colorlinks=true,
			citecolor = DarkBlue,
			pdftitle = {\titleinfo}, 
			pdfauthor = {\authorinfo},
			pdfpagelabels=TRUE,
			pdfstartview = {FitH},
			linkcolor = DarkBlue,
			urlcolor  = DarkBlue
			}

\usepackage{framed} 
\usepackage{caption}
\usepackage{booktabs}
\usepackage{bigstrut}
\usepackage{multirow,multicol}
\usepackage{longtable, tabularx, array, pbox, rotating,fullpage}
\usepackage{enumitem} 
\usepackage[flushleft]{threeparttable}
\usepackage{placeins}  

\usepackage{todonotes} 				   
\RequirePackage[l2tabu, orthodox]{nag} 
\usepackage{float}

\usepackage[style=apa, backend=biber, url=false, doi=false, isbn=false, eprint=false, maxcitenames=2, uniquename=false, uniquelist=false]{biblatex}  
\bibliography{BaKa_bibliography}

\newtheoremstyle{standard}
  {0,5cm}      
  {0,5cm}      
  {\upshape} 
  {}         
  {\bfseries}
  {}        
  {\newline} 
  {}         
\theoremstyle{standard}

\newtheorem{definition}{Definition}[section]
\newtheorem*{definition*}{}
\newtheorem{assumption}{Assumption}[section]
\newtheorem{theorem}{Theorem}[section]
\newtheorem{corollary}{Corollary}[section]

\newtheorem{lemma}{Lemma}[section]
\newtheorem{remark}{Remark}[section]

\usepackage{algpseudocode}[noEnd=true]  
\tikzset{algpxIndentLine/.style={style}}

\newtheoremstyle{notes}
  {0.5cm}      
  {0.5cm}      
  {} 
  {}         
  {\itshape}
  {}        
  {\newline} 
  {}         
\theoremstyle{notes}

\newcounter{saveenumi}

\DeclareMathOperator*{\Ew}{E}

\newcommand{\ec}[2]{\Ew \left[ \left. #1 \right| #2 \right]}

\newcommand{\1}{\mbox{$\mathrm{1\hspace*{-2.5pt}l}$\,}}

\DeclareMathOperator*{\argmax}{arg\,max\,}

\newcommand{\beq}{\begin{equation}}
\newcommand{\eeq}{\end{equation}}
\newcommand{\beqno}{\begin{equation*}}
\newcommand{\eeqno}{\end{equation*}}
\newcommand{\beqn}{\begin{eqnarray}}
\newcommand{\eeqn}{\end{eqnarray}}
\newcommand{\beqnn}{\begin{eqnarray*}}
\newcommand{\eeqnn}{\end{eqnarray*}}
\newcommand{\balgn}{\begin{align}}
\newcommand{\ealgn}{\end{align}}
\newcommand{\balgnn}{\begin{align*}}
\newcommand{\ealgnn}{\end{align*}}

\newcommand{\ben}{\begin{enumerate}}
\newcommand{\een}{\end{enumerate}}
\newcommand{\bit}{\begin{itemize}}
\newcommand{\eit}{\end{itemize}}

\newcommand{\bbm}{\begin{bmatrix}}
\newcommand{\ebm}{\end{bmatrix}}

\DeclareMathOperator{\sign}{sgn}
\DeclareMathOperator{\diag}{diag}
\graphicspath{{images/}}

\begin{document}

\begin{titlepage}
	\title{\titleinfo \thanks{
			The authors would like to thank Roxana Halbleib, Winfried Pohlmeier and Christian Gouriéroux for helpful comments.
			Early versions of the paper have been presented at the CFE 2024 in London and the FFM 2025 in Venice.
			All remaining errors are ours.
		} \\
	}
	\author{\Large\authora\thanks{Corresponding author: Faculty of Economics and Behavioural Sciences, University of Freiburg, Rempartstr. 16, D-79098 Freiburg, Germany. Email: \emaila.} \\[0.2cm] \large \affa  \and \Large \authorb\thanks{Department of Economics, Birmingham Business School, University House, 116 Edgbaston Park Rd, Birmingham B15 2TY, UK. Email: \emailb.} \\[0.2cm] \large \affb}
	\date{This version: \today \\}
	\maketitle
	\thispagestyle{empty}
	\begin{abstract}
		This paper proposes a Conditional Method Confidence Set (CMCS) which allows to select the best subset of forecasting methods with equal predictive ability conditional on a specific economic regime. The test resembles the Model Confidence Set by \textcite{Hansen.2011} and is adapted for conditional forecast evaluation.
		We show the asymptotic validity of the proposed test and illustrate its properties in a simulation study.
		The proposed testing procedure is particularly suitable for stress-testing of financial risk models required by the regulators. We showcase the empirical relevance of the CMCS using the stress-testing scenario of Expected Shortfall. The empirical evidence suggests that the proposed CMCS procedure can be used as a robust tool for forecast evaluation of market risk models for different economic regimes. 
	\end{abstract}
	\vfill
\end{titlepage}
\newpage

\section{Introduction}
\label{sec:intro}

Forecasting performance of econometric methods has been a critical focus in financial and macroeconomic research, largely driven by the strict demands of institutional regulators who aim at mitigating systemic risk \parencite{Ellis.2022}. In the aftermath of the 2008 Financial Crisis and the recent COVID-19 pandemic, the need for robust, reliable risk assessment models has become more pressing. Regulators, such as the Basel Committee on Banking Supervision, have responded with comprehensive guidelines, including stress-testing requirements, to ensure that the econometric tools employed by banks and financial institutions are sufficiently robust during periods of market turbulence. From an econometric standpoint, these regulatory requirements highlight the necessity for advanced statistical procedures capable of selecting and evaluating models under various market regimes. A prominent example of such regulations is the requirement to ``stress-test'' methods employed by financial institutions to forecast Expected Shortfall, or the worst expected loss associated with the market portfolio. The periods of stress are defined by a variety of risk factors provided by the regulator, associated with different \textit{liquidity horizons}, where a forecasting method is expected to be robust with respect to high values of risk factors across all liquidity horizons, which can be interpreted as economic regimes.

This paper contributes to the literature by proposing a Conditional Method Confidence Set (CMCS), a robust tool for evaluating the performance of forecasting methods under various financial stress scenarios. The CMCS builds on the Model Confidence Set (MCS) test by \textcite{Hansen.2011} and extends it to allow for state-dependent model comparisons, where forecast accuracy is assessed conditional on a specific economic regime. Our empirical application focuses on stress-testing Expected Shortfall forecasts across different liquidity horizons as required by \textcite{BaselCommittee.2019}, offering new insights into the robustness and reliability of downside risk forecasts in financial markets. 

Existing literature extensively debates the role of model misspecification, estimation errors, and the choice of loss functions in evaluating forecast accuracy. The impact of these factors on forecast comparison is critical when models are misspecified or rely on non-nested information sets. 
\textcite{Patton.2020} demonstrates that the rankings of models, within the class of consistent loss functions, can change based on the specific choice of the loss function. In regulated environments, where financial downside risks such as Value-at-Risk and Expected Shortfall need to be forecasted, robustness of forecast evaluation is crucial. \textcite{Gourieroux.2021} address this challenge by introducing robust forecast intervals that account for model misspecification, which is particularly important for stress-testing.
\textcite{Zhu.2022} propose a framework where conditioning variables or instruments are used to detect the best performing forecast. This approach builds on the earlier methods by introducing the role of conditioning variables to improve forecast quality, thereby shifting the focus from unconditional to conditional forecast evaluation.

The challenge of evaluating forecast accuracy becomes more complex in the context of misspecified models. \textcite{Giacomini.2006} introduced the first formal approach to comparing potentially misspecified forecasts, conditional on an information set, by developing the Conditional Predictive Ability (CPA) test. Their framework allows for a pairwise comparison of forecasts based on conditional expected loss.
Building on this, \textcite{Li.2022} with \textcite{Li.2020} extended the CPA test to allow for uniform inference on conditional loss differentials, proposing tests that can assess the hypothesis that the loss differential between models remains zero across all conditioning variables. 
\textcite{Giacomini.2010}  and \textcite{Richter.2020} further contributed to the growing body of literature by proposing a dynamic framework to test for equal predictive ability, which adjusts model rankings as forecast performance shifts over time due to structural changes in the data-generating process.

A more recent contribution by \textcite{Borup.2017} and \textcite{Borup.2024} introduced a dynamic forecast combination (DFC) framework. This method enables the comparison of models in a multivariate setting, extending the pairwise comparison by \textcite{Giacomini.2006} to identify the best set of models based on their conditional performance. 
\textcite{Hansen.2011} provided a crucial advancement in the model selection with the concept of the Model Confidence Set (MCS), which allows for the identification of a subset of models that exhibit equal predictive ability. This set-based approach is particularly valuable in contexts where multiple models need to be considered, and their predictive accuracy needs to be assessed simultaneously. Recent work by \textcite{Arnold.2024} extends this concept by incorporating sequential testing methods, allowing the MCS to be applied in a real-time framework with continuous updates as new data becomes available.

Our paper complements the unconditional approach of the MCS by \textcite{Hansen.2011}: even if one cannot distinguish between the predictive ability of  a set of models on average, the models may display very different forecast accuracy conditional on the state of the environment. In this case, the CMCS refines the MCS as, for each state, it delivers a set of models with indistinguishable predictive ability that may differ strongly from the unconditional one. 

The methods of \textcite{Li.2020} and \textcite{Li.2022} apply to a setting that is conceptually different from ours. They consider a conditional expected loss that is a continuous function, which rules out conditioning on indicator variables - yet, many observable states of the economy are discrete, and only a few states may be of interest. 
Moreover, the method of \textcite{Li.2022} yields a set of models that are weakly superior over all values of the conditioning set. Consequently, as pointed out by the authors, this set is potentially empty, as uniform weak superiority may be too strong an assumption under misspecification. In contrast, our CMCS is less restrictive, as it also applies if the relative conditional predictive ability between two models varies from state to state. 

While the testing procedures by \textcites{Giacomini.2006}{Borup.2024} indicate if there are differences in the conditional predictive ability, they do not directly identify the state (variable) that underlies these differences. They address this issue by proposing an selection rule based on the predicted forecasting loss, which, however, may lead to eliminations based on weak evidence. This procedure is conceptually analogous to the regression F-test, which \textcite{Hansen.2011} discuss. However, our CMCS procedure ensures that eliminations are based on sufficient evidence by adopting statewise the coherency requirement that \textcite{Hansen.2011} lay out.


The remainder of the  paper is organized as follows. In Section \ref{sec:model} we propose the Conditional Method Confidence Set and discuss its theoretical properties.
Section \ref{sec:sim} illustrates the theoretical properties of the proposed test and compares it to the existing Wald-type tests for conditional predictive ability.
Section \ref{sec:empir} provides empirical evidence of the performance of the proposed test in the context of stress-testing the Expected Shortfall forecasts. Section \ref{sec:conclusions} summarizes the main findings and gives an outlook on future research.

%
\section{Conditional MCS}
\label{sec:model}
This section develops a state-wise multiple testing procedure that delivers conditional method confidence sets. 
We extend the Model Confidence Set (MCS) procedure by \textcite{Hansen.2011} by applying it to statewise losses, i.e., subsamples of losses that are selected based on the state of the world at the forecast origin. We show the asymptotic validity of our approach.   

\subsection{Motivating example}\label{sec:motivating_example}
To illustrate the theoretical results presented below and introduce some notation consider an example of the risk manager making a choice between two forecasting methods: a GARCH(1,1) model with Gaussian innovations and a GARCH(1,1) model with innovations following a standardized Student-t distribution \( sst(\nu) \) with \( \nu \) degrees of freedom. Both models can be used to e.g.~forecast market volatility and then assume a location-scale parametrization to forecast the Expected Shortfall and the Value-at-Risk.

The risk manager estimates the parameters of each model \( i \), \( i \in \{1, 2\} \), using an estimation window of length \( r_{i} \), i.e., she uses a maximum of \( r= \max_{i} \{r_{1}, r_{2} \} \) observations. Using estimated parameters, each model makes a forecast of VaR and ES at time \( t+k \), where, for simplicity, we consider the case \( k=1 \), i.e., a one-day-ahead forecast. Following the realization of the financial return, she uses a statistical loss function \( L \) to assess the accuracy of the forecasts, which yields a scalar loss corresponding to each forecast. Thus, for each time \( t \) such that \( t+1\) is included in the out-of-sample window, she obtains the loss differential \( d_{12, t} = l_{1, t} - l_{2, t} \) to compare the forecasts of the two GARCHes.

Next, assume that the DGP of the financial return at time \( t+1 \) follows a GARCH model with innovations from the random variable \( \Psi_{t+1} \) that has a state-dependent distribution: conditional on a bivariate state variable \( S_{t} \), the innovations are distributed as

\begin{equation} 
    \Psi_{t+1} \sim 
    \begin{cases} \mathcal{N}(0, 1), \text{ if } S_{t} = 1, \\
    sst(\nu), \text{ if }  S_{t} = 2.
    \end{cases} 
\end{equation}

The state variable \( S_{t} \) takes two values, non-crisis and crisis, which we index and refer to by \( l=1, 2 \). From the DGP, it is clear that both models are misspecified: they assume the wrong distribution of the innovations with (unconditional) probability \( P(S_{t} = 2) \) and \( P(S_{t} = 1) \), respectively. In the following, we use the term method to stress that the models are misspecified and that the forecasts depend on estimated parameters. 

Assume that the regulator observes the state of the market, \( S_{t} \), with some error, and reports its observations. The error is such that, when the regulator reports non-crisis, the GARCH with Gaussian innovations has a smaller expected loss than the GARCH with \( sst(\nu) \) innovations, and vice versa when the regulator reports a crisis. Consequently, exploiting the information provided by the regulator can lead to improved forecasting accuracy if one chooses the superior method. 

Thus, the risk manager needs to evaluate the methods' forecasting performance in both states  \( l \in \{1, 2\}\), for which she makes \(n^{1} \) and \( n^{2} \) observations each of the loss differential \( d_{12, t} \). Therefore, she tries to make inference about the expected value of \( d^{1}_{12, \tau}, \tau \in {1, 2, \ldots, n^{1}} \) and \( d^{2}_{12, \tau}, \tau \in {1, 2, \ldots, n^{2}} \), which are the loss differentials that correspond to the state of non-crisis and crisis, respectively.

The CMCS approach performs two separate tests about the mean of the conditional loss differentials \( d^{1}_{12}, \tau \) and \( d^{2}_{12, \tau} \) each, based on the test statistics 
\begin{align}\label{eq:conditional_T_statistic}
T^{1} &= \dfrac{\bar{d}^{1}_{12}}{ \widehat{var}(\bar{d}^{1}_{12}) },   &T^{2}&=\dfrac{\bar{d}^{2}_{12}}{ \widehat{var}(\bar{d}^{2}_{12}) } ,
\end{align}
where  \( \widehat{var}(\bar{d}^{l}_{12}), \; l \in \{1, 2\} \) is a consistent estimator of the variance of \( \bar{d}^{l}_{12} \). For \(n^{1} \) and \( n^{2} \) large enough, the tests reveal for each state the method that is superior.  

If she decides to use the test by \textcites{Giacomini.2006} (which is nested by the multivariate extension of \textcite{Borup.2024}) to select one method for each state, she applies a two-step-procedure. First, she performs a Wald type test that uses the vector of test functions \( h_{t} = (1, \mathbbm{I}_{\{s_{t}=1\}})^{\prime} \in R^{2 \times 1}\) to obtain the instrumented loss differential \( z_{t} = h_{t} d_{t} = (d_{t}, d_{t}\mathbbm{I}_{\{s_{t}=1\}} )^{\prime} \). 
The test statistic is then
\begin{align}\label{eq:Wald_statistic}
    T^{h} &= n \bar{z}^{\prime} \hat{\Sigma}^{-1}  \bar{z},
\end{align}
where \( \hat{\Sigma} \) is a consistent estimator of the variance-covariance matrix of \( z_{t} \). 

If the test rejects that \( \mathbb{E}[z_{t}] = (0,0)^{\prime} \), she concludes that there are differences in the conditional predictive ability, but she does not have evidence in which state they exist. 

Second, she thus applies the proposed decision rule. With two methods and two disjoint states, the decision rule is based on the sign of \( \bar{d}^{1}_{12} \) and \( \bar{d}^{2}_{12} \), i.e., she chooses the method with the smaller conditional average loss. 

Compared to this combination of conditional test and decision rule by \textcites{Giacomini.2006}{Borup.2024}, the statewise testing of the CMCS ensures that choosing one model over the other is based on sufficiently strong evidence in finite samples, which we exemplify in Section \ref{sec:two_models_Wald_vs_t}.

The unconditional MCS approach reduces to a Diebold-Mariano test (\textcite{Diebold.1995}) that uses all observations to perform a test about the unconditional mean of \( d_{12, t} \) based on 
\begin{align}\label{eq:unconditional_T_statistic}
    T^{0} = \dfrac{\bar{d}_{12}}{ \widehat{var}(\bar{d}_{12})}, 
\end{align} 
where  \( \widehat{var}(\bar{d}_{12}) \) is a consistent estimator of the variance of \( \bar{d}_{12} \).
Thus, for \( n \) large enough, the MCS reveals the method that is better on average, but will make inferior forecasts in one of the states. The CMCS, however, yields for each state the method that is superior, and therefore enables the risk manager to make more accurate forecasts.

For ease of exposition, this illustration considers only two forecasting methods, which reduces the comparison problem to a pairwise one. 
For \( m \geq 3 \), the CMCS approach can be implemented using \( T^{l}_{max}, \; l \in \{1, 2\} \), which takes the maximum over the individual conditional t statistics, while the MCS test can be implemented analoguously based on the unconditional \( T_{max} \) test statistic. The multivariate test by \textcite{Borup.2024} is based on the Kronecker product of \( h_{t} \) and the \(m-1\) vector of loss differences between a baseline method and the remaining ones.

\subsection{Description of the environment}\label{sec:description_environment}

We consider a stochastic process \( \mathbf{W} \equiv {W_{t}:\Omega \to R^{s+1}, s \in N, t=1, 2, \ldots} \) on a complete probability space \( ( \Omega, \mathcal{F}, P) \), where \( W_{t} \equiv (Y_{t}, X_{t}^{'})^{'}\) is observable. \( Y_{t}: \Omega \to R \) is the variable of interest, while \( X_{t}: \Omega \to R^{s} \) are the predictor variables, and \( \mathcal{F}_{t}=\sigma(W^{'}_{1}, \ldots,  W^{'}_{t})^{'} \) (cf. \textcites{White.1994}{Giacomini.2006}{Borup.2024}).

Assume there is a finite set of \( m \) competing forecasting methods that make univariate forecasts of some functional \( F \) of the variable \( Y \), e.g., the conditional mean, median or quantile. At each time \( t \), method \( i \) makes a forecast of \( F( Y_{t+k}) \), i.e., \( k\)-steps-ahead. We denote the forecast as \( \hat{f}^{i}_{t, k, r^{i}} = f^{i}(W_{t}, W_{t-1}, \ldots, W_{t-r^{i}+1}; \hat{\beta}^{i}_{t, r^{i}} ) \) for \( i=1, \ldots, m\), where \( f^{i} \) is an \( \mathcal{F}_{t}\)-measurable function.  

The subscript \(r^{i} \) on \( \hat{f} \) indicates that the forecast is generated using \(r^{i}\) observations prior to time \( t \). Moreover, \( \hat{\beta}^{i}_{t, r^{i}} \) denotes the estimates that the \( i^{th} \) forecasting method uses to generate the forecast.  These estimates can be parametric, semi-parametric, or nonparametric.

Let \( r = \max \{r^{1}, \ldots, r^{m} \} \), i.e., \( r \) is the maximum length of an estimation window over the set of competing forecasting methods. Along the lines of \textcites{Giacomini.2006}{Borup.2024}, we require that \( r < \infty \) is finite. This rules out an expanding estimation window, but includes the rolling or fixed window estimation scheme with both fixed and time-varying \( r^{i} \). Consequently, we may compare nested models, as the estimation error does not vanish, which would lead to degenerate limiting distributions (\textcite{Giacomini.2006}).

For each method \( i \), we obtain a total of \( n \) pairs of the target variable \(Y_{t+k} \) and the \(k\)-steps-ahead forecasts made at time \( t \). We evaluate the forecasts using a real-valued, scalar loss function \(L(Y_{t+k}, \hat{f}^{i}_{t, k, r_{i}}) \). While we focus on statistical loss functions that are strictly consistent in the sense of \textcite{Gneiting.2007}, economic measures such as utility or a monetary criterion are also admissible. The shorthand notation \( L_{i, t} \) denotes the loss associated with the forecast that method \( i \) makes at time \( t \).

We aim to test for differences in the predictive ability of the competing forecasting methods when we observe a specific economic condition, e.g., an oil price shock or a financial crisis. We condition on \( d \) disjoint states\footnote{ While we restrict the admissible conditioning variables as compared to \textcite{Giacomini.2006} and \textcite{Borup.2024}, this allows us to obtain interpretable MCSs. For further discussion, see Section \ref{sec:testing_procedure}.}, i.e., we may think of the state variable \( S_{t} \) as a one-dimensional categorical random variable. We can thus represent realizations of \( S_{t} \) as a \( d-1 \) vector of indicator variables \( \tilde{s}_{t} \), which is \( \mathcal{F}_{t}\)-measurable. This yields the specific \textit{test functions} \( h_{t} = (1, \tilde{s}_{t}^{\prime})^{\prime} \) as used in the tests by \textcite{Giacomini.2006} and \textcite{Borup.2024}.

Let superscript \( l \) indicate that a random variable or observation is conditional on state \( l =1, \ldots, d \). If helpful, \( l=0 \) denotes the unconditional case. We denote the index set associated with all \( n \) out-of-sample losses \( I = \{ t_{0}, t_{0}+1, \ldots, t_{0} + n-1 \}  \), where \( t_{0} \) is the earliest forecast origin. Moreover, we write \( I^{l} = \{ t: S_{t} = s^{l}, \: t \in I \} \), i.e., the observations for which we observe state \( l \) at the forecast origin, with \( n^{l} = | I^{l}| \). In the following, for each \( l \), we use \( \tau \in \{1, \ldots, n^{l} \} \) to index the observations of \(L_{i,t} | S_{t} = s^{l} \).
Thus, \( \{ L^{l}_{\tau} \} \) denotes the subsequence of \( \{ L_{t}\} \) such that \(t \in I^{l}\), i.e., \( \{L_{t} | S_{t} = s^{l}\} \).

\subsection{Conditional hypotheses}

For a given state or economic regime $l$, e.g., a stress period as defined by regulator (\textcite{BaselCommittee.2019}), we want to formally test for differences in the predictive ability of \( m \) competing forecasting methods. 
Thus, for the state $l$ we define the conditional relative performance variables as
\begin{equation}
	d^{l}_{ij, \tau} \equiv L^{l}_{i, \tau} - L^{l}_{j, \tau} \; \text{for } i, j \in  \mathcal{M}^{\cdot, 0}, \tau \in \{ 1, \ldots, n^{l} \}, 
\end{equation} 
where \( \mathcal{M}^{\cdot, 0} \) is the initial set of all \( m \) methods. The competing methods are ranked in terms of their conditional expected losses.
\begin{definition}\label{def:def1hln_conditional}
    The set of conditionally superior objects is defined by \\
    \( \mathcal{M}^{l,*} \equiv \{ i \in \mathcal{M}^{\cdot, 0}: \mu^{l}_{ij} \leq 0 \text{ for all } j \in \mathcal{M}^{\cdot, 0} , \; l \in \{0, 1, \cdots, d\} \}\).     
\end{definition}

\noindent Definition \ref{def:def1hln_conditional} is the conditional equivalent to \textcite{Hansen.2011}. In this paper we impose an assumption that \( \mu^{l}_{ij} \equiv \mathbb{E}[d^{l}_{ij, n^{l}\tau}] < \infty \), and that \( \sign ( \mathbb{E}[d^{l}_{ij, n^{l}\tau}] ) \) does not depend on \( \tau \) for all \( i, j \in \mathcal{M}^{\cdot, 0} \), which implies that the conditional ranking of the forecasting methods is stable over time. 

The hypotheses which need to be tested to find the set $\mathcal{M}^{l,*}$ take the form 
\[ H^{l}_{0, \mathcal{M} }: \mu^{l}_{ij} = 0 \text{  for all  } i,j \in \mathcal{M}, \; l \in \{0, 1, \cdots, d\} \] and 
\[ H^{l}_{A, \mathcal{M} }: \mu^{l}_{ij} \neq 0 \text{  for some  } i,j \in \mathcal{M}, \; l \in \{0, 1, \cdots, d\}, \]
where \( \mathcal{M}^{l} \subset \mathcal{M}^{\cdot, 0}. \)

Equivalently, we can express the hypotheses in terms of \( \mu^{l}_{i \cdot} \equiv \mathbb{E}[d^{l}_{i \cdot, n^{l}\tau}] \), where \( d^{l}_{i \cdot, n^{l}\tau} = L^{l}_{i, \tau} - \dfrac{1}{m} \displaystyle \sum_{j=1}^{m} L^{l}_{j, \tau} \). We define the conditional method confidence set as any subset of \( \mathcal{M}^{\cdot, 0} \) that contains \( \mathcal{M}^{l, *} \).

\subsection{CMCS Testing procedure}\label{sec:testing_procedure}

Operationally, the CMCS procedure is very similar to the unconditional MCS procedure by \textcite{Hansen.2011}. The difference of the CMCS compared to the unconditional testing is that the CMCS is performed statewise on the time series of losses, which correspond to a state $l$.

Generally, a main concern in testing multiple hypotheses is the control of the familywise error rate (FWER), which is defined as making at least one false rejection. In the context of testing predictive ability, this means to eliminate at least one method with the smallest expected loss. 
One principle to control the FWER while avoiding pairwise comparisons is the so-called \textit{closure method} by \textcite{Marcus.1976}. To reject a hypothesis, every intersection hypothesis, i.e., a hypothesis that nests the individual hypothesis, must be rejected (see \textcite{Lehmann.2022}, ch.~9.2). This ensures that the rejection of a hypothesis is based on sufficiently strong evidence.

In the MCS testing procedure, \textcite{Hansen.2011} impose a similar requirement that they call ``coherency'' between test and elimination rule, and which needs to be fulfilled to devise a valid multiple testing procedure that avoids pairwise comparison.

We impose analogous assumptions about the statewise testing procedures that we provide in Section \ref{sec:assumptions_on_testing_procedure}. We adopt the coherency requirement statewise, which implies the we control the FWER for each state. As we consider disjoint states, CMCS-based decisions are robust against false rejections.

Moreover, we impose assumptions such that the conditional loss differentials have finite moments and display finite temporal dependence, thus central limit theorems for mixing random variables apply. 

\begin{assumption}\label{ass:GW_mixing_HLN_adj}
    For some \( r>2\) and \( \gamma >0 \), it holds that \( \{ W_{t} \equiv (Y_{t}, X_{t}^{'})^{'}\} \) and test functions \( \{ h_{t} \} \) are \( \alpha \)-mixing of order \( -r/(r-2) \).  
\end{assumption}

Define \( z_{ij, t} = h_{t} d_{ij,t} \). With Corollary \ref{cor:product_conditional_GW}, \( \{ z_{ij, t}\}_{i,j \in \mathcal{M}^{0}} \) is \( \alpha \)-mixing of order \( -r/(r-2) \), and thus \( \{ d^{l}_{ij, n^{l}\tau}\}_{i,j \in \mathcal{M}^{0}}  \) is \( \alpha \)-mixing of at most order \( -r/(r-2) \).

\begin{assumption}\label{ass:GW_finite_moments}
    For some \( r>2\) and \( \gamma >0 \), it holds that \( \mathbb{E} | d^{l}_{ij, n^{l}\tau}|^{r+\gamma} < \infty \) for all \( i, j, l, n^{l}, \tau \).
\end{assumption}

\begin{assumption}\label{ass:subsampled_loss_differentials}
    Moreover, assume about \( \{ d^{l}_{ij, n^{l}\tau}\}_{i,j \in \mathcal{M}^{0}} \) for all \( l, n^{l}, \tau \) that (i) \( \sign \mathbb{E}(d^{l}_{ij, n^{l}\tau}) \) is constant across $n^{l}\tau$, (ii) \( var(d^{l}_{ij, n^{l}\tau}) > 0 \), and (iii) \( var((n^{l})^{-1/2} \sum^{n^{l}}_{\tau=1}d^{l}_{ij, n^{l}\tau}) > \delta > 0 \) for all \( n^{l} \) sufficiently large.
\end{assumption}

\begin{assumption}\label{ass:probability_of_states}
    We assume that \( P(S_{t}=s^{l}) > \varepsilon > 0 \) for each \( t, l \), with \( \varepsilon >0\) such that \( n^{l}  \to \infty \text{ as } n \to \infty \).
\end{assumption}

Assumptions \ref{ass:GW_mixing_HLN_adj} to Assumptions \ref{ass:subsampled_loss_differentials} are essentially analogous to those of \textcites{Giacomini.2006}{Hansen.2011}{Borup.2024}; the data may exhibit both considerable heterogeneity and temporal dependence. 
Assumption \ref{ass:subsampled_loss_differentials} (i) is an additional assumption to \textcites{Giacomini.2006}{Borup.2024}. It implies that we permit arbitrary structural breaks, except for changes in the sign of the expected loss differentials, i.e., as long as the structural breaks do not alter the conditional ranking of the forecasting methods. Especially, our method imposes restrictions on the state variables, in the sense that it requires careful conditioning by the researcher or regulator. This is less restrictive than the strict stationarity that \textcite{Hansen.2011} impose on \( d_{ij, t} \) in the unconditional case.  

In contrast to \textcites{Giacomini.2006}{Borup.2024}, assuming that the conditional ranking of the forecasting methods is stable rules out arbitrary shifts in the means of the conditional loss differentials also under non-stationarity. This implies that the null and the alternative are exhaustive, i.e., if \( \mathbb{E}[\{ d^{l}_{ij, n^{l}\tau}\}] = 0 \) then it also holds that \( \mathbb{E}[\{ d^{l}_{ij, n^{l *}\tau}\}] = 0 \) for any sequence \( \{n^{l*}\} \), and analogously in the case of an inequality. 

Finally, Assumption \ref{ass:probability_of_states} ensures that the statewise sample size \(n^{l} \) goes to infinity as \( n \) goes to infinity for each state \(l\). 

\paragraph{}
The specific multiple testing procedure that we consider to test \( H^{l}_{0, \mathcal{M} } \) uses \( t^{l}_{i\cdot} = \bar{d}^{l}_{i\cdot, n^{l}} / \sqrt{\widehat{var}(\bar{d}^{l}_{i\cdot, n^{l}})} \). In brief, if \( T^{l}_{max, \mathcal{M}} = \max_{i} t^{l}_{i\cdot} \) exceeds a critical value \( c \), apply the elimination rule \(e^{l}_{max, \mathcal{M}} \) to decide which method to eliminate, where \( e^{l}_{max, \mathcal{M}}=\argmax_{i \in \mathcal{M}^{l}} t^{l}_{i\cdot} \) removes the method with the largest conditional standardized excess loss. Theorem \ref{th:th4hln_conditional} below provides the asymptotic properties of this testing procedure. 

Let \( \rho^{l}_{n^{l}} \) denote the conditional \( m \times m \) correlation matrix that is implied by the covariance matrix \( \Omega^{l}_{n^{l}} \) of Lemma \ref{lm:lm2hln_conditional}. Further, given the vector of random variables \( \xi^{l}_{n^{l}} \sim N_{m}(0, \rho^{l}_{n^{l}})\), we let \( F^{l}_{\rho_{n^{l}}} \) denote the distribution of \( \max_{i} \xi^{l}_{i, n^{l}} \). We define \( \bar{V}^{l}_{n^{l}} = (\bar{d}^{l}_{i\cdot, n^{l}}, \ldots, \bar{d}^{l}_{m\cdot, n^{l}})^{'}. \), \( l =1, \ldots, d \).\\

\begin{theorem}\label{th:th4hln_conditional}
    Let Assumptions \ref{ass:GW_mixing_HLN_adj}, \ref{ass:GW_finite_moments} and \ref{ass:subsampled_loss_differentials} hold and suppose that \( \widehat{(\omega^{l}_{i, n^{l}})}^{2} \equiv \widehat{var}((n^{l})^{1/2}\bar{d}^{l}_{i \cdot, n^{l}}) = n^{l} \widehat{var}(\bar{d}^{l}_{i\cdot, n^{l}}) \overset{p}{\to}(\omega^{l}_{i, n^{l}})^{2}, \) where \( (\omega^{l}_{i, n^{l}})^{2}, i =1, \ldots, m\), are the diagonal elements of \( \Omega^{l}_{n^{l}}\).
    
    Under \( H^{l}_{0, \mathcal{M}} \), we have \( T^{l}_{max, \mathcal{M}} \overset{d}{\to} F^{l}_{\rho_{n^{l}}} \), and under the alternative hypothesis \( H^{l}_{A, \mathcal{M}} \), we have \( T^{l}_{max, \mathcal{M}} \to \infty \) in probability. Moreover, under the alternative hypothesis, we have \( T^{l}_{max, \mathcal{M}} = t^{l}_{j \cdot} \), where \( j=e^{l}_{max, \mathcal{M}} \notin \mathcal{M}^{l,*} \) for \( n^{l} \) sufficiently large. 
\end{theorem}

\begin{proof}
    Let \( \Lambda^{l}_{n^{l}} \equiv \diag((\omega^{l}_{1, n^{l}})^{2}, \ldots, (\omega^{l}_{m, n^{l}})^{2}) \) and  \( \hat{\Lambda}^{l}_{n^{l}} \equiv \diag(\widehat{(\omega^{l}_{1, n^{l}})}^{2}, \ldots, \widehat{(\omega^{l}_{m, n^{l}})}^{2}) \). 
    
    From Lemma \ref{lm:lm2hln_conditional} it follows that \( \xi^{l}_{n^{l}} = ( \xi^{l}_{1,n^{l}}, \ldots, \xi^{l}_{m,n^{l}})^{'} \equiv (\Lambda^{l}_{n^{l}})^{-1/2} (n^{l})^{-1/2} \bar{V^{l}_{n^{l}}} \overset{d}{\to}N_{n^{l}}(0, \rho^{l}_{n^{l}})\), since \( \rho^{l}_{n^{l}} = (\Lambda^{l}_{n^{l}})^{-1/2} \Omega^{l}_{n^{l}} (\Lambda^{l}_{n^{l}})^{-1/2} \).\\

    From \( t^{l}_{i\cdot} = \dfrac{\bar{d}^{l}_{i\cdot, n^{l}}}{\sqrt{\widehat{var}(\bar{d}^{l}_{i\cdot, n^{l}})}} = (n^{l})^{1/2} \bar{d}^{l}_{i\cdot, n^{l}} / \hat{\omega}^{l}_{i, n^{l}} = \xi^{l}_{i,n^{l}} \frac{\omega^{l}_{i, n^{l}}}{\hat{\omega}^{l}_{i, n^{l}}}\), it now follows that \( T^{l}_{max, \mathcal{M}} = \max_{i} t^{l}_{i\cdot} = \max^{l}_{i}((\widehat{\Lambda}^{l}_{n^{l}})^{-1/2} (n^{l})^{1/2}\bar{V^{l}_{n^{l}}})_{i} \overset{d}{\to}F^{l}_{\rho_{n^{l}}}.\) \\ 

    Under the alternative, \( \bar{d}^{l}_{j\cdot, n^{l}} \overset{d}{\to} \mathbb{E}(\bar{d}^{l}_{j\cdot, n^{l}}) > 0 \) for any \( j \notin \mathcal{M}^{*}\), so that both \( t_{j \cdot, n} \) and \(T^{l}_{max, \mathcal{M}} \) converge to infinity at rate \((n^{l})^{1/2} \) in probability. Moreover, it follows that \( j=e_{max, \mathcal{M}} \notin \mathcal{M}^{l,*} \) for \( n \) sufficiently large.
\end{proof}

Theorem \ref{th:th4hln_conditional} shows that \( T^{l}_{max, \mathcal{M}} \) converges to the maximum of a normal distribution under the null, while it detects inferior methods under the alternative. 

\subsection{Bootstrap implementation}
We implement the CMCS testing procedure by performing the bootstrapped MCS testing procedure as put forth by \textcite{Hansen.2011}, but use the statewise losses separately for each state $l$. Consequently, we bootstrap a functional of the mean of weakly dependent time series, i.e., the conditional losses. As discussed by \textcite{Hansen.2011}, the bootstrap implicitly accounts for the correlations among the (conditional) loss differentials, which influence \( F^{l}_{\rho_{n^{l}}} \). \textcite{Goncalves.2002} show that the block bootstrap variance estimator for the sample mean is consistent under the type of dependence that we consider. The CMCS bootstrap procedure is outlined in Appendix \ref{app:bootstrap}.


\section{Simulation Study}
\label{sec:sim}
This section illustrates the finite sample properties of the proposed CMCS by means of Monte Carlo simulation. In Section \ref{sec:sim_power} we illustrate the power properties of the conditional tests compared to the unconditional ones. In Section \ref{sec:two_models_Wald_vs_t} we compare the properties of the t-test, which the CMCS is built upon, to the Wald test underlying the DFC by \textcite{Borup.2024}. Finally, Section \ref{sec:simulation_wald_vs_t_rejection_regions} sheds light on the reasons behind the differences between t-test and Wald type test in the context of conditional forecast evaluation.

\subsection{CMCS power properties} \label{sec:sim_power}

\textcite{Hansen.2011} define the power of the unconditional MCS procedure as the average number of elements in the unconditional MCS \( \mathcal{ \widehat{M}}_{1-\alpha}^{*} \).\footnote{The power of a multiple testing procedure can be defined in several ways, see, e.g., \textcite{Romano.2005}, for a discussion thereof.} Analoguously, we define the power of the CMCS procedure for each state as the average number of elements in the statewise CMCS \( \mathcal{ \widehat{M}}_{1-\alpha}^{l, *} \).

We consider a setting in which all competing methods have the same unconditional predictive ability, while their conditional predictive ability varies according to the current state. 
The setup of these simulations is similar to the one in \textcites{Giacomini.2006}{Borup.2017}. Accordingly, we define a state variable \( S_{t} \), where \(P(S_{t} = 1)=p \), and \(P(S_{t} = 2)=1-p \). 
Similarly to the simulation studies in the literature, we set \( p=0.5 \). For \( n \in \{150, 500 , 1000\} \) observations, we generate \(5000 \) sequences of losses according to \( \mathbf{L}_{t+1} = \boldsymbol{\mu}_{t+1} + \boldsymbol{\epsilon}_{t+1} \), where for each method \( i \in \{1, \ldots, m \} \)

\begin{equation}\label{eq:monte_carlo_losses}
	\boldsymbol{\mu}_{i, t+1} = 
	\begin{cases} -\mu(1-c_{i} ),  \text{ if } S_{t} = 1, \\
	\mu(1-c_{i} ), \text{ if } S_{t} = 2,
\end{cases} 
\end{equation}
with \( c_{i} = \frac{2(i-1)}{m-1}\), and error terms \(\boldsymbol{\epsilon}_{t+1} \), that follow a multivariate normal \( \mathcal{N} (\mathbf{0}, \Phi) \) with \( \Phi = \mathbb{I}_{m_{0}} \).
The choice of the \( c_{i} \) is such that the vector of expected conditional losses is equally spaced, and consequently \( \boldsymbol{\mu}_{t} \) is such that the largest difference in the conditional predictive ability is \( 2 \mu \), and that the conditional predictive ability is symmetric in the two states: in state 1 (\(S_{t} = 1\) ), model 1 is conditionally the best and model \(m \) the worst. In state 2, (\(S_{t} = 2\) ), the ranking is reversed. Moreover, it holds that the absolute value of the conditional loss differential between two models \( i, j \) is the same in both states for all \( i, j \in \mathcal{M}^{\cdot, 0}\).

We consider \( \mu \in [0.1, 0.5 ] \), which covers the same range of the expected loss differential between the two models as in \textcite{Giacomini.2006}. We perform both the unconditional MCS testing procedure, and the statewise CMCS for \(m=10\) models, and set the level of the test \( \alpha=0.05 \). 

Figure \ref{fig:monte_carlo_power_10_models} presents the power - the average number of models in the final set of models - of both the MCS and the CMCS for \( n \in \{150, 500, 1000\} \) observations in total, i.e., an expected number of conditional observations of \( n/2 \). The DGP implies that, in expectation, the power is the same in state 1 and state 2, though the CMCS will contain different models according to the state of the economy: it holds that \(\mathbb{E}[d^{1}_{ij}] = -\mathbb{E}[d^{2}_{ij}] \) for all \( i,j, \in \{1, \ldots, m \}\). Varying the conditional expected loss differentials and the state probability \( p \) would imply that the power of the CMCS is different in the two states.

\begin{figure}[h!]
    \centering
    \caption{Power properties}
	\includegraphics[trim={0.5cm 0.1cm 0.1cm 0.2cm},clip,width=1\linewidth]{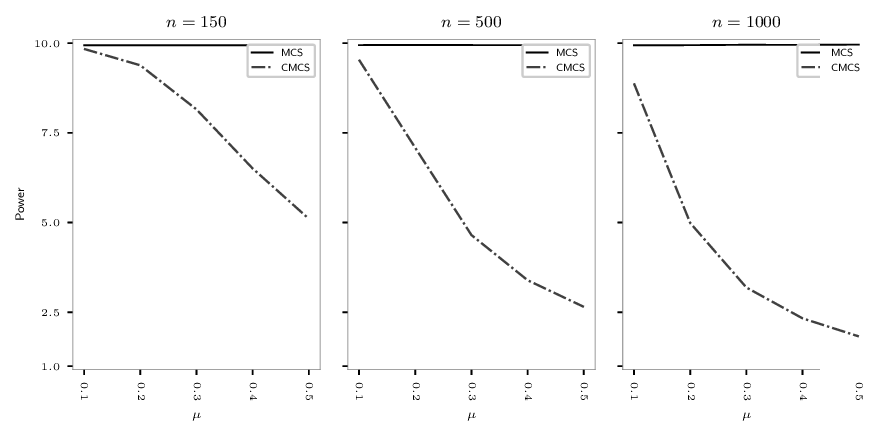} 
    \label{fig:monte_carlo_power_10_models}
    \caption*{ \scriptsize The figure displays the simulated power of the tests defined as the number of models in the confidence set averaged over 5000 simulation iterations, of the CMCS procedure (in dashed-dotted line) and the unconditional MCS procedure (in solid line). The x-axis shows \( \mu \) which controls the difference in the CPA between the competing methods according to the DGP in Equation (\ref{eq:monte_carlo_losses}). The number of methods is set to \(m=10\), the level of the test \( \alpha=0.05 \), and the two conditioning states have equal probability. The results are shown for state 1, but the DGP implies that the power curves are the same in expectation in state 2. The unconditional predictive ability is the same for all methods.}
\end{figure}

Both testing procedures display the desired behaviour. On the one hand, the unconditional MCS rarely eliminates any methods, as the null of unconditional equal predictive ability holds. 
On the other hand, the CMCS clearly has power, as the number of methods in the confidence set decreases with $\mu$. The proposed CMCS considerably narrows down the set of models for which we cannot reject equal conditional predictive ability. Comparing the different panels on Figure \ref{fig:monte_carlo_power_10_models} indicates that the power of the CMCS test increases with the sample size $n$, as the number of methods in the confidence set, for a given distance form the null $\mu$, is lower, the larger the $n$.

\subsection{Two models, two states: Wald vs. t-test}\label{sec:two_models_Wald_vs_t} 

The simulation study below illustrates the difference between the proposed CMCS approach and the two-step-approach of \textcites{Giacomini.2006}{Borup.2024}, i.e., the combination of a Wald type test and a decision rule. 
For clarity and ease of interpretation we consider a simple case of two forecasting methods and two disjoint states, where the forecast evaluation is based on the sign of the loss differential for a given state: if the Wald test provides evidence to reject the null of equal predictive ability, the procedure defines the MCS as the method that has the smaller average loss conditional on the current state. 

As above, we define a state variable \( S_{t} \), where \(P(S_{t} = 1)=p \), and \(P(S_{t} = 2)=1-p \), with  \( p\in \{0.2, 0.5\} \). For \( n=500 \) observations, we generate \(10,000 \) time series of losses according to \( \mathbf{L}_{t+1} = \boldsymbol{\mu}_{t+1} + \boldsymbol{\epsilon}_{t+1} \), where 

\begin{equation}\label{eq:monte_carlo_losses_2_models}
	\boldsymbol{\mu}_{t+1} = 
	\begin{cases} (-\mu, \mu)^{\prime},  \text{ if } S_{t} = 1, \\
        v \cdot (\mu, -\mu)^{\prime}, \text{ if } S_{t} = 2,
\end{cases} 
\end{equation}
where \( v \in [0, 1] \) and the error terms \(\boldsymbol{\epsilon}_{t+1} \) follow a multivariate normal \( \mathcal{N} (\mathbf{0}, \Phi) \) with \( \Phi = \mathbb{I}_{m_{0}} \).

We consider \( \mu \in [0.05, 0.3 ] \) and first focus on the case \(p=0.5\), when both states occur with equal probability.
In state 1, method 1 is superior, and the difference in the conditional predictive ability is \( \Delta_{1} =- 2 \mu \).  In state 2, if \( v =1\), the difference in the conditional predictive ability is \( \Delta_{2} = 2 \mu \), i.e. method 2 is superior, whereas the unconditional predictive ability is the same. For \( v \in (0, 1) \), method 2 is superior in state 2, while method 1 is unconditionally superior. For \(v=0 \), both methods have the same conditional equal predictive ability in state 2, and method 1 is uniformly weakly superior.

We use \( c^{\star} \) as shorthand notation for \( c( \mathcal{D}^{\star}, 1-\alpha) \), the critical value of a distribution \( \mathcal{D}^{\star} \) when we perform a test at the \( 1-\alpha \) confidence level. \( P(|T^{\star}| > c^{\star}) \) denotes the simulated probability with which the absolute value of test statistic \(T^{\star} \) exceeds the critical value \(c^{\star}\) of its distribution under the null hypothesis, which is the rejection rate of the test. \( T^{1}\) and \( T^{2}\) denote the test statistics for the statewise t-tests, while \( T^{h} \) is the Wald-type test statistic used by \textcite{Giacomini.2006}. The test statistics are defined in Section \ref{sec:motivating_example}, Equations \ref{eq:conditional_T_statistic} and \ref{eq:Wald_statistic}. We report results for the level of the test \( \alpha=0.05\). 

Table \ref{tab:simulation_wald_vs_t_0_p50} presents the results for \(v=0 \), where in state 1, we would expect the method confidence set to contain the model 1. The further is the distance from the null, captured by increasing $\Delta_{1}$ in table rows, the greater is the power of the t-test for state 1. The t-test for state 2, as expected, produces nominal rejection rates of about 5\%, as both methods have the same CPA for this state. The null of equal predictive ability holds in state 2, and conditional on a rejection of the Wald test, we always eliminate one method, which implies a type 1 error of 100\% for this testing procedure. 

Table \ref{tab:simulation_wald_vs_t_p_50} presents the results for \( v \in (0, 1]\).
Firstly, for small values of \( v \), \( v < 0.5 \), we observe that using the Wald type test has less power than the t test in state 1. Compared to the statewise t-test, the critical value increases with the degrees of freedom of the \( \chi^{2}_{df} \)-distribution, and state 2 fails to provide sufficient evidence to exceed this value.  

Secondly, small values of \( v \) are associated with frequently estimating the wrong sign of the expected loss differential in state 2. If \( \mu \) is large enough such that the differences in predictive ability in state 1 drive the rejections of the Wald test, this then leads to large ``type III'' or directional errors in state 2: the decision rule eliminates the superior or equally accurate method 2 far more often than the nominal level of the test. 

Thirdly, if \( v \) is large, the distance from the conditional null in state 2 is larger and the sign of the predicted loss differential is usually estimated correctly. In this case, the Wald type testing procedure based on insufficient evidence has more power than the statewise t tests, without committing large type III errors. 

These simulation results illustrate the need for additional testing as formalized in the closed testing procedure (\textcite{Marcus.1976}), and also highlight the loss of power of the Wald type test if the difference in predictive ability in one state is very small.  

\renewcommand{\arraystretch}{0.9}
\begin{table}[h!]
	\scriptsize
	\centering
	\caption{Rejection rate of Wald and statewise t-tests for $v=0$, $p=0.5$}
	\begin{tabularx}{0.60\textwidth}{l X X X }
	\toprule
						&\( P(|T^{1}| >c^{1}) \) & \( P(|T^{2}| >c^{2}) \) & \( P(|T^{h}| > c^{h})\)    \\
	\toprule
	\( \Delta_{1}=-0.1\)   &  0.126 &  0.050 &  0.102   \\
	\( \Delta_{1}=-0.2\)   &  0.354 &  0.054 &  0.282   \\
	\( \Delta_{1}=-0.3\)   &  0.657 &  0.050 &  0.550   \\
	\( \Delta_{1}=-0.4\)   &  0.880 &  0.052 &  0.797   \\
	\( \Delta_{1}=-0.5\)   &  0.976 &  0.054 &  0.945   \\
	\( \Delta_{1}=-0.6\)   &  0.996 &  0.054 &  0.990   \\
	\bottomrule
	\end{tabularx}
    \caption*{\scriptsize This table displays the rejection rates of the state-wise t-tests vs. the Wald type test. Holding \( v=0 \) fixed, which implies equal conditional predictive ability in state 2, the parameter \( \Delta_{1} \) characterizes the DGP as in Equation \ref{eq:monte_carlo_losses_2_models}. \( P(|T^{\star}| > c^{\star}) \) denotes the simulated probability with which the absolute value of test statistic \(T \) exceeds the corresponding critical value \(c\), which is the rejection rate of the test. \( T^{1}\) and \( T^{2}\) denote the test statistics for the statewise t-tests, while \( T^{h} \) is the Wald-type test statistic. The level of all tests is \( \alpha=0.05 \). The simulation is based on 10000 sequences of $n = 500$ losses, and both state 1 and state 2 have probability \(p=0.5 \).}
    \label{tab:simulation_wald_vs_t_0_p50}
\end{table}

\begin{table}[h!]
    \scriptsize
    \centering
    \caption{Rejection rate of Wald and statewise t-tests for $v\in (0,1]$ and \(p=0.5\)}
	\begin{tabularx}{.9\textwidth}{l X X X X }
		&\(v\) & \( P(|T^{2}| >c^{2}) \) & \( P\left(|T^{h}| > c^{h}\right)\)  & \( P(|T^{h}| > c^{h} \cap \bar{d^{2}} < 0) \)   \\
\toprule
\multirow{3.0}{*}{\( \Delta_{1}=-0.1\)}           &  0.050 &  0.054 &  0.100 &  0.049   \\
		&  0.100 &  0.052 &  0.099 &  0.046   \\
		&  0.250 &  0.057 &  0.107 &  0.038   \\
\multirow{3.0}{*}{\( P(|T^{1}| >c^{1})=0.126\)}   &  0.500 &  0.071 &  0.114 &  0.029   \\
		&  0.750 &  0.095 &  0.132 &  0.023   \\
		&  1.000 &  0.120 &  0.149 &  0.013   \\
\bottomrule
\multirow{3.0}{*}{\( \Delta_{1}=-0.2\)}           &  0.050 &  0.054 &  0.268 &  0.123   \\
		&  0.100 &  0.053 &  0.276 &  0.115   \\
		&  0.250 &  0.071 &  0.293 &  0.088   \\
\multirow{3.0}{*}{\( P(|T^{1}| >c^{1})=0.351\)}   &  0.500 &  0.125 &  0.334 &  0.051   \\
		&  0.750 &  0.222 &  0.402 &  0.028   \\
		&  1.000 &  0.356 &  0.507 &  0.013   \\
\bottomrule
\multirow{3.0}{*}{\( \Delta_{1}=-0.3\)}           &  0.050 &  0.056 &  0.547 &  0.248   \\
		&  0.100 &  0.055 &  0.553 &  0.217   \\
		&  0.250 &  0.096 &  0.572 &  0.141   \\
\multirow{3.0}{*}{\( P(|T^{1}| >c^{1})=0.659\)}   &  0.500 &  0.233 &  0.661 &  0.061   \\
		&  0.750 &  0.435 &  0.761 &  0.019   \\
		&  1.000 &  0.663 &  0.864 &  0.005   \\
\bottomrule
\multirow{3.0}{*}{\( \Delta_{1}=-0.4\)}           &  0.050 &  0.053 &  0.810 &  0.342   \\
		&  0.100 &  0.062 &  0.812 &  0.304   \\
		&  0.250 &  0.129 &  0.835 &  0.165   \\
\multirow{3.0}{*}{\( P(|T^{1}| >c^{1})=0.882\)}   &  0.500 &  0.351 &  0.891 &  0.040   \\
		&  0.750 &  0.660 &  0.951 &  0.008   \\
		&  1.000 &  0.888 &  0.987 &  0.001   \\
\bottomrule
\multirow{3.0}{*}{\( \Delta_{1}=-0.5\)}           &  0.050 &  0.056 &  0.946 &  0.396   \\
		&  0.100 &  0.071 &  0.950 &  0.328   \\
		&  0.250 &  0.174 &  0.958 &  0.147   \\
\multirow{3.0}{*}{\( P(|T^{1}| >c^{1})=0.975\)}   &  0.500 &  0.506 &  0.981 &  0.024   \\
		&  0.750 &  0.841 &  0.995 &  0.002   \\
		&  1.000 &  0.975 &  1.000 &  0.000   \\
\bottomrule
\multirow{3.0}{*}{\( \Delta_{1}=-0.6\)}           &  0.050 &  0.059 &  0.992 &  0.397   \\
		&  0.100 &  0.080 &  0.990 &  0.316   \\
		&  0.250 &  0.215 &  0.994 &  0.120   \\
\multirow{3.0}{*}{\( P(|T^{1}| >c^{1})=0.997\)}   &  0.500 &  0.654 &  0.998 &  0.009   \\
		&  0.750 &  0.943 &  1.000 &  0.000   \\
		&  1.000 &  0.997 &  1.000 &  0.000   \\
\bottomrule
\end{tabularx}
    \raggedright \\
    \caption*{\scriptsize This table displays the rejection rates of the statewise t-tests and the Wald type test. The parameters \( \Delta_{1} \) and \( v \) describe the DGP as in Equation \ref{eq:monte_carlo_losses_2_models}. \( P(|T^{\star}| > c^{\star}) \) denotes the simulated probability with which the absolute value of test statistic \(T \) exceeds the corresponding critical value \(c\), which is the rejection rate of the test. \( T^{1}\) and \( T^{2}\) denote the test statistics for the statewise t-tests, while \( T^{h} \) is the Wald-type test statistic. The last column displays the simulated probability with which method 2, which is superior in state 2, is eliminated in state 2 after applying the Wald-type test and decision rule by \textcite{Giacomini.2006}. The level of all tests is \( \alpha=0.05 \). The simulation is based on 10000 sequences of 500 losses, and both state 1 and state 2 have probability \(p=0.5 \). The two methods have the same unconditional predictive ability if \( v=1\).}
    \label{tab:simulation_wald_vs_t_p_50}
\end{table}


For $p = 0.2$, Tables \ref{tab:simulation_wald_vs_t_0_p20} and \ref{tab:simulation_wald_vs_t_p_20} present the results for \(v=0 \) and \( v \in (0, 1]\), respectively. In addition to the above mentioned mechanisms, these tables show the loss of power of the Wald test if there are large differences in predictive ability in the state that occurs with smaller probability (here \( p=0.2\) ), while the more frequent state displays only small differences in predictive ability.

\begin{table}[h!]
	\scriptsize
	\centering
	\caption{Rejection rate of Wald and statewise t-tests for $v=0$ and \(p=0.2\)}
	\label{tab:simulation_wald_vs_t_0_p20}
	\begin{tabularx}{0.7\textwidth}{l X X X }
	\toprule
						&\( P(|T^{1}| >c^{1}) \) & \( P(|T^{2}| >c^{2}) \) & \( P(|T^{h}| > c^{h})\)    \\
	\toprule
	\( \Delta_{1}=-0.1\)   &  0.086 &  0.053 &  0.069   \\
	\( \Delta_{1}=-0.2\)   &  0.178 &  0.054 &  0.133   \\
	\( \Delta_{1}=-0.3\)   &  0.329 &  0.052 &  0.238   \\
	\( \Delta_{1}=-0.4\)   &  0.523 &  0.050 &  0.395   \\
	\( \Delta_{1}=-0.5\)   &  0.695 &  0.049 &  0.568   \\
	\( \Delta_{1}=-0.6\)   &  0.845 &  0.049 &  0.743   \\
	\bottomrule
	\end{tabularx}
    \caption*{\scriptsize The table displays the rejection rates of the statewise t-tests and the Wald test. Holding \( v=0 \) fixed, which implies equal conditional predictive ability in state 2, the parameter \( \Delta_{1} \) characterizes the DGP as in Equation \ref{eq:monte_carlo_losses_2_models}. \( P(|T^{\star}| > c^{\star}) \) denotes the simulated probability with which the absolute value of test statistic \(T \) exceeds the corresponding critical value \(c\), which is the rejection rate of the test. \( T^{1}\) and \( T^{2}\) denote the test statistics for the statewise t-tests, while \( T^{h} \) is the Wald-type test statistic. The level of all tests is \( \alpha=0.05 \). The level of all tests is \( \alpha=0.05 \). The simulation is based on 10000 sequences of 500 losses, and state 1 has probability \(p=0.2 \) and state 2 \(1-p=0.8 \).}
\end{table}

\begin{table}[h!]
    \scriptsize
    \centering
    \caption{Rejection rate of Wald and statewise t-tests for $v\in (0,1]$ and \(p=0.2\)}
    \label{tab:simulation_wald_vs_t_p_20}
	\begin{tabularx}{.9\textwidth}{l X X X X }
		&\(v\) & \( P(|T^{2}| >c^{2}) \) & \( P(|T^{h}| > c^{h})\)  & \( P(|T^{h}| > c^{h} \cap \bar{d^{2}} < 0) \)   \\
\toprule
\multirow{3.0}{*}{\( \Delta_{1}=-0.1\)}           &  0.050 &  0.056 &  0.069 &  0.033   \\
		&  0.100 &  0.052 &  0.069 &  0.027   \\
		&  0.250 &  0.058 &  0.075 &  0.024   \\
\multirow{3.0}{*}{\( P(|T^{1}| >c^{1})=0.086\)}   &  0.500 &  0.084 &  0.093 &  0.017   \\
		&  0.750 &  0.118 &  0.116 &  0.009   \\
		&  1.000 &  0.171 &  0.161 &  0.007   \\
\bottomrule
\multirow{3.0}{*}{\( \Delta_{1}=-0.2\)}           &  0.050 &  0.050 &  0.122 &  0.054   \\
		&  0.100 &  0.061 &  0.142 &  0.053   \\
		&  0.250 &  0.081 &  0.151 &  0.032   \\
\multirow{3.0}{*}{\( P(|T^{1}| >c^{1})=0.175\)}   &  0.500 &  0.172 &  0.227 &  0.015   \\
		&  0.750 &  0.329 &  0.352 &  0.006   \\
		&  1.000 &  0.518 &  0.499 &  0.002   \\
\bottomrule
\multirow{3.0}{*}{\( \Delta_{1}=-0.3\)}           &  0.050 &  0.055 &  0.240 &  0.099   \\
		&  0.100 &  0.061 &  0.248 &  0.082   \\
		&  0.250 &  0.121 &  0.298 &  0.045   \\
\multirow{3.0}{*}{\( P(|T^{1}| >c^{1})=0.329\)}   &  0.500 &  0.324 &  0.457 &  0.010   \\
		&  0.750 &  0.608 &  0.671 &  0.002   \\
		&  1.000 &  0.847 &  0.860 &  0.000   \\
\bottomrule
\multirow{3.0}{*}{\( \Delta_{1}=-0.4\)}           &  0.050 &  0.062 &  0.397 &  0.154   \\
		&  0.100 &  0.070 &  0.415 &  0.128   \\
		&  0.250 &  0.168 &  0.498 &  0.056   \\
\multirow{3.0}{*}{\( P(|T^{1}| >c^{1})=0.517\)}   &  0.500 &  0.508 &  0.705 &  0.006   \\
		&  0.750 &  0.853 &  0.905 &  0.000   \\
		&  1.000 &  0.975 &  0.983 &  0.000   \\
\bottomrule
\multirow{3.0}{*}{\( \Delta_{1}=-0.5\)}           &  0.050 &  0.056 &  0.575 &  0.222   \\
		&  0.100 &  0.084 &  0.594 &  0.171   \\
		&  0.250 &  0.238 &  0.690 &  0.058   \\
\multirow{3.0}{*}{\( P(|T^{1}| >c^{1})=0.701\)}   &  0.500 &  0.701 &  0.890 &  0.003   \\
		&  0.750 &  0.963 &  0.986 &  0.000   \\
		&  1.000 &  0.999 &  0.999 &  0.000   \\
\bottomrule
\multirow{3.0}{*}{\( \Delta_{1}=-0.6\)}           &  0.050 &  0.059 &  0.741 &  0.279   \\
		&  0.100 &  0.091 &  0.760 &  0.194   \\
		&  0.250 &  0.325 &  0.845 &  0.047   \\
\multirow{3.0}{*}{\( P(|T^{1}| >c^{1})=0.842\)}   &  0.500 &  0.845 &  0.972 &  0.002   \\
		&  0.750 &  0.994 &  0.999 &  0.000   \\
		&  1.000 &  1.000 &  1.000 &  0.000   \\
\bottomrule
\end{tabularx}
    \caption*{\scriptsize This table displays the rejection rates of the statewise t-tests and the Wald test. The parameters \( \Delta_{1} \) and \( v \) describe the DGP as in Equation \ref{eq:monte_carlo_losses_2_models}. \( P(|T^{\star}| > c^{\star}) \) denotes the simulated probability with which the absolute value of test statistic \(T \) exceeds the corresponding critical value \(c\), which is the rejection rate of the test. \( T^{1}\) and \( T^{2}\) denote the test statistics for the statewise t-tests, while \( T^{h} \) is the Wald-type test statistic. The last column displays the simulated probability with which method 2, which is superior in state 2, is eliminated in state 2 after applying the Wald-type test and decision rule by \textcite{Giacomini.2006}. The level of all tests is \( \alpha=0.05 \). The simulation is based on 10000 sequences of 500 losses, and state 1 has probability \(p=0.2 \) and state 2 \(1-p=0.8 \). The two methods have the same unconditional predictive ability if \( v=0.25\).}
\end{table}
\renewcommand{\arraystretch}{1}  

\FloatBarrier

To provide a theoretical explanation for the smaller power of the Wald test for small values of \( v \), i.e, \( \Delta_{2} \), we provide the following lemma.

\begin{lemma}\label{lemma:two_models_power_wald_type}
	Assume that the DGP follows Equation \ref{eq:monte_carlo_losses_2_models}, where \( \Delta_{1} <0 \) and \( \Delta_{2} \), denote the conditional expected losses in state 1 and state 2, respectively, while \( \sigma^{2} \) denotes the common variance of the conditional losses. Moreover, assume that the covariance matrix \( \Sigma \)  of \( z_{t} = h_{t} d_{t} = (d_{t}, d_{t}\mathbbm{I}_{\{s_{t}=1\}} )^{\prime} \) is known, and that \(n_{1} = p n \), \( n_{2} = (1-p) n \) are fixed. Then, it holds that 
	\begin{align*}
		T^{h} & \equiv n \left[ D + E + F \right] \\
		&=n  \bigl [ p^{2} (\overline{d^{1}})^{2} \dfrac{ \sigma^{2} + p \Delta_{2}^{2} }{ p \sigma^2 \left( (1 - p) \Delta_1^2 + p \Delta_2^2 + \sigma^2 \right)} +2 p(1-p) \bar{d^{1}} \bar{d^{2}} \dfrac{\Delta_{1} \Delta_{2}}{\sigma^2 \left( (1 - p) \Delta_1^2 + p \Delta_2^2 + \sigma^2 \right)} \\
		&+  (1-p)^{2} (\overline{d^{2}})^{2}  \dfrac{ \sigma^{2} + (1-p) \Delta_{1}^{2} }{(1 - p) \sigma^2 \left( (1 - p) \Delta_1^2 + p \Delta_2^2 + \sigma^2 \right)}  \bigr ] \\ 
		&< n_{1} \dfrac{(\overline{d^{1}})^{2}}{\sigma^{2}} + n [E +F], 
	\end{align*}
\end{lemma}

\begin{proof}
	See the Algebra in Appendix \ref{app:wald_type_two_states_simulation}.
\end{proof}

\begin{remark}
	We see that \( n_{1} \dfrac{(\overline{d^{1}})^{2}}{\sigma^{2}} \) is proportional to \( T^{1} \) (see Equation \ref{eq:conditional_T_statistic}). In the special case \( \Delta_{2} = 0 \), it holds that \( E=F=0 \), i.e., \( T^{h} = n D < n_{1} \dfrac{(\overline{d^{1}})^{2}}{\sigma^{2}} \).
	If the critical value \( c(\chi^{2}_{1}, 1-\alpha) < n_{1} \dfrac{(\overline{d^{1}})^{2}}{\sigma^{2}} \), but \( n D < c(\chi^{2}_{2}, 1-\alpha) \), a rejection depends on the evidence that comes from state 2 in terms \( E \) and \( F \). Additionally, \( E \) and \( F \) need to compensate for the inflated denominator in term \( A \). For \( \Delta_{2} \) small enough, the Wald type test thus has less power than the statewise test in state 1. 
\end{remark}

\FloatBarrier
\subsection{Rejection regions} \label{sec:simulation_wald_vs_t_rejection_regions}
To further examine the behaviour of the Wald test and statewise t-tests, we plot the rejection regions as a function of the average conditional out-of-sample loss for the same DGP as in Section \ref{sec:two_models_Wald_vs_t}. We fix \( n=500 \), and vary \( p \), i.e., the probability of being in state 1, and the expected conditional loss differentials \( \mathbf{\Delta}=(\Delta_{1}, \Delta_{2})^{\prime} \), to convey how state probabilities and distance from the null impact the power of the tests.

Moreover, we use the true value of the covariance matrix of the loss differentials obtained from the simulations to construct the test statistics, and set the number of observations per state to their expected values \( np \) and \( n(1-p)\), respectively. We derive the expression for the Wald type test statistic in Appendix \ref{app:wald_type_two_states_simulation}. 

On each panel of Figure \ref{fig:rejection_regions}, the black square depicts the expected value of the conditional loss differentials \( \Delta \). To illustrate the distance from the null of conditional equal predictive ability, the outer black bars around the black squares correspond to the 2.5\% and 97.5\% quantiles of the sample mean of the loss differentials, and the inner bars to the 25\% and 75\% quantiles respectively. 
The areas where the null of equal conditional predictive ability is not rejected of the statewise t-tests are shaded in grey, while the non-rejection area of the Wald test lies within the black ellipsis.

The two upper panels on Figure \ref{fig:rejection_regions} show the case when there are large conditional differences in the predictive ability, while the two bottom panels illustrate a setting with a large expected loss difference in state 1, and a much smaller one in state 2, i.e., when the predictive ability difference is close to the null in state 2. 
In the first case, as illustrated by the upper panels, the Wald test has sufficient power to reject the null of CPA, and the decision rule yields reliable results as the sign of the expected conditional loss differential is usually estimated correctly. 

The lower left panel demonstrates one pitfall of combining the Wald type test with the elimination rule as in \textcite{Giacomini.2006}: while the rejection is driven by the large conditional loss differential in state 1, the superior model is eliminated very frequently in state 2 if the sign of the loss differential was estimated incorrectly. 

The lower right panel illustrates the loss of power if there are large differences in predictive ability in the state that occurs with smaller probability, \( p=0.3\) in this case, while the more frequent state displays only small differences in the predictive ability. This setting relates to the states of crisis / non-crisis in financial markets: the forecasting methods' performances differ strongly in times of crisis, for which relatively few observations are available, while their performance is very similar during calm periods. Relative to the Wald type test, the loss of power of the statewise t-test for state 1 is much smaller.

\begin{figure}[H]
    \centering
    \caption{Rejection region of Wald-type and statewise t-tests}
    \includegraphics[trim = {0.5cm, 1cm, .2cm, 2cm},clip, width=.85\linewidth]{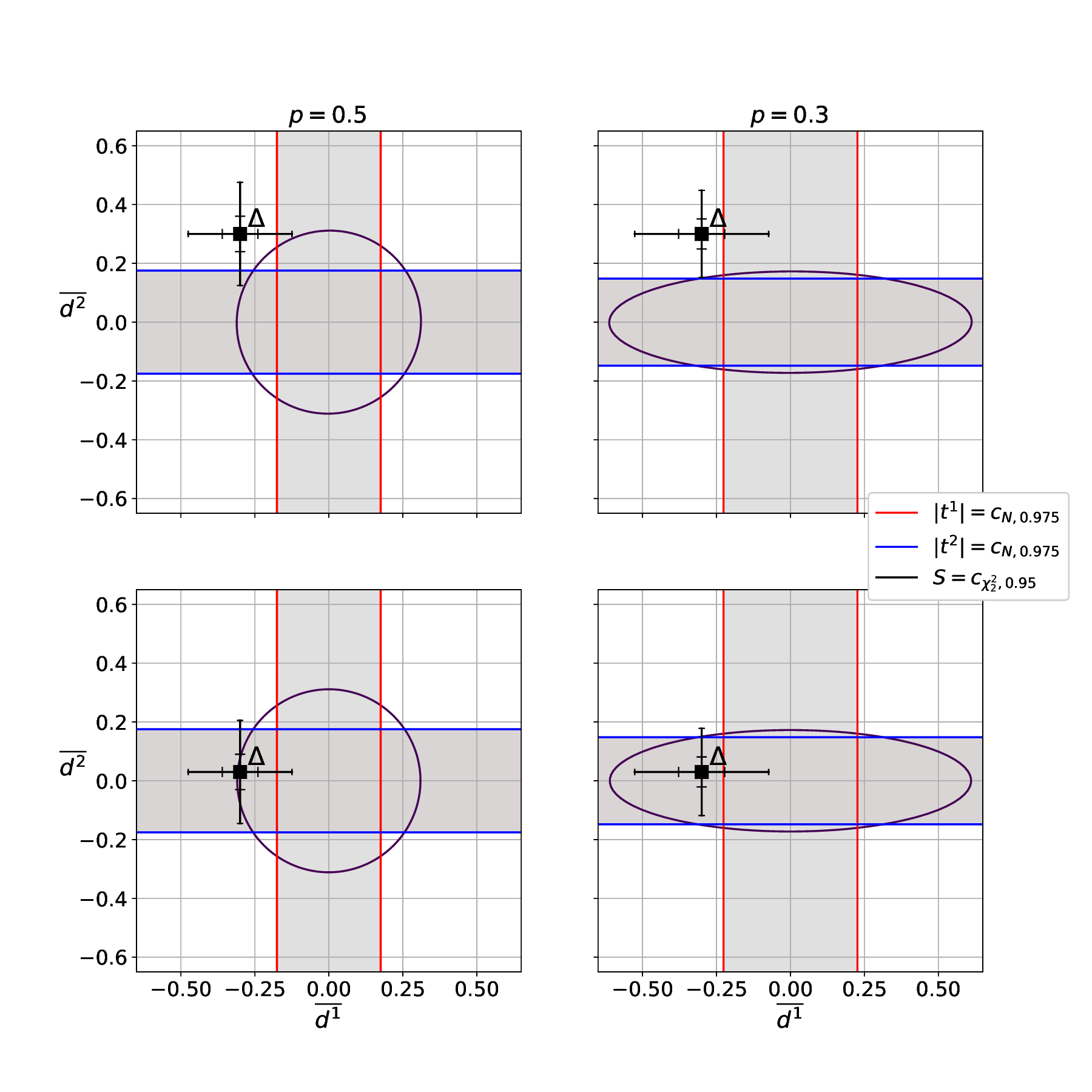}   
    \label{fig:rejection_regions}   	
	\caption*{\scriptsize This plot visualizes the rejection regions of the statewise t-tests (outside the respective grey area) and of the joint Wald type test (outside the black ellipsis). The solid square depicts the expected conditional loss differential \( \Delta \). The level of all tests is \( \alpha=0.05 \). The sample size is \( n=500\), while the number of observations in state 1 is \(n_{1}=p \cdot 500 \) and in state 2 \(n_{2}=(1-p) \cdot 500 \).}
\end{figure}



\section{Empirical Evidence}
\label{sec:empir}

This section demonstrates how CMCS can be used in the context of forecasting and stress testing Expected Shortfall (ES). 

\subsection{Downside measures of market risk}
Value-at-Risk (VaR) has been used by the Basel Committee on Banking Supervision (BCBS) to assess market risk since 1996 and it measures the worst possible loss of a portfolio, which can happen with small probability. Empirically VaR is estimated as a 1\% quantile of the profit and loss distribution of a portfolio. 
Let  $r_t = \ln(P_t) - \ln(P_{t-1})$ be the daily log return process,
where $P_t$ is the closing price on day $t$, $t= 1, \ldots T$.
We assume that daily returns $r_t$ follow a specific conditional distribution with a cumulative distribution function $D_t$: $r_t|\mathcal{F}_{t-1} \sim d_t$, where $d_t$ is the probability density function corresponding to $D_t$ and  $\mathcal{F}_{t-1}$ denotes past filtration. A one-day ahead forecast of the return quantile at a level $p$, known as the VaR, is defined as $\text{VaR}_{t+1}(p) = D_{t+1}^{-1}(p)$. 
ES has been introduced in 2016 as a more ``prudent'' risk measure, which estimates the expected value of the loss, conditional on VaR threshold being crossed, usually at the $p = 2.5\%$ level. The ES is then defined as 
$\text{ES}_{t+1}(p) = \ec{r_{t+1}}{r_{t+1}<\text{VaR}_{t+1}(p),\mathcal{F}_{t}}$.
\textcite{BaselCommittee.2019} requires the banks to report one day ahead VaR forecasts at $p  =1\%$ and 10 days ahead ES forecasts at $p = 2.5\%$. The forecasts of VaR are checked to fit the definition of a quantile with various tests, e.g.~out of 100 reported forecasts the nominal level of 1\% implies 1 day where the loss is lower than the provided VaR. 
The quality of VaR forecasts determines the scaling, or the penalty factor for the banks' capital requirements, which in turn are measured based on ES.
Additionally ES forecasts should be stress-tested based on the risk factors of different liquidity horizons. 
BCBS specifies several liquidity horizons and associated risk factors for the ES stress-testing, which are summarized in Table \ref{tab:LH}.

\begin{table}[h!]
	\centering
	\caption{Liquidity horizons and examples of risk factors}
	\begin{tabular}{l|l|p{12cm}}\hline \hline
	Horizon $j$	& $LH_j$ & Possible risk factors \\ \hline
        $j=1$   & 10    & Sovereign bond interest rates, foreign exchange rates \\
		$j=2$   & 20    & VIX, small cap equity prices  \\
		$j=3$   & 40    & Foreign exchange rate volatility, Corporate bonds credit spread \\
		$j=4$   & 60    & Commodities price \\
		$j=5$   & 120   & Other commodities price and volatility \\ \hline \hline
	\end{tabular}%
	\label{tab:LH}%
    \caption*{\scriptsize More information on the liquidity horizons is available in e.g. Minimum capital requirements for market risk (2019) by BCBS, p.98.}
\end{table}%
The period of ``stress'' is defined as the most severe time period of 252 days, where the risk factor was at its highest, and the considered sample of the risk factor should include the time period starting from January 2007. 
The coefficient to be reported, $ES_{BCBS}$ is defined in \eqref{eq:ES_bcbs}:
\begin{equation}
	ES_{BCBS} = \sqrt{ (ES(j = 1))^{2} + \sum_{j\geq2} \left( ES(j) \sqrt{\frac{LH_{j} - LH_{j-1}}{T}}\right)^{2}},
	 \label{eq:ES_bcbs}
\end{equation}
where $ES(j = k), \; k = 2,...,5$ denotes the ES forecast for the ``stress'' period of the risk factor at liquidity horizon $LH_k$, $T$ denotes the forecasting horizon, e.g. 1 or 10 days. The scaling factors $LH_j$ are defied in Table \ref{tab:LH}. Table \ref{tab:LH_ours} in the Appendix reports risk factors implemented in the empirical study below. 

\begin{figure}[h!]
	\caption{Illustration of stress periods.}
	\centering
	\includegraphics[trim={0cm 0cm 0cm 0cm},clip,width = 0.8\textwidth]{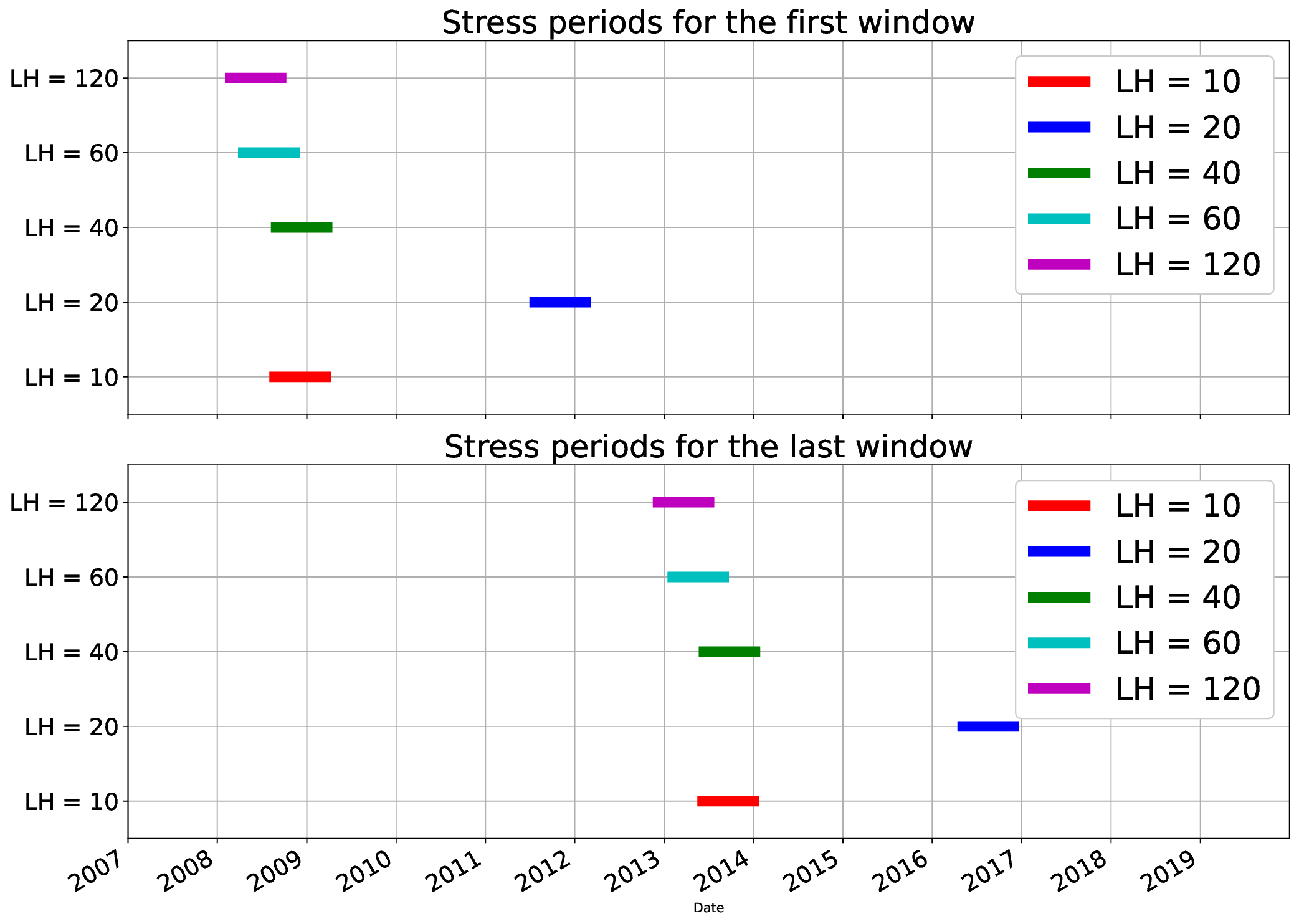}\vspace{-5pt}
	\caption*{\scriptsize Coloured lines correspond to stress periods for different liquidity horizons. Panels correspond to the evaluation periods.}
	\label{fig:stress}
\end{figure}
Figure \ref{fig:stress} provides a graphical illustration of the stress period associated with different liquidity horizons. Notably, the more ``stressful'' periods for risk factors such as volatility or small cap prices are quite different from the high risk periods associated with interest or foreign exchange rates. Given the heterogeneity of the risk factors we expect the conditional forecast combination decomposition vary significantly accords $LH$.

In this paper we consider daily return data of 30 large cap stock returns. The details about the stocks and some selective descriptive statistics are presented in Appendix \ref{app:data}. 
The testing period, which is used to construct the forecast combinations and to stress test the methods, is from January 2000 until December 2014 and the performance of forecast combinations is evaluated during the period from January 2015 until December 2019 ($H = 1250$).
We consider 20 models for the forecast combination, which are estimated on a 2 or 4 years of data prior to 2014 in a rolling window fashion, i.e.~model parameters are updated with every shift of the estimation and testing windows. 

\begin{table}[h!]
	\centering
	\caption{Models considered for the forecast combination}
	\begin{tabular}{p{2cm}|p{3cm}|p{9cm}} \hline \hline
		Model name & Estimation window length & Parameters \\ \hline
		GARCH & T = 500, 1000 & innovation distributions: Normal, Student-t \\
		EGARCH & T = 500, 1000 & innovation distributions: Normal, Student-t \\
		RiskMetrics & T = 500, 1000 & fixed parameter $\lambda = 0.97$ \\
		EVT   & T = 500, 1000 & Tail parameters are estimated on returns standardized by conditional volatilities of GARCH(1,1) with Normal innovations \\
		HS    & T = 500, 1000 &  \\
		FHS   & T = 500, 1000 & VaR and ES are estimated on returns standardized by conditional volatilities from  GARCH(1,1), EGARCH(1,1) with Normal innovations and Risk Metrics \\ \hline \hline
	\end{tabular}%
	\label{tab:m_i}%
\end{table}%

The model choice is motivated by the empirical properties of the daily return data, which exhibit strong overkurtosis and non-zero skewness. We also include historical simulation (HS) and filtered historical simulation (FHS) models, which are widespread in the professional world \parencite{BaroneAdesi.1999}. The details on model implementation are presented in Table \ref{tab:m_i}.

\subsection{Forecasting and stress-testing ES}  \label{sec:ea_es} 

To construct a forecast combination for different stress periods we compare the proposed CMCS test with the dynamic forecast combination (DFC) by \textcite{Borup.2024}. Both testing procedures are implemented based on the consistent loss function for the ES, which is elicitable jointly with VaR \parencite{Fissler.2016}:
\begin{align}\label{eq:es.loss}
	 L\left(VaR_{t}(p), ES_t(p), r_t\right) & = G_1(VaR_t)\cdot \left(H_t - p\right) - H_t\cdot G_1(r_t) \\
	 &+ G_2(ES_t) \left(ES_t-VaR_t +H_t\cdot\frac{VaR_t-r_t}{p}\right)  \nonumber\\
	 &- \xi_2(ES_t) + a(r_t), \nonumber
\end{align} 
where $p$ denotes the probability level and equals to 2.5\% as specified in \textcite{BaselCommittee.2019}, $H_t = \1(r_t\leq VaR_t) $ denotes the hit at time $t$, $G_1(x) = x$, $G_2(x) = \exp(x)/(1+exp(x))$, $\xi_2(x) = \ln(1+ \exp(x))$, $a = \ln(2)$ \parencite[e.g.][]{Taylor.2020}.
For both CMCS and Wald test we use the significance level $\alpha = 0.05$.
The CMCS is implemented with a bootstrap of $B = 100$ iterations. The forecast combination for both testing procedures is constructed as an average of the ES forecasts which remain in the Model Confidence Set. 

Appendix Tables \ref{tab:cmcs1} and \ref{tab:cmcs10} report average out-of-sample conditional ES forecasts for all liquidity horizons and across all stocks for the forecast combination based on the proposed CMCS for 1 and 10 days ahead respectively. The last column of the table reports the coefficient $ES_{BCBS}$ defined in \eqref{eq:ES_bcbs} and indicates the overall riskiness of the investment across the liquidity horizons. The variability in ES forecasts across liquidity horizons highlights that assets do not exhibit uniform risk exposure. Some assets might exhibit increased risk exposure to shorter-term factors (e.g.~interest rates, equity prices), while others might be more susceptible to longer-term factors (e.g.~commodity prices). 
 For instance, the ES forecast for \texttt{BAC} (Bank of America) is more sensitive to stress periods in foreign exchange rates and sovereign bond interest rates, rather than commodity prices.
In contrast, a stock like \texttt{GE} exhibits relatively small changes across liquidity horizons, suggesting it is less sensitive to the associated risk factors.

\begin{figure}[h!]
	\caption{Time series of 10-day ahead ES forecasts for BAC: CMCS.}
	\centering
	\includegraphics[trim={0cm 0cm 0cm 0cm},clip,width = 0.8\textwidth]{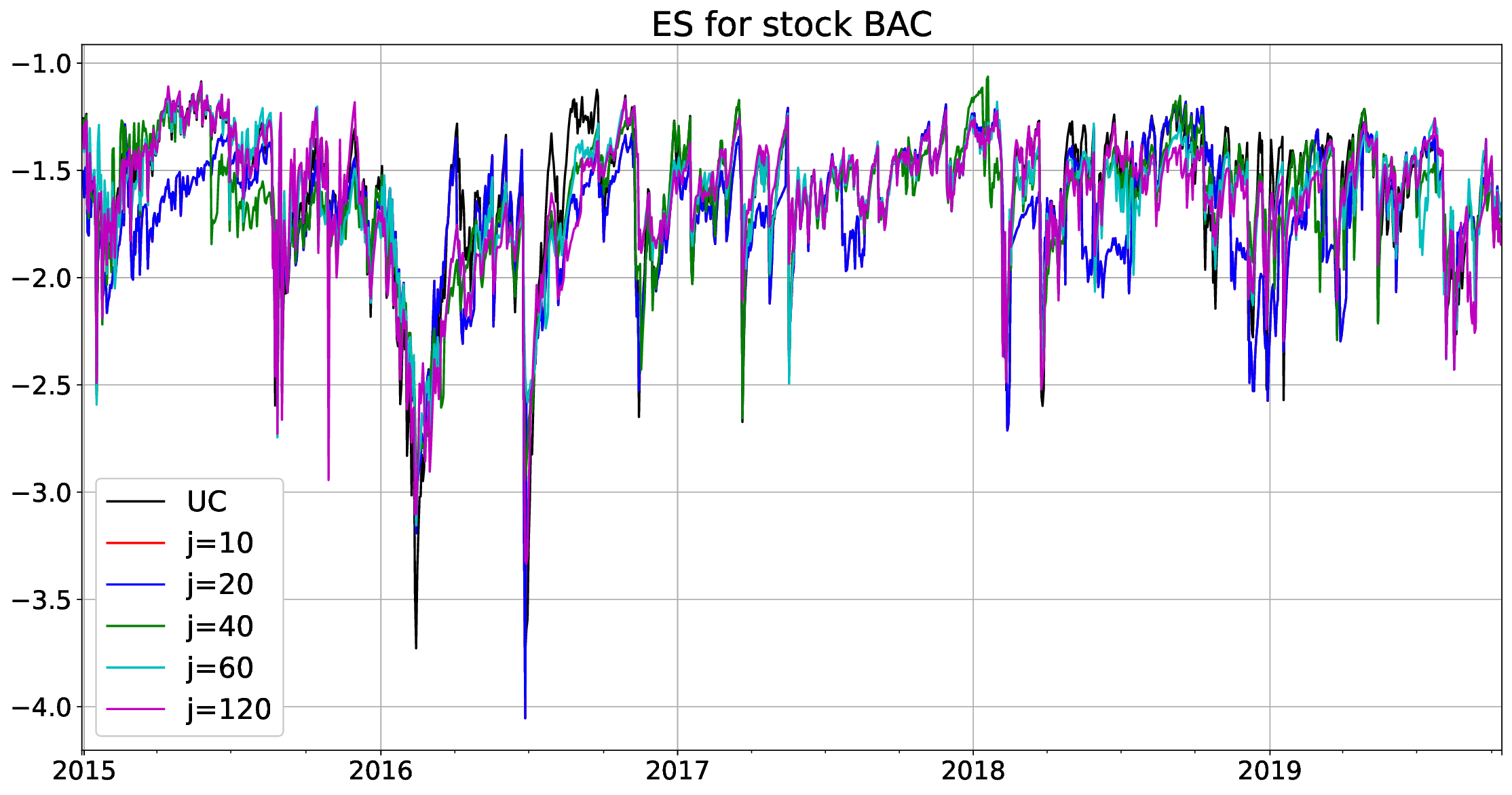}\vspace{-5pt}
	\caption*{\scriptsize 
		Coloured lines correspond to 10-day ahead ES forecasts for different liquidity horizons.}
	\label{fig:MCS_BAC}
\end{figure}

\begin{figure}[h!]
	\caption{Time series of 10-day ahead ES forecasts for GE: CMCS.}
	\centering
	\includegraphics[trim={0cm 0cm 0cm 0cm},clip,width = 0.8\textwidth]{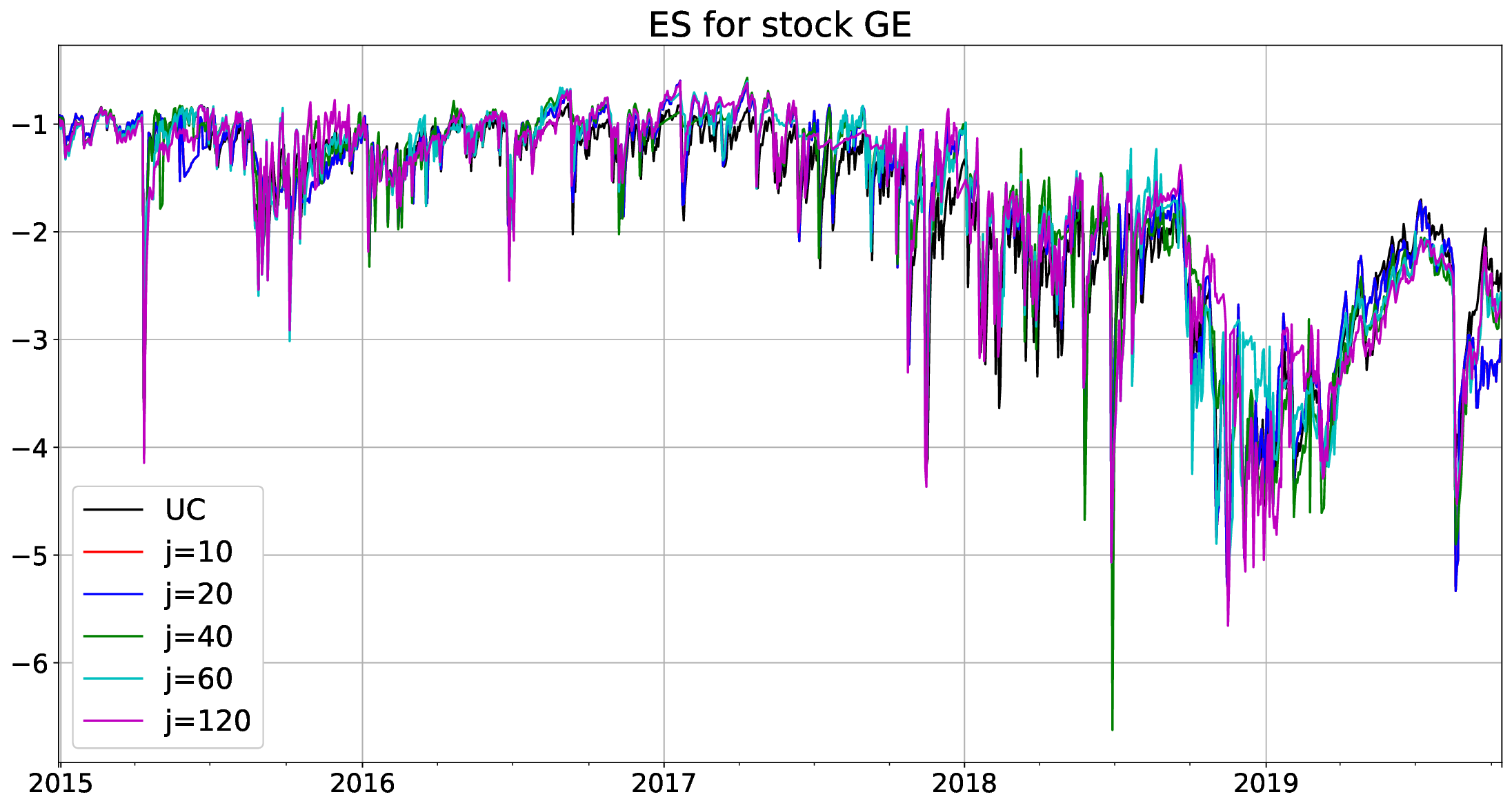}\vspace{-5pt}
	\caption*{\scriptsize 
		Coloured lines correspond to 10-day ahead ES forecasts for different liquidity horizons.}
	\label{fig:MCS_GE}
\end{figure}

The time series of the CMCS ES forecast combination for the \texttt{BAC} and \texttt{GE} stocks are presented in Figures \ref{fig:MCS_BAC} and \ref{fig:MCS_GE}. Lines of different colours correspond to liquidity horizons. 

The plots illustrate that the CMCS dynamically adopts the composition of the forecast: a large expected loss for the \texttt{BAC} stock at the beginning of the year is associated with the $LH = 20$ depicted as a blue line, whereas at the end of the year all of the risk factors, including the unconditional forecast, are of a similar level. The forecast of expected loss for the \texttt{GE} on the contrary does not highlight a specific risk factor to which the stock could be particularly sensitive to.

\begin{figure}[h!]
	\caption{Model selection for 1 day ahead ES forecasts for different liquidity horizons: CMCS.}
	\centering
	\includegraphics[trim={0cm 0cm 0cm 0cm},clip,width = 0.9\textwidth]{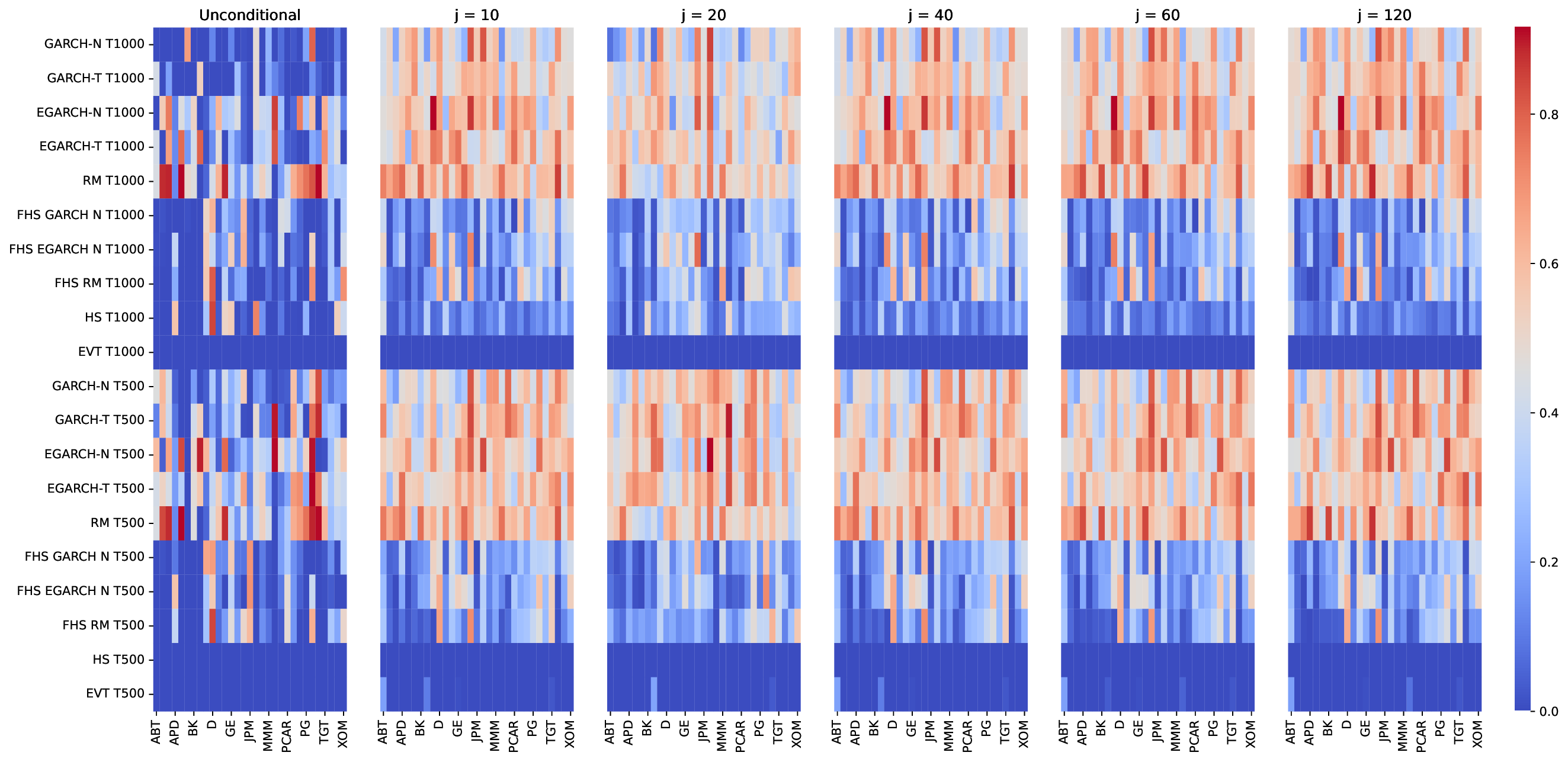}
	\caption*{\scriptsize Heatmaps correspond to the liquidity horizons as specified in Table \ref{tab:LH}. On each heatmap the x-axis corresponds to stocks and the y-axis to models. The heatmap cells correspond to the average number of the out-of-sample periods where the model is included in the MCS. The warmer the colour, the more frequently was the model selected for the forecast combination.}
	\label{fig:MCS_stocks}
\end{figure}

\begin{figure}[h!]
	\caption{Model selection for 10 days ahead ES forecasts for different liquidity horizons: CMCS.}
	\centering
	\includegraphics[trim={0cm 0cm 0cm 0cm},clip,width = 0.9\textwidth]{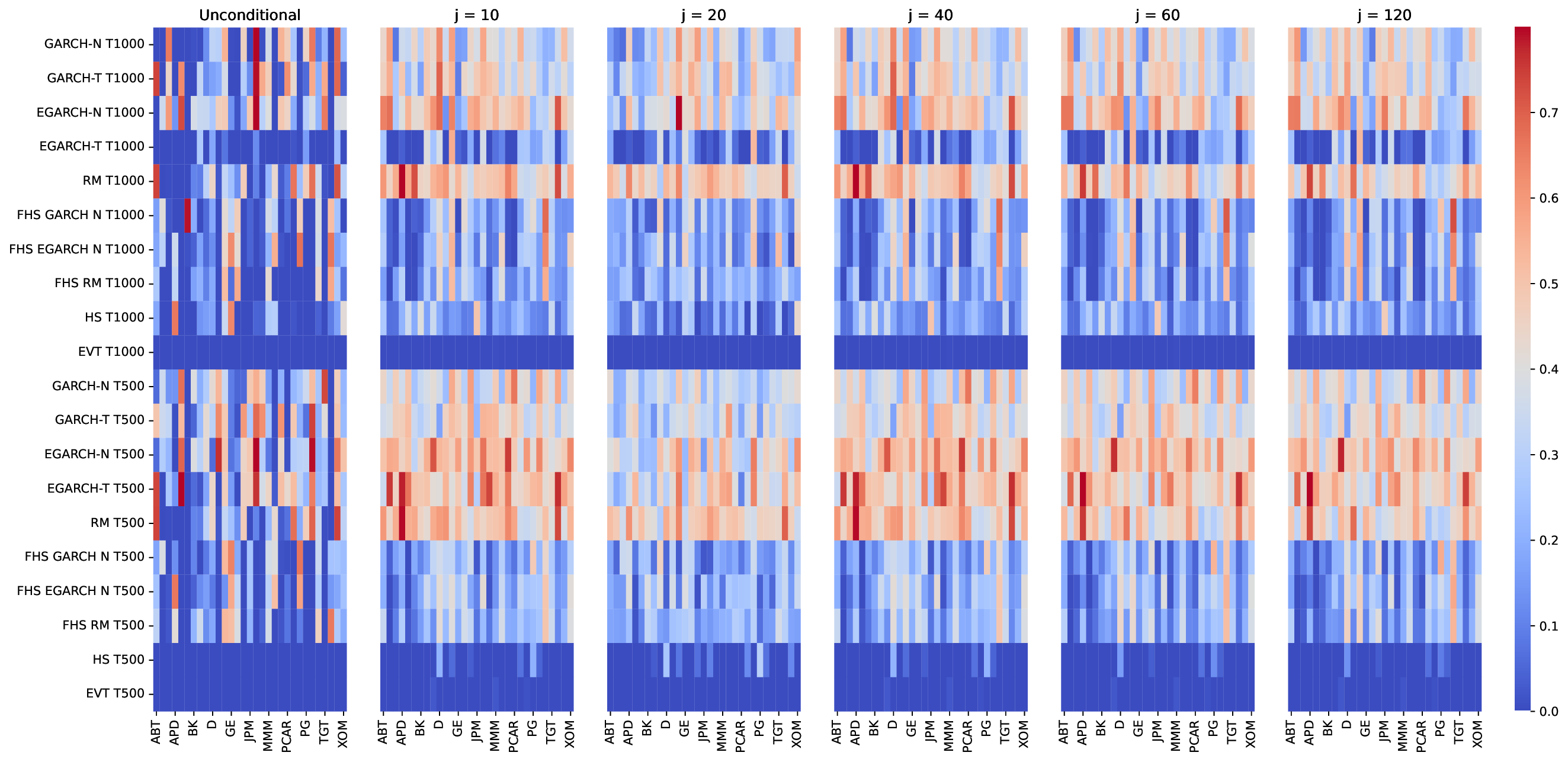}
	\caption*{\scriptsize Heatmaps correspond to the liquidity horizons as specified in Table \ref{tab:LH}. On each heatmap the x-axis corresponds to stocks and the y-axis to models. The heatmap cells correspond to the average number of the out-of-sample periods where the model is included in the MCS. The warmer the colour, the more frequently was the model selected for the forecast combination.}
	\label{fig:MCS_stocks10}
\end{figure}

Figures \ref{fig:MCS_stocks} and \ref{fig:MCS_stocks10} provide insights into the forecast composition of the CMCS. The heatmaps depict the out-of-sample average frequency of a given model (in y-axis) to be selected for a forecast combination for a given stock (in x-axis) for 1 and 10 days ahead forecasts respectively. The heatmaps illustrate that the model selection considerably differs across \textit{(i)} forecasting horizons (Figures \ref{fig:MCS_stocks} and \ref{fig:MCS_stocks10}); \textit{(ii)} risk factors (panels on each figure); and \textit{(iii)} across stocks (patterns of each heatmap). The warmer the colour, the more frequently the method was included in the CMSC.
Notably, the conservative Extreme Value Theory (EVT) models are excluded from the CMCS for all stocks and liquidity horizons and ES models which are based on historical simulation are in general selected fewer times. Furthermore, the forecasting horizon plays a role, e.g.~an EGARCH with Student-t innovations estimated on a large window of 1000 observations is rarely picked up for 10 day forecasting horizon, compared to the 1 day forecasting horizon. The selection of the same model, but estimated on a different estimation window, differs as well. Figure \ref{fig:MCS_stocks10} demonstrates that for the 10 days forecasting horizon the EGARCH-t model estimated on a shorter window of 500 observations is selected in the CMCS rather frequently, compared to the same model estimated on 1000 observations. Furthermore, the conditional model selection drastically differs from the unconditional one, the latter being more sparse. The conditional method selection varies largely across stocks, as well as liquidity horizons: comparison of the second and the last heatmaps on Figure \ref{fig:MCS_stocks} demonstrate that for different risk factors a different set of models is selected for the same stock. 
These results indicate that the model performance in downside risk measurement is very heterogeneous and call for data-driven methods in the risk management.

Next, we consider an alternative MCS-based forecast combination - the DFC by \textcite{Borup.2024}. The test statistics of DFC crucially depends on the covariance estimator of the loss differentials, see \textcite{Borup.2017} for the details. Given the daily frequency of the ES forecasts, we implement the HAC estimator with a truncated kernel with a truncation lag of a quarter of the sample and report the results based on sample covariance estimator and larger truncation lag in Appendix \ref{app:ea}. The results for the sample covariance estimator, as expected, differ substantially from the HAC estimator, however the performance of the testing procedure for a different truncation lag is comparable to the results reported below.

\begin{figure}[h!]
	\caption{Method selection for 10 day ahead ES forecasts for different liquidity horizons:  DFC with truncation $T/4$.}
	\centering
	\includegraphics[trim={0cm 0cm 0cm 0cm},clip,width = 0.9\textwidth]{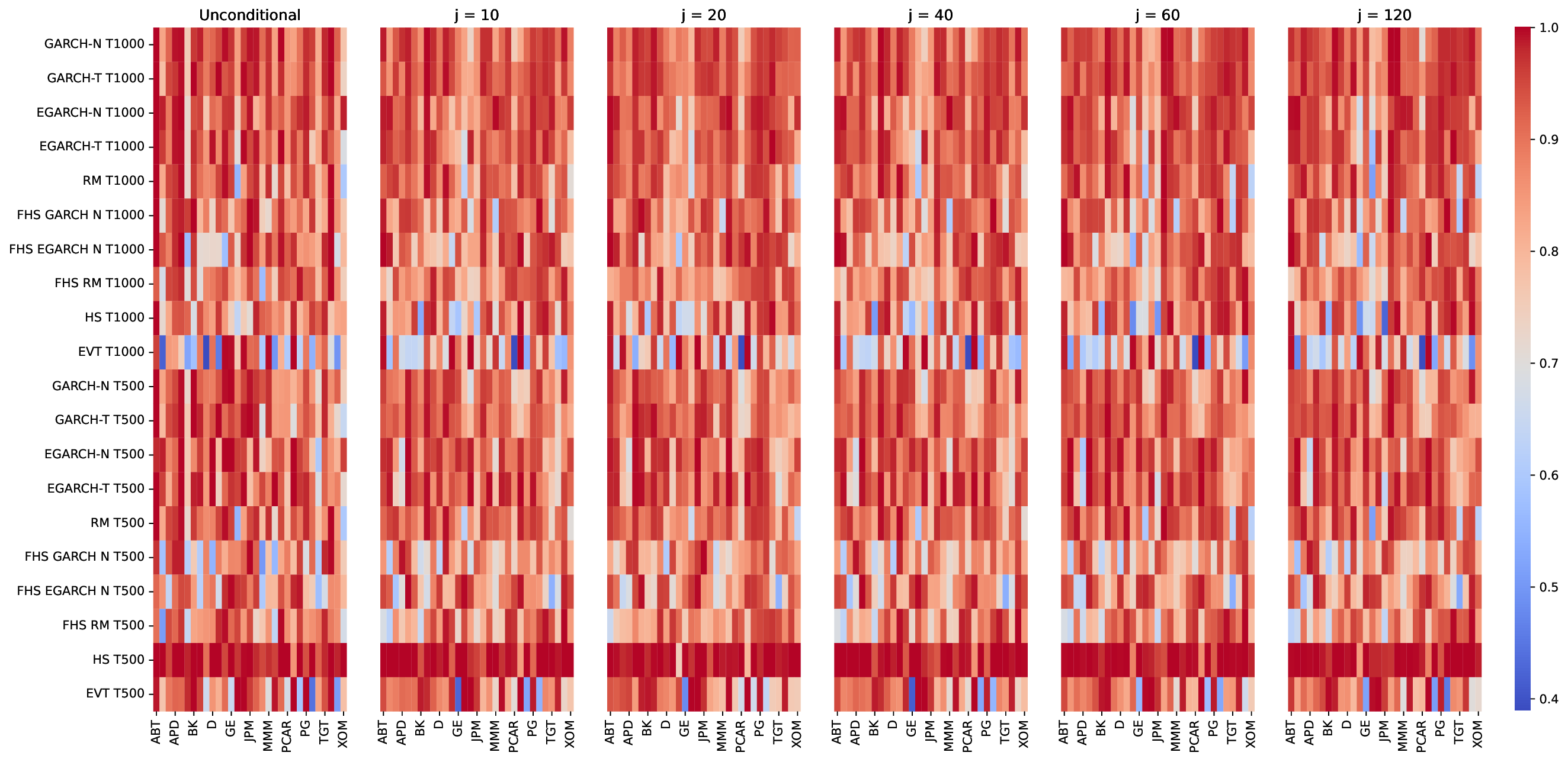}
	\caption*{\scriptsize Heatmaps correspond to the liquidity horizons as specified in Table \ref{tab:LH}. On each heatmap the x-axis corresponds to stocks and the y-axis to models. The heatmap cells correspond to the average number of the out-of-sample periods where the model is included in the MCS. The warmer the colour, the more frequently was the model selected for the forecast combination.}
	\label{fig:BT1_stocks}
\end{figure}

Appendix Tables \ref{tab:bt_sample1}  - \ref{tab:bt_tr2_10} report the average ES forecasts of the DFC method for each stock. The reported forecasts for the expected loss are more negative compared to the results for the proposed CMCS in Tables \ref{tab:cmcs1} and \ref{tab:cmcs10}. This difference can be explained by the model selection patterns in Figure \ref{fig:BT1_stocks}: the more conservative ES forecasts of EVT and RiskMetrics are often included in the DFC forecast. Overall, the testing procedure seem to lack in power to discriminate the performance of candidate models during the stress periods of different risk factors. The regulatory requirements of the ES stress testing imply that the considered sample in which the forecasts are stress-tested is considerably large, however the ``conditioning'' period of stress is relatively short, making it difficult for the test of DFC, which is based on the whole sample, to differentiate between the models. 

\begin{figure}[h!]
	\caption{Time series of 10-day ahead ES forecasts for BAC:  DFC with truncation $T/4$.}
	\centering
	\includegraphics[trim={0cm 0cm 0cm 0cm},clip,width = 0.8\textwidth]{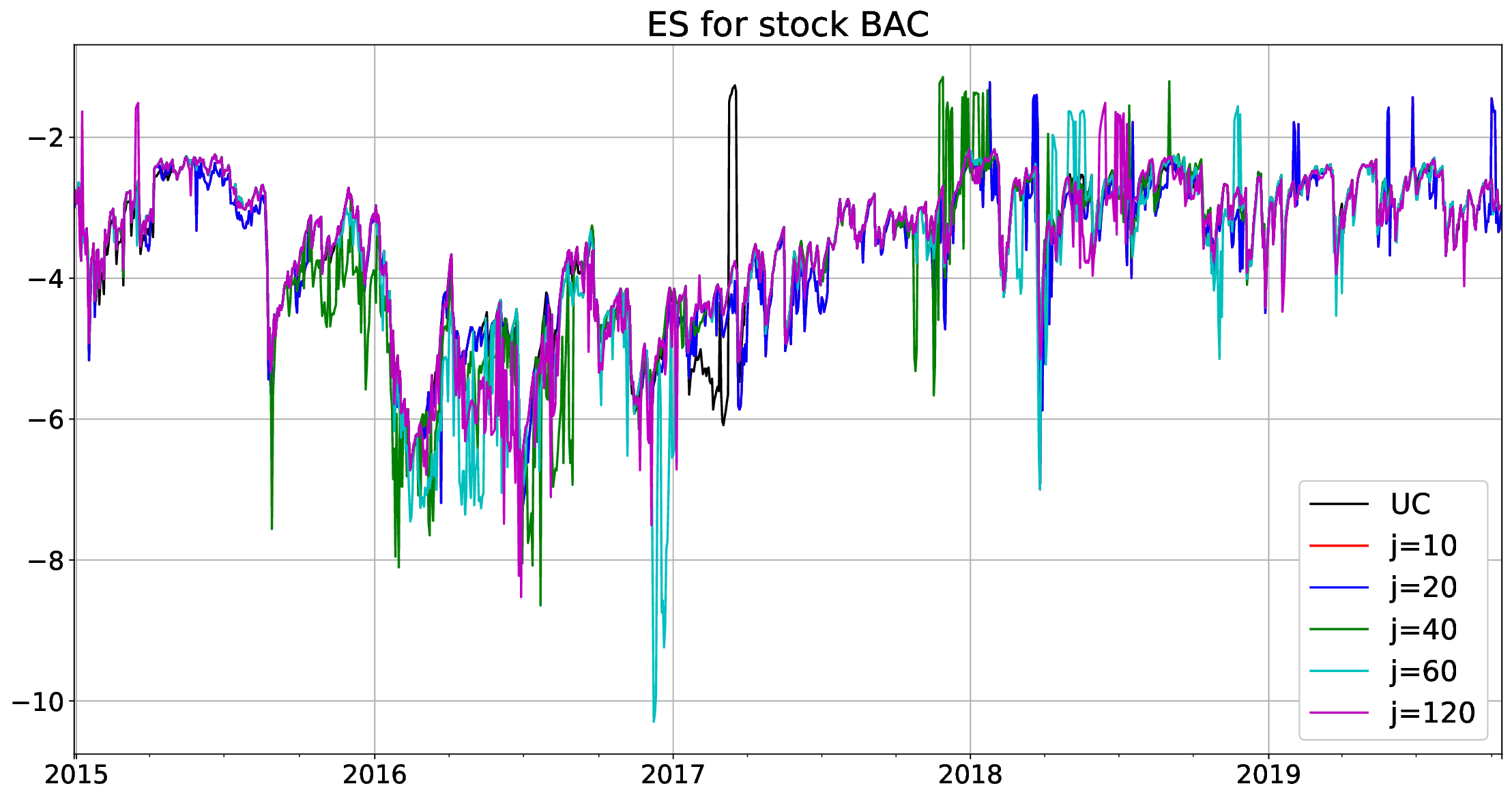}\vspace{-5pt}
	\caption*{\scriptsize 
		Coloured lines correspond to 10-day ahead ES forecasts for different liquidity horizons.}
	\label{fig:BT1_BAC}
\end{figure}
As a consequence, the resulting forecast combination includes under-performing models, which results in over-estimation of downside risk. Providing overly conservative foreacsts implies that the financial institutions would have to increase their reserves drastically.  Furthermore, as depicted in Figure \ref{fig:BT1_BAC}, the forecast combination is unstable over time, as the test does not exclude volatile forecasts. The performace of the Wald test-bsed DFC for different forecasting horizons and truncation lags are presented in Appedix \ref{app:ea}. These results imply that (i) the testing procedure is rather sensitive to the truncation lag; (ii) for the HAC estimator the test lacks power to differentiate between the methods, risk factors and forecasting horizon. For the down-side risk stress testing and forecasting we therefore suggest to rely on the proposed state-wise CMCS, which offers a robust tool for data driven conditional method selection.


%
\section{Conclusions}
\label{sec:conclusions}
This paper proposes a new methodology for evaluating the forecasting performance of econometric models under varying financial conditions, introducing the Conditional Method Confidence Set (CMCS). The CMCS extends the traditional Model Confidence Set (MCS) of \textcite{Hansen.2011} by incorporating conditional forecast evaluation, allowing for model comparison based on specific economic regimes. This approach is particularly relevant in the context of stress-testing financial downside risk measures, such as Value-at-Risk (VaR) and Expected Shortfall (ES), which are critical for regulatory compliance under the Basel Accords.

Our theoretical contributions include the development of a state-dependent testing procedure that is asymptotically valid. By allowing forecast accuracy to be evaluated conditionally on specific market states, the CMCS approach offers a more robust and adaptable method for selecting models in volatile financial environments.

Empirically, we demonstrate the efficacy of the CMCS in stress-testing ES forecasts across different liquidity horizons implied by a variety of risk factors and economic conditions. The results show that different assets react differently to stress scenarios for different risk factors, and that the proposed CMCS method provides a consistent framework for identifying the best-performing models under these conditions. Importantly, the conditional testing approach captures the heterogeneous nature of risk exposure across assets and forecasting horizons, offering more insights into forecast performance compared to traditional unconditional methods.

Future work may examine different means of accounting for false discoveries, such as controlling the false discovery rate or the mixed directional FWER. This may improve the power of the CMCS and its ability to lead to more accurate forecasts, but is beyond the scope of this paper. Moreover, exploring other test statistics to construct the CMCS could result in improved finite sample properties. Finally, investigating additional empirical applications may help assess the usefulness of our proposed method. 
\newpage

\addcontentsline{toc}{section}{\hspace{0.6cm}References}
\newpage
\printbibliography

\newpage
\begin{appendix}
\section{Appendix}
\label{sec:appendix}

\subsection{Assumptions on the statewise testing procedure}\label{sec:assumptions_on_testing_procedure}

The CMCS procedure is based on an equivalence test \( \delta_{\mathcal{M}} \) and an elimination rule \( \epsilon_{\mathcal{M}} \). An equivalence test,  \( \delta_{\mathcal{M}} \) is used to test for the null of CEPA \( H^{l}_{0, \mathcal{M}} \) for any \( \mathcal{M}^{l} \subset \mathcal{M}^{\cdot, 0} \). The elimination rule \( \epsilon_{\mathcal{M}} \) identifies the \( e \in \mathcal{M}^{l} \) that we eliminate if we reject \( H^{l}_{0} \). For \( \delta_{\mathcal{M}} \) and \( \epsilon_{\mathcal{M}} \), we suppress \( (\cdot) ^{l} \) for ease of notation, though the pairs of equivalence test and elimination rule could differ per state. 

Regarding notation, \( \delta_{\mathcal{M}} = 0 \) refers to the case when \( H^{l}_{0, \mathcal{M}} \) cannot be rejected, and \( \delta_{\mathcal{M}} = 1 \) refers to the case when it is rejected. Note that the former also implies that the testing procedure is halted. In the latter case, the notation \( \epsilon_{\mathcal{M}} = j \) indicates that method \( j \) is eliminated from the current set of candidate methods \( \mathcal{M}^{l} \). 

\subsubsection*{Asymptotic size and power}\label{subsec:mcs_asymptotics_conditional}
\begin{assumption}\label{ass:ass1hln_conditional}
	For any \( \mathcal{M}^{l} \subset \mathcal{M}^{\cdot, 0}\), \( l =1, \ldots, d \), we assume about \( (\delta_{\mathcal{M}}, \epsilon_{\mathcal{M}}) \), that 
	\begin{enumerate}[label=(\alph*)]
		\item \( \displaystyle \limsup_{n^{l} \to \infty} P(\delta_{\mathcal{M}}=1 | H^{l}_{0, \mathcal{M}}) \leq \alpha  \),  
		\item \( \displaystyle \lim_{n^{l} \to \infty} P(\delta_{\mathcal{M}}=1 | H^{l}_{A, \mathcal{M}}) = 1  \) and 
		\item\( \displaystyle \lim_{n^{l} \to \infty} P(\epsilon_{\mathcal{M}} \in \mathcal{M}^{l, *} | H^{l}_{A, \mathcal{M}}) = 0  \),
	\end{enumerate}
	where the case \( \delta_{\mathcal{M}}=1 \) indicates elimination.
\end{assumption}

These assumptions are rather standard for hypothesis tests and may be verbalized as follows. (a) requires the asymptotic level not to exceed \( \alpha \), (b) requires the asymptotic power to be 1, and (c) requires that a superior object \( i^{*} \in \mathcal{M}^{l, *} \) is not eliminated as \( n^{l} \to \infty \) as long as there are inferior methods in \( \mathcal{M}^{l} \).

\begin{theorem}[Properties of the conditional MCS \( \widehat{\mathcal{M}}^{l,*}_{1-\alpha} \)]\label{th:th1hln_conditional}
	Given Assumption 1, for each \( l =1, \ldots, d \) it holds that 
	\begin{enumerate}
		\item[(i)] \( \displaystyle \liminf_{n^{l} \to \infty} P(\mathcal{M}^{l,*} \subset \widehat{\mathcal{M}}^{l,*}_{1-\alpha} )  \geq 1 - \alpha \) and 
		\item[(ii)] \(  \displaystyle \liminf_{n^{l} \to \infty} P(i \in \widehat{\mathcal{M}}^{l,*}_{1-\alpha} ) = 0  \text{ for all } i \notin \mathcal{M}^{l,*}   \).
	\end{enumerate} 
\end{theorem}
\begin{proof}
	Let \( i^{l, *} \in \mathcal{M}^{l, *}. \) (i): Consider the event that \( i^{l, *} \) is eliminated from \( \mathcal{M} \). From Assumption 4.1 c it follows that \( P(\delta_{\mathcal{M}}=1, \epsilon_{\mathcal{M}}=i^{l, *}| H^{l}_{A, \mathcal{M}}) \leq P(\epsilon_{\mathcal{M}}= i^{l, *} | H^{l}_{A, \mathcal{M}}) \to 0 \) as \( n^{l} \to \infty \). Next, from Assumption 1(a) it follows that \( \displaystyle \limsup_{n^{l} \to \infty} P(\delta_{\mathcal{M}}=1, \epsilon_{\mathcal{M}}=i^{l, *}| H^{l}_{0, \mathcal{M}}) = \displaystyle \limsup_{n^{l} \to \infty} P(\delta_{\mathcal{M}}=1 | H^{l}_{0, \mathcal{M}}) \leq \alpha  \) such that the probability that \( i^{l, *} \) is eliminiated when \( \mathcal{M}^{l} = \mathcal{M}^{l, *} \) is asymptotically bounded by \( \alpha \). To prove (ii), we first note that \( \displaystyle \limsup_{n^{l} \to \infty} P(\epsilon_{\mathcal{M}}=1 | H^{l}_{A, \mathcal{M}}) = 0  \), such that only methods \( i \notin \mathcal{M}^{l, *} \) will be eliminated (asymptotically) (under the alternative, it holds that \( \mathcal{M}^{l} \neq \mathcal{M}^{l, *} \) ). On the other hand, Assumption 1(b) ensures that inferior methods will be eliminated as long as the null hypothesis is false. 
\end{proof}

\begin{corollary}\label{cor:cor1hln_conditional}
	Suppose that Assumption 1 holds and that \( \mathcal{M}^{l, *} \) is a singleton, \( l  \in \{1, \ldots, d \} \). \\
	Then \( \lim_{n^{l} \to \infty} P(\mathcal{M}^{l,*} = \widehat{\mathcal{M}}^{l,*}_{1-\alpha} )  =1 \). 
\end{corollary}

\begin{proof}
	When \( \mathcal{M}^{l, *} = \{i^{l, *}\} \) it follows from follows from Theorem \ref{th:th1hln_conditional} that \( i^{l, *} \) will be the last surviving element with probability approaching 1 as \( n^{l} \to \infty \).
\end{proof}

\subsection*{Coherency}
\begin{theorem}\label{th:th2hln_conditional}
	Suppose that \( P(\delta=1, e\in \mathcal{M}^{l, *}) \leq \alpha\). Then we have \\
	\( P(\mathcal{M}^{l, *} \subset \widehat{\mathcal{M}}^{l,*}_{1-\alpha}) \geq 1- \alpha \).
\end{theorem}

\begin{proof}
	Consider the first case when \( \epsilon_{\mathcal{M}} \in \mathcal{M}^{l,*} \). The claim  for this first instance follows from Assumption \ref{ass:ass1hln_conditional}. Further tests only take place if all previous ones reject, so that \( P( \epsilon \in \mathcal{M}^{l,*} )\) is bound from above by \( \alpha \).
\end{proof}


If a test has correct size, it holds that \( P(\delta=1 | H_{0, \mathcal{M}}) = \alpha\) which implies \( P(\delta=1, \epsilon \in \mathcal{M}^{*}| H_{0, \mathcal{M}}) \leq \alpha\). This is the usual notion of size in statistical testing. \\
\textcite{Hansen.2011} impose the further condition that \( P(\delta=1, \epsilon \in \mathcal{M}^{*}| H_{A, \mathcal{M}}) \leq \alpha\), i.e. we are sufficiently sure that the method we eliminate is not from \( \mathcal{M}^{*} \).


\begin{corollary}\label{cor:product_conditional_GW}
	Define \( z_{ij, t} = h_{t} d_{ij,t} \). Under Assumption \ref{ass:GW_mixing_HLN_adj}, for some \( r >2 \) it holds for all \( n, t \) that \( \{ z_{ij, t}\}_{i,j \in \mathcal{M}^{0}} \) is \( \alpha \)-mixing of order \( -r/(r-2) \).  
\end{corollary}

\begin{proof}\label{proof:product_conditional_GW}
	Write \( z_{ij, t} = f(h_{t}, W_{t+1}, \ldots, W_{t-r} ) \) for some measurable function \( f(\cdot) \). As  \( \{ d_{ij,t}\}_{i,j \in \mathcal{M}^{0}} \) and \( \{ h_{t}\} \) are mixing of the same size under Assumption \ref{ass:GW_mixing_HLN_adj}, and \( f \) is a function of a finite number of leads and lags of \( W_{t+1} \), it follows from Lemma 2.1 in \textcite{White.1984} that \( \{z_{t, ij} \} \) is mixing of the same size as \( \{ d_{ij,t}\}_{i,j \in \mathcal{M}^{0}} \) and \( \{ h_{t}\} \).
\end{proof}

\subsubsection{Quadratic-Form Test}
Let \( m \) denote the number of methods in \( \mathcal{M} \), \( L^{l}_{\tau} \in \mathbb{R}^{1\times m} \) the vector of loss variables, \( \bar{L}^{l} \equiv (n^{l})^{-1} \sum^{n^{l}}_{\tau=1} L^{l}_{\tau} \) its sample average. Let \( \iota \equiv (1, \ldots, 1)^{'} \) be the column vector where all \( m \) entries equal 1. The orthogonal complement \( \iota_{\perp} \in \mathbb{R}^{m \times m-1} \) has full column rank and satisfies \( \iota_{\perp}^{'} \iota = \mathbf{0} \). The \( m-1 \) dimensional vector \( X^{l}_{\tau} \equiv \iota^{'}_{\perp} L^{l}_{\tau} \) can be viewed as \( m-1 \) contrasts. Each element is a linear combination of the loss differentials \( d^{l}_{ij, \tau} \) that have mean zero under the null. 

\begin{lemma}\label{lm:lm1hln_conditional} 
	Given Assumptions \ref{ass:GW_mixing_HLN_adj}, \ref{ass:GW_finite_moments} and \ref{ass:subsampled_loss_differentials}, let \( X^{l}_{\tau} \equiv \iota^{'}_{\perp} L^{l}_{\tau} \), \( \tau \in \{1, \ldots, n^{l} \} \), and define \( \theta^{l}_{n^{l}} \equiv E(X^{l}_{n^{l} \tau}) \). The null hypothesis \( H^{l}_{0, \mathcal{M}} \) is equivalent to \( \theta^{l}_{n^{l}} =0 \) and it holds that \( (n^{l})^{1/2} ( \bar{X^{l}_{n^{l}}} -\theta^{l}_{n^{l}} ) \overset{d}{\to} N(0, \Sigma^{l}_{n^{l}} ) \), where \( \bar{X^{l}_{n^{l}} } \equiv (n^{l})^{-1} \sum^{n^{l}}_{\tau=1} X^{l}_{\tau} \) and \( \Sigma^{l}_{n^{l}}  \equiv var((n^{l})^{1/2} \bar{X^{l}_{n^{l}}}).\)
\end{lemma}

\begin{proof}
	We can write \( X^{l}_{t} \equiv \iota^{'}_{\perp} L^{l}_{t} \) as a linear combination of the \( d^{l}_{ij,t} \) as \( \iota^{'}_{\perp} \iota  = 0 \).
	Thus we see that \( H^{l}_{0, \mathcal{M}}: \theta^{l} = 0 \). 
	The asymptotic normality follows by the CLT for \( \alpha \)-mixing processes. (see, e.g., \textcite{White.2001}, Th.\ 5.20.)
\end{proof}

\begin{lemma}\label{lm:lm2hln_conditional}
	Suppose that Assumptions \ref{ass:GW_mixing_HLN_adj}, \ref{ass:GW_finite_moments} and \ref{ass:subsampled_loss_differentials} hold and define \( \bar{V}^{l}_{n^{l}} = (\bar{d}^{l}_{i\cdot, n^{l}}, \ldots, \bar{d}^{l}_{m\cdot, n^{l}})^{'}. \), \( l =1, \ldots, d \).\\
	Then for each \( l \in {1, \ldots, d} \)
	\begin{equation}
		(n^{l})^{1/2}(\bar{V}^{l}_{n^{l}} - \psi^{l}_{n^{l}}) \overset{d}{\to} N_{m}(0, \Omega^{l}_{n^{l}}) \text{ as } n^{l} \to \infty,
	\end{equation}
	where \( \psi^{l}_{n^{l}} \equiv E( \bar{V}^{l}_{n^{l}}) \text { and } \Omega^{l}_{n^{l}} \equiv var((n^{l})^{1/2} \bar{V}^{l}_{n^{l}}) \) and the null hypothesis \( H^{l}_{0, \mathcal{M}} \) is equivalent to \( \psi^{l}_{n^{l}}=0.\) 
\end{lemma}

\begin{proof}
	From the identity \( \bar{d}^{l}_{i\cdot, n^{l}} = \bar{L}^{l}_{i, n^{l}} - \bar{L}^{l}_{\cdot, n^{l}}  = \bar{L}^{l}_{i, n^{l}} - m^{-1} \displaystyle \sum_{j \in \mathcal{M}} \bar{L}^{l}_{j, n^{l}} = m^{-1} \displaystyle \sum_{j \in \mathcal{M}} \bar{d}^{l}_{ij, n^{l}} \), we see that the elements of \( \bar{V}^{l}_{n^{l}} \) are linear transformations of \( \bar{X}^{l}, n^{l} \) from Lemma \ref{lm:lm1hln_conditional}. Thus for some \( (m-1) \times m \) matrix \( G^{l}_{n^{l}} \), we have \( \bar{V}^{l}_{n^{l}} = G^{l \prime}_{n^{l}} \bar{X}^{l}_{n^{l}} \) and the result now follows, where \( \psi^{l}_{n^{l}} = G^{l\prime}_{n^{l}} \theta^{l}_{n^{l}} \) and \( \Omega^{l}_{n^{l}} = G^{l \prime}_{n^{l}} \Sigma^{l}_{n^{l}} G^{l}_{n^{l}}\). 
\end{proof}

\clearpage

\subsection{CMCS Bootstrap procedure}\label{app:bootstrap}
This section outlines the CMCS bootstrap, which is operationally the same as the MCS bootstrap performed on the statewise losses. 

The block bootstrap \textcite{Hansen.2011} suggest using is the circular block bootstrap introduced by \textcite{Politis.1991}. In practice, the researcher needs to choose an appropriate block length \( p \), which is tied to persistence in the loss differentials \( d_{i\cdot, t} \) and \( d_{ij} \). \textcite{Hansen.2011} acknowledge the difficulty of choosing the appropriate block length \( p \) and recommend trying different specifications for robustness.  

\subsubsection*{Step 0:  Select the statewise losses}
For each state \( l \in \{1, \ldots, d \} \), split the unconditional loss matrix \( \mathbf{L} \in \mathbb{R}^{m \times n} \) into \( d \)  matrices \( \mathbf{L}^{l} \in \mathbb{R}^{ m \times n^{l}}, \; l \in \{1, \ldots, d \}\), where \( \mathbf{L}^{l} = \mathbf{L} | t \in I^{l}  \), where \( t \) denotes the column index of \( \mathbf{L} \), thus selecting the observations for which we observe state \( l \) at the forecast origin (cf. Section \ref{sec:description_environment}).
Next, perform the bootstrap procedure from \textcite{Hansen.2011} on each statewise loss matrix \( \mathbf{L}^{l} \). For ease of notation, in the following we suppress the superscript \( l \) that indicates the state, except for \( n^{l} \) and \( \mathcal{M}^{l}\). For instance, we write \( L_{i, \tau} \) instead of \( L_{i, \tau}^{l} \).

\paragraph*{Step \(l.1\): Bootstrap indices for resampling}
\begin{enumerate}
    \item[(a)] Choose the block-length parameter \( p \). 
    \item[(b)] Generate B bootstrap resamples of the indices \( \{1, \ldots, n^{l}\} \). I.e., for \( b = 1, \ldots, B\): 
        \begin{enumerate}[label=\roman*]
            \item Choose \( \xi_{b_1} \sim U\{1, \ldots, n^{l} \} \) and set \( (v_{b, 1}, \ldots, v_{b, p} ) = ( \xi_{b_1}, \xi_{b_1}+1, \xi_{b_1}+p-1 ) \), with the convention that \( n^{l} + i = i \) for \( i \geq 1 \).
            \item Choose \( \xi_{b_2} \sim U\{1, \ldots, n^{l} \} \) and set \( (v_{b, p+1}, \ldots, v_{b, 2p} ) = ( \xi_{b_2}, \xi_{b_2}+1, \xi_{b_2}+p-1 ) \)
            \item Continue until a sample size of \( n^{l} \) is constructed.
            \item This is repeated for all resamples \( b=1, \ldots, B \) using independent draws of the \( \xi \)'s.
        \end{enumerate} 
    \item[(c)] Save the full matrix of bootstrap indices. 
\end{enumerate}

\paragraph*{Step \(l.2\): Sample and Bootstrap Statistics}
\begin{enumerate}
    \item[(a)] For each method and each point of time we evaluate the performance to obtain the variables \( L_{i,\tau} \) for \( i=1, \ldots, m \) and \( \tau = 1, \ldots, n^{l} \). These variables are used to calculate the sample averages for each method \( \bar{L}_{i, \cdot} \equiv \dfrac{1}{n^{l}}  \displaystyle \sum_{\tau=1}^{n^{l}}  L_{i, \tau}, i = 1,\ldots, m \). 
    \item[(b)] The corresponding bootstrap variables are now given by \(  L_{b,i,\tau}^{*} = L_{i, v_{b,\tau}} \) for \( b=1, \ldots, B, \; i=1, \ldots, m \text{ and } \tau=1, \ldots, n^{l}, \) and calculate the bootstrap sample averages \(  \bar{L}_{b,i}^{*} \equiv \displaystyle \sum_{\tau=1}^{n^{l}}  L_{b,i,\tau}^{*} \). It suffices to store the variables \( \bar{L}_{i} \) and \( \xi_{b,i}^{*} \equiv \bar{L}_{b,i}^{*} - \bar{L}_{i} \),  as all subsequent statistics can be calculated from these. 
\end{enumerate}

\paragraph*{Step \(l.3\): Sequential Testing}
\begin{enumerate}
    \item[(a)] Let \( m \) denote the number of elements in \( \mathcal{M}^{l} \) and calculate 
    \begin{align*} \bar{L}_{\cdot} \equiv \dfrac{1}{m} \displaystyle \sum_{i=1}^{m}  \bar{L}_{i}, \quad  \xi^{*}_{b,\cdot} \equiv \dfrac{1}{m} \displaystyle \sum_{i=1}^{m}  \xi^{*}_{b,i} \quad \text{and   } \widehat{var}(\bar{d}_{i \cdot}) \equiv \dfrac{1}{B} \displaystyle \sum_{b=1}^{B}(\xi^{*}_{b,i} - \xi^{*}_{b,\cdot})^{2}.
    \end{align*} Next, define \( t_{i \cdot} \equiv \dfrac{ \bar{d}_{i\cdot}}{\widehat{var}(\bar{d}_{i \cdot})} \) and calculate the test statistic \( T_{\max} = \max_{i} t_{i \cdot} \) . 
    \item[(b)] The bootstrap estimate of \( T_{D}\)'s distribution is given by the empirical distribution of \( T_{b,\max}^{*} = \displaystyle \max_{i} t_{b,i}^{*}, \; \text{ for } b=1, \ldots, B  \), where \( t_{b,i}^{*} = \dfrac{ \xi_{b,i}^{*} - \xi_{b,\cdot}^{*}}{\sqrt{\widehat{var}(\bar{d}_{i \cdot})}} \)
    \item[(c)] The p-value of \( H_{0, \mathcal{M}^{l}} \) is given by \( P_{H_{0, \mathcal{M}^{l}}} \equiv \dfrac{1}{B} \displaystyle \sum_{b=1}^{B} \mathbbm{I}_{\{T_{max} < T_{b,\max}^{*} \} } \), where \( \mathbbm{I}_{\{ \cdot \} } \) is the indicator function. 
    \item[(d)] If \( P_{H_{0, \mathcal{M}^{l}}}  < \alpha \), where \( \alpha \) is the level of the test, \( H_{0, \mathcal{M}^{l}} \) is rejected and \( \epsilon_{\mathcal{M}^{l}} \equiv \argmax_{i} t_{i \cdot} \) is eliminated from \( \mathcal{M}^{l} \).
    \item[(e)] The steps in (a)-(d) are repeated until the \( H_{0, \mathcal{M}^{l}} \) cannot be rejected for the first time. The resulting set of methods \( \widehat{\mathcal{M}^{l}}^{*}_{1-\alpha} \) is the \( (1-\alpha ) \) CMCS for state \( l \).
\end{enumerate}

Thus, for each state \( l \), one obtains a CMCS  \( \widehat{\mathcal{M}^{l}}^{*}_{1-\alpha} \).


\subsection{Two methods, two states: Wald type test statistic}\label{app:wald_type_two_states_simulation}

To theoretically illustrate the mechanisms that drive the results in Sections \ref{sec:two_methods_Wald_vs_t} and \ref{sec:simulation_wald_vs_t_rejection_regions}, we calculate the expression of the Wald type test statistic in a simple setting with two forecasting methods and two disjoint states.

We assume that the loss differentials follow a normal distribution conditional on a state variable \( S_{t} \), where \(P(S_{t} = 1)=p \), and \(P(S_{t} = 2)=1-p \).
\begin{equation}\label{eq:Wald_vs_t_loss_differential}
    D_{t} = 
    \begin{cases} X \sim \mathcal{N}(\Delta_{1}, \sigma^{2}),  \text{ if } S_{t} = 1, \\
        Y \sim \mathcal{N}(\Delta_{2}, \sigma^{2}), \text{ if } S_{t} = 2,
    \end{cases} 
\end{equation}

where \( \Delta_{1} < 0 \) and \( \Delta_{2} = -v \Delta_{1} \) with \( v \in [0, 1] \), and we write \( \Delta = (\Delta_{1}, \Delta_{2})^{\prime} \). 

We consider a sample of \( n \) observations.  
We denote \( d_{t},  t \in I = \{ 1, \ldots, n \} \) the realizations of \( D_{t} \), \( I_{1} = \{t \in I: s_{t}=1\} \), \( I_{2} = \{t \in I: s_{t}=2\} \), \( n_{1} = |I_{1} | \) and \( n_{2} = |I_{2} | \).\footnote{We use subscript \( l \) when indexing \( \Delta_{l}, n_{l}, I_{l} \) to simplify the notation of squared terms.}

We have the unconditional average loss differential \( \bar{d} = \dfrac{1}{n} \displaystyle \sum_{t \in I} d_{t}\), as well as the two conditional averages \( \bar{d}^{1} = \dfrac{1}{n_{1}} \displaystyle \sum_{t \in I_{1}} d_{t} \) and \( \bar{d}^{2} = \dfrac{1}{n_{2}} \displaystyle \sum_{t \in I_{2}} d_{t}\). 

We consider the vector of test functions \( h_{t} = (1, \mathbbm{I}_{\{s_{t}=1\}})^{\prime} \in R^{2 \times 1}\), which yields the instrumented loss differential \( z_{t} = h_{t} d_{t} = (d_{t}, d_{t}\mathbbm{I}_{\{s_{t}=1\}} )^{\prime} \), and we write \( z_{t}^{1} \) for \( d_{t} \mathbbm{I}_{\{s_{t}=1\}} \). 
For the Wald type test of \textcite{Giacomini.2006}, we obtain 
\begin{align*}
    T^{h} &= n \bar{z}^{\prime} \hat{\Sigma}^{-1}  \bar{z} \\
    &= n \left( \bar{d}, \bar{z^{1}} \right)  \hat{\Sigma}^{-1} \left( \bar{d}, \bar{z^{1}} \right){\prime} \\
    &= n \left[ (\bar{d})^{2} \hat{\Sigma}^{-1}_{11} + 2 \bar{d} \bar{z^{1}} \hat{\Sigma}^{-1}_{12} + (\bar{z^{1}})^{2}  \hat{\Sigma}^{-1}_{22} \right] \\
    &= n \left[ (\bar{d})^{2} \hat{\Sigma}^{-1}_{11} + 2 \bar{d} \dfrac{n_{1}}{n}\bar{d^{1}} \hat{\Sigma}^{-1}_{12} + \dfrac{n_{1}^{2}}{n^{2}}(\bar{d^{1}})^{2}  \hat{\Sigma}^{-1}_{22} \right] \\
    & \equiv n \left[ A + B + C \right].
\end{align*}
We rewrite \( A \) and \( B \) explicitly in terms of the conditional sample averages  \( \bar{d}^{1} \) and \( \bar{d}^{2} \). 
\begin{align*}
    A &=  (\bar{d})^{2} \hat{\Sigma}^{-1}_{11} \\ 
        &= \dfrac{1}{n^{2}} \left( \displaystyle \sum_{t \in I_{1}} d_{t} + \displaystyle \sum_{t \in I_{2}} d_{t} \right)^{2}  \hat{\Sigma}^{-1}_{11}\\
        &=\dfrac{1}{n^{2}} \left( n_{1}^{2} (\overline{d^{1}})^{2} + 2 n_{1} n_{2} \overline{d^{1}} \; \overline{d^{2}} + n_{2}^{2} (\overline{d^{2}})^{2} \right) \hat{\Sigma}^{-1}_{11},\\
    B &=  2 \bar{d} \dfrac{n_{1}}{n} \bar{d^{1}} \hat{\Sigma}^{-1}_{21} \\ 
    &\overset{*}{=} 2 \dfrac{n_{1}}{n^{2}} \left( n_{2} \bar{d^{1}} \bar{d^{2}} + n_{1} (\bar{d^{1}})^{2} \right) \hat{\Sigma}^{-1}_{21} ,\\
    C &=  \dfrac{n_{1}^{2}}{n^{2}}(\bar{d^{1}})^{2}  \hat{\Sigma}^{-1}_{22},\\
	& \mathbb{*} \\
	\bar{d} \bar{d^{1}} &= \dfrac{1}{n n_{1}} \displaystyle \sum_{t \in I} d_{t} \displaystyle \sum_{t \in I_{1}} d_{t}  \\
	&= \dfrac{1}{n n_{1}}  \left[ n_{2} \bar{d^{2}} n_{1} \bar{d^{1}} + n_{1}^{2} (\bar{d^{1}})^{2}  \right]  \\
	&=  \dfrac{n_{2}}{n} \bar{d^{2}} \bar{d^{1}} + \dfrac{n_{1}}{n}  (\bar{d^{1}})^{2},
\end{align*}
We consider a simplified expression for \( T^{h} \) in the plots in Section \ref{sec:simulation_wald_vs_t_rejection_regions} by using the true covariance matrix \( \Sigma \) and setting \( n_{1} = np \) and \( n_{2}=n(1-p) \). We can thus write \( T^{h} \) as a function of the sample average of the conditional loss differentials and obtain
\begin{align}\label{eq:S_h_D_E_F}
    T^{h} &= n \left[ A + B + C \right] \nonumber \\
    &=n \bigl [ \dfrac{1}{n^{2}} \left( n_{1}^{2} (\overline{d^{1}})^{2} + 2 n_{1} n_{2} \overline{d^{1}} \; \overline{d^{2}} + n_{2}^{2} (\overline{d^{2}})^{2} \right) \Sigma^{-1}_{11} + \nonumber \\
	&+ 2 \dfrac{n_{1}}{n^{2}} \left( n_{2} \bar{d^{1}} \bar{d^{2}} + n_{1} (\bar{d^{1}})^{2} \right) \Sigma^{-1}_{21}  + \dfrac{n_{1}^{2}}{n^{2}}(\bar{d^{1}})^{2} \Sigma^{-1}_{22} \bigr ] \nonumber \\
    &=n \bigl [ \left( p^{2} (\overline{d^{1}})^{2} + 2 p (1-p) \overline{d^{1}} \; \overline{d^{2}} + (1-p)^{2} (\overline{d^{2}})^{2} \right) \Sigma^{-1}_{11} \nonumber \\
	&+ 2 \left( p(1-p) \bar{d^{1}} \bar{d^{2}} + p^{2} (\bar{d^{1}})^{2} \right) \Sigma^{-1}_{21}  + p^{2}(\bar{d^{1}})^{2} \Sigma^{-1}_{22} \bigr ] \nonumber \\
    &=n \left[ p^{2} (\overline{d^{1}})^{2} \left(  \Sigma^{-1}_{11} + 2 \Sigma^{-1}_{12} +\Sigma^{-1}_{22}  \right) +2 p(1-p) \bar{d^{1}} \bar{d^{2}} \left( \Sigma^{-1}_{11} + \Sigma^{-1}_{12} \right) +  (1-p)^{2} (\overline{d^{2}})^{2}  \Sigma^{-1}_{11}  \right]  \nonumber \\ 
    & \equiv n \left[ D + E + F \right].
\end{align}

\subsubsection*{Calculation of \( \Sigma, \Sigma^{-1} \)}

We obtain as the covariance matrix \( \Sigma\) of \( (D_{t}, D_{t} \mathbbm{1}_{{S_{t}=1}})^{\prime} \)
\begin{align*}
    \Sigma &= \begin{pmatrix}
        \sigma^{2} + p (1-p) (\Delta_{1} - \Delta_{2})^{2}              &   p \sigma^{2} + p(1-p)(\Delta_{1}^{2}-\Delta_{1} \Delta_{2}) \\
        p \sigma^{2} + p(1-p)(\Delta_{1}^{2}-\Delta_{1} \Delta_{2})     &    p \sigma^{2} + p(1-p) \Delta_{1}^{2}  \\
        \end{pmatrix} 
\end{align*}
where 
\begin{align*}
\Sigma_{11} =  \sigma^{2} + p (1-p) (\Delta_{1} - \Delta_{2})^{2} \\
\end{align*}
is the variance of the mixture of normals \( D_{t} \) according to Equation \ref{eq:Wald_vs_t_loss_differential}, 
\begin{align*}
    \Sigma_{12} &= Cov(D_{t}, D_{t} \cdot \mathbbm{1}_{{S_{t}=1}}) \\
    &= E[D_{t} \cdot D_{t} \cdot \mathbbm{1}_{{S_{t}=1}}] - E[D_{t}] \cdot E[{D_{t} \cdot \mathbbm{1}_{{S_{t}=1}}}] \\
    &= E[D_{t}^{2} | S_{t}=1 ] \cdot E[\mathbbm{1}_{{S_{t}=1}}] - E[D_{t}] E[D_{t} | S_{t}=1] \cdot  E[\mathbbm{1}_{{S_{t}=1}}] \\
    &= p (\sigma^{2} + \Delta_{1}^{2}) - (p \Delta_{1} + (1-p) \Delta_{2}) p \Delta_{1} \\
    &= p \sigma^{2} + p(1-p)(\Delta_{1}^{2}-\Delta_{1} \Delta_{2}) \\
\Sigma_{22} &= Var(D_{t} \cdot \mathbbm{1}_{{S_{t}=1}})\\
&= E[(D_{t} \cdot \mathbbm{1}_{{S_{t}=1}})^{2}] - E[D_{t} \cdot \mathbbm{1}_{{S_{t}=1}}]^{2} \\
&= E[D_{t}^{2} | S_{t}=1 ] \cdot E[\mathbbm{1}_{{S_{t}=1}}] - E[D_{t} | S_{t}=1]^{2} \cdot  E[\mathbbm{1}_{{S_{t}=1}}]^{2} \\
&= p (\sigma^{2} + \Delta_{1}^{2}) - p^{2} \Delta_{1}^{2}\\
&= p \sigma^{2} + p(1-p) \Delta_{1}^{2} 
\end{align*}
Given that it exists, the inverse of the \( 2 \times 2 \) matrix \(  \Sigma \) is given by 
\begin{align*}
    \Sigma^{-1} = \dfrac{1}{det (\Sigma)}  
    \begin{pmatrix}
        \Sigma_{22}            &   - \Sigma_{12} \\
        - \Sigma_{12}    &    \Sigma_{11}  \\
        \end{pmatrix},
\end{align*}
with
\begin{align}
    det (\Sigma) = (1 - p) p \sigma^2 \left( (1 - p) \Delta_1^2 + p \Delta_2^2 + \sigma^2 \right).
\end{align}

\subsubsection*{Plugging into \( T^{h} \) }
Thus, as factors in \( D, E, F \) in the expression of \( T^{h} \)  in Equation \ref{eq:S_h_D_E_F} above we obtain 
\begin{align*}  
    \Sigma^{-1}_{11} + 2 \Sigma^{-1}_{12} +\Sigma^{-1}_{22}   
    &= \dfrac{1}{det (\Sigma)}  \times \Bigl ( p \sigma^{2} + p(1-p) \Delta_{1}^{2} - 2 [p \sigma^{2} + p(1-p)(\Delta_{1}^{2}-\Delta_{1} \Delta_{2})]   \\
	&+ \sigma^{2} + p (1-p) (\Delta_{1} - \Delta_{2})^{2} \Bigr )\\
    &= \dfrac{1}{det (\Sigma)}  \times \left( (1-p) \sigma^{2} + p(1-p) \Delta_{2}^{2} \right)\\
    &= \dfrac{ (1-p) (\sigma^{2} + p \Delta_{2}^{2}) }{(1 - p) p \sigma^2 \left( (1 - p) \Delta_1^2 + p \Delta_2^2 + \sigma^2 \right)} \\
    &= \dfrac{ \sigma^{2} + p \Delta_{2}^{2} }{ p \sigma^2 \left( (1 - p) \Delta_1^2 + p \Delta_2^2 + \sigma^2 \right)} \\
    \Sigma^{-1}_{11} + \Sigma^{-1}_{12}
    &= \dfrac{1}{det (\Sigma)}  \times \left( p \sigma^{2} + p(1-p) \Delta_{1}^{2} - p \sigma^{2} + p(1-p)(\Delta_{1}^{2}-\Delta_{1} \Delta_{2})\right)\\
    &= \dfrac{1}{det (\Sigma)}  \times \left( p(1-p) \Delta_{1} \Delta_{2} \right)\\
    &= \dfrac{ p(1-p) \Delta_{1} \Delta_{2}}{(1 - p) p \sigma^2 \left( (1 - p) \Delta_1^2 + p \Delta_2^2 + \sigma^2 \right)} \\
    &= \dfrac{\Delta_{1} \Delta_{2}}{\sigma^2 \left( (1 - p) \Delta_1^2 + p \Delta_2^2 + \sigma^2 \right)} \\
    \Sigma^{-1}_{11} 
	&= \dfrac{1}{det (\Sigma)}  \times \left( p \sigma^{2} + p(1-p) \Delta_{1}^{2} \right)\\
    &= \dfrac{ p \sigma^{2} + p(1-p) \Delta_{1}^{2} }{(1 - p) p \sigma^2 \left( (1 - p) \Delta_1^2 + p \Delta_2^2 + \sigma^2 \right)} \\
    &= \dfrac{ \sigma^{2} + (1-p) \Delta_{1}^{2} }{(1 - p) \sigma^2 \left( (1 - p) \Delta_1^2 + p \Delta_2^2 + \sigma^2 \right)} \\
\end{align*}
Overall, this yields
\begin{align*}
    T^{h} &= n \left[ D + E + F \right] \\
    &=n \left[ p^{2} (\overline{d^{1}})^{2} \left(  \Sigma^{-1}_{11} + 2 \Sigma^{-1}_{12} +\Sigma^{-1}_{22}  \right) +2 p(1-p) \bar{d^{1}} \bar{d^{2}} \left( \Sigma^{-1}_{11} + \Sigma^{-1}_{12} \right) +  (1-p)^{2} (\overline{d^{2}})^{2}  \Sigma^{-1}_{11}  \right] \\ 
    &=n  \bigl [ p^{2} (\overline{d^{1}})^{2} \dfrac{ \sigma^{2} + p \Delta_{2}^{2} }{ p \sigma^2 \left( (1 - p) \Delta_1^2 + p \Delta_2^2 + \sigma^2 \right)} +2 p(1-p) \bar{d^{1}} \bar{d^{2}} \dfrac{\Delta_{1} \Delta_{2}}{\sigma^2 \left( (1 - p) \Delta_1^2 + p \Delta_2^2 + \sigma^2 \right)} \\
    &+  (1-p)^{2} (\overline{d^{2}})^{2}  \dfrac{ \sigma^{2} + (1-p) \Delta_{1}^{2} }{(1 - p) \sigma^2 \left( (1 - p) \Delta_1^2 + p \Delta_2^2 + \sigma^2 \right)}  \bigr ] \\ 
    &< n_{1} \dfrac{(\overline{d^{1}})^{2}}{\sigma^{2}} + n [E +F] .
\end{align*}

As before, we denote \( c( \mathcal{D}, 1-\alpha) \) the critical value of a distribution \( \mathcal{D} \) at the \( 1-\alpha \) confidence level. 
We see that if \( c(\chi^{2}_{1}, 1-\alpha) < n_{1} \dfrac{(\overline{d^{1}})^{2}}{\sigma^{2}} \) , but \( n D < c(\chi^{2}_{2}, 1-\alpha) \) \footnote{It is clear that \( c(\chi^{2}_{1}, 1-\alpha) < c(\chi^{2}_{2}, 1-\alpha) \) and  \( n_{1} \dfrac{(\overline{d^{1}})^{2}}{\sigma^{2}} >  n D  \) }, i.e., when the statewise t test rejects, but the Wald type test would not reject solely on the evidence that comes from state 1, a rejection depends on the evidence that comes from state 2 in terms \( E \) and \( F \). Additionally, \( E \) and \( F \) need to compensate for the inflated denominator in term \( A \). This illustrates why for small values of \( \Delta_{2} \), the Wald type test has less power than the statewise test in state 1. Note that the term \( E \) is positive if \( sgn(\bar{d^{1}}) = sgn(\Delta_{1}) \) and \( sgn(\bar{d^{2}}) =  sgn(\Delta_{2}) \) and then makes a rejection of the Wald type test more likely.

\subsubsection*{Special case: \( \Delta_{2} = 0 \)}

In the special case of \( \Delta_{2} = 0 \), we have as factors in \( D, E, F \) in the expression of \( T^{h} \)  in Equation \ref{eq:S_h_D_E_F} above 
\begin{align*}
    \Sigma^{-1}_{11} + 2 \Sigma^{-1}_{12} +\Sigma^{-1}_{22}   &= \dfrac{1}{p \sigma^{2} + (1-p)p \Delta_{1}^{2}} \\
    \Sigma^{-1}_{11} + \Sigma^{-1}_{12} &= 0 \\
    \Sigma^{-1}_{11} &= \dfrac{1}{(1-p) \sigma^{2}}.
\end{align*}
Overall, we obtain 
\begin{align*}
    T^{h} &= n \left[ D + E + F \right] \\
	&=n \left[ p^{2} (\overline{d^{1}})^{2} \left(  \Sigma^{-1}_{11} + 2 \Sigma^{-1}_{12} +\Sigma^{-1}_{22}  \right) +2 p(1-p) \bar{d^{1}} \bar{d^{2}} \left( \Sigma^{-1}_{11} + \Sigma^{-1}_{12} \right) +  (1-p)^{2} (\overline{d^{2}})^{2}  \Sigma^{-1}_{11}  \right] \\ 
	&=n \left[ p^{2} (\overline{d^{1}})^{2} \dfrac{1}{p \sigma^{2} + (1-p)p \Delta_{1}^{2}}  +2 p(1-p) \bar{d^{1}} \bar{d^{2}} \cdot 0 +  (1-p)^{2} (\overline{d^{2}})^{2}  \dfrac{1}{(1-p) \sigma^{2}}  \right] \\ 
	&= n_{1}  \dfrac{ (\overline{d^{1}})^{2} }{\sigma^{2} + (1-p)\Delta_{1}^{2}} + n_{2} \dfrac{(\overline{d^{2}})^{2}  }{\sigma^{2}} .
\end{align*}
We see that the second term lives under the null of \( \Delta_{2}= 0 \) and does not make rejections of the Wald type test more likely as \( n_{2} \) grows, while the denominator of the first term is inflated by the term \( (1-p) \Delta_{1}^{2} \) relative to the statistic based on the statewise loss differentials in state 1. This shows the lower power of the Wald type test relative to the statewise test in state 1. 
\clearpage

\subsection{Data} \label{app:data}
\textbf{Large cap stocks:} 

\noindent Abbott Laboratories (ABT), Adobe Incorporation (ADBE), Applied Materials Incorporation (AMAT), Air Products and Chemicals Incorporation (APD),  
Bank of America (BAC), Franklin Resources Incorporation (BEN), Bank of New York Mellon Corporation (BK), Colgate-Palmolive Company (CL), Chevron Corporation (CVX), Dominion Energy Incorporation (D), Emerson Electric Company (EMR), Nextera Energy Incorporation (FPL), General Electric Company (GE), Illinois Tool Works Incorporation (ITW), Johnson \& Johnson (JNJ), JP Morgan Chase and Company (JPM), Coca-Cola Company (KO), Lowe's Companies Incorporation (LOW), 3M (MMM), Merck \& Co. Incorporation (MRK), Occidental Petroleum Corporation (OXY), Paccar Incorporation (PCAR), PepsiCo Incorporation (PEP), Pfizer Incorporation (PFE), Procter \& Gamble Company (PG), Southern Company (SO), Sysco Corporation (SYY), Target Corporation (TGT), Thermo Fisher Scientific Incorporation (TMO), Walmart Incorporation (WMT), Exxon Mobil Corporation (XOM).

\begin{table}[h!]
	\centering \footnotesize
	\caption{Descriptive statistics of large cap stocks}
	\begin{tabular}{l|ccc|ccc} \hline \hline
		& \multicolumn{3}{c|}{In-sample} & \multicolumn{3}{c}{Out-of-sample} \\ \hline
		stock & Mean  & Skewness & Kurtosis & Mean  & Skewness & Kurtosis \\ \hline
		ABT   & -0.005 & -14.792 & 557.536 & 0.096 & -0.398 & 3.001 \\
    ADBE  & 0.022 & -0.628 & 13.738 & 0.141 & 0.173 & 7.411 \\
    AMAT  & -0.024 & 0.289 & 4.828 & 0.067 & -0.563 & 2.231 \\
    APD   & 0.031 & -0.412 & 4.879 & 0.043 & -0.013 & 3.106 \\
    BAC   & -0.022 & -0.348 & 25.807 & 0.090 & -0.375 & 3.039 \\
    BEN   & -0.006 & -15.213 & 589.914 & -0.031 & -0.020 & 7.697 \\
    BK    & -0.015 & -0.249 & 18.934 & 0.021 & -0.857 & 5.195 \\
    CL    & 0.004 & -17.033 & 676.708 & -0.019 & -0.503 & 6.272 \\
    CVX   & 0.017 & -0.378 & 7.388 & 0.024 & -0.277 & 2.816 \\
    D     & 0.024 & -0.776 & 10.400 & 0.004 & -0.788 & 3.012 \\
    EMR   & 0.009 & -0.223 & 6.594 & 0.033 & -0.181 & 1.803 \\
    FPL   & 0.037 & 0.073 & 8.878 & 0.061 & -0.537 & 3.381 \\
    GE    & -0.016 & 0.037 & 9.110 & -0.187 & -0.259 & 4.406 \\
    ITW   & 0.028 & -0.055 & 4.503 & 0.037 & -1.371 & 5.758 \\
    JNJ   & 0.019 & -0.776 & 16.677 & 0.029 & -2.054 & 18.538 \\
    JPM   & 0.001 & 0.267 & 13.376 & 0.072 & -0.349 & 3.506 \\
    KO    & 0.010 & -0.304 & 7.203 & 0.014 & -0.628 & 2.795 \\
    LOW   & 0.040 & 0.295 & 4.989 & 0.031 & 0.418 & 6.582 \\
    MMM   & 0.031 & -0.107 & 5.217 & 0.028 & -0.993 & 6.265 \\
    MRK   & -0.011 & -1.895 & 32.944 & 0.048 & 0.727 & 10.287 \\
    OXY   & 0.046 & -0.355 & 8.841 & -0.015 & -0.055 & 2.704 \\
    PCAR  & 0.044 & 0.020 & 4.417 & 0.015 & -0.587 & 2.870 \\
    PEP   & 0.020 & -0.072 & 12.930 & 0.017 & -0.286 & 2.826 \\
    PFE   & -0.009 & -0.406 & 5.566 & 0.041 & 0.196 & 5.460 \\
    PG    & 0.023 & -0.222 & 6.709 & 0.023 & 0.576 & 11.279 \\
    SO    & 0.022 & 0.031 & 6.228 & -0.008 & -0.309 & 1.194 \\
    SYY   & 0.021 & 0.132 & 6.640 & 0.040 & -0.788 & 17.339 \\
    TGT   & 0.020 & 0.051 & 5.564 & -0.005 & -1.326 & 10.984 \\
    TMO   & 0.048 & -0.126 & 4.848 & 0.078 & -0.174 & 3.417 \\
    WMT   & 0.009 & 0.114 & 5.705 & 0.043 & 0.281 & 19.986 \\
    XOM   & 0.016 & -0.298 & 8.909 & -0.028 & -0.358 & 3.185 \\ \hline \hline
	\end{tabular}%
	\label{tab:dstat}%
\caption*{\scriptsize Table rows correspond to stocks. The left block of the table reports sample mean, skewness and kurtosis for in-sample period and the right block for the out-of-sample period.}
\end{table}%

\begin{table}[h!]
	\centering \small
	\caption{Implemented risk factors}
	\begin{tabular}{l|p{13cm}} \hline \hline
		$LH = 10$ & Interest rates: GS10 (Tbill), Sovereign bonds interest rates: UK, Japan, Australia, Canada, Sweden; exchange rates between ``standard'' pairs of currencies:
    USD/AUD, USD/BRL, USD/CAD, USD/CHF, USD/CNY, USD/EUR, USD/GBP, USD/HKD, 
    USD/INR, USD/JPY, USD/KRW, USD/MXN, USD/NOK, USD/NZD, USD/RUB, USD/SEK, 
    USD/SGD, USD/TRY, USD/ZAR \\
		$LH = 20$ & Sovereign bonds interest rates: France, Germany, Greece, Italy, Mexico, Netherlands, South Africa, Spain, Switzerland; Small cap index price; Large cap volatility (VIX); Equity market volatility tracker; exchange rates ``non-standard'' currencies USD/DKK, USD/LKR, USD/MYR, USD/THB, USD/TWD, USD/VEB; price of metals\\
		$LH = 40$ & Corporate bonds credit spread; exchnage rate volatility of currencies in $LH = 10$ \\
		$LH = 60$ & Interest rate volatility of sovereign bonds in $LH =10$; Volatility of small cap index; Other commodities price:  Potash, fertilizer, phosphate rocks, Rare earths,
    terephthalic acid, flat glass; Metal price volatility\\ 
        $LH = 120$ & Credit spread volatility; Food price: Barley, Maize, Sorghum,Rice, Wheat, Wheat, Banana, Banana, Orange, Beef, Chicken,
          Lamb, Shrimps, Mexican, Sugar\\ \hline \hline
	\end{tabular}%
    \caption*{\scriptsize The table lists examples of the risk factors suggested by the BCBS, e.g. p. 92 of the 2019 regulation. The variables listed in the table are sourced from St Louis Fed, Fed, Worldbank, Yahoo Finance.}
	\label{tab:LH_ours}%
\end{table}%


\clearpage
\subsection{Empirical application} \label{app:ea}

\subsubsection{One day ahead forecasts}
\begin{table}[h!]
	\centering
	\footnotesize
	\caption{Conditional Expected Shortfall forecasts $t+1$ for CMCS.}
	\begin{tabular}{lccccccc} \hline\hline
    & UC    & $LH = 10$ & $LH =20$ & $LH = 40$ & $LH = 60$ & $LH = 120$ & $ES_{BCBS}$ \bigstrut[b]\\
    \hline 
    ABT   & -1.35 & -1.39 & -1.39 & -1.66 & -1.68 & -1.67 & 17.52 \bigstrut[t]\\
    ADBE  & -1.71 & -1.71 & -1.71 & -1.68 & -1.67 & -1.66 & 17.64 \\
    AMAT  & -2.01 & -2.00 & -2.00 & -2.04 & -2.02 & -2.01 & 21.30 \\
    APD   & -1.31 & -1.24 & -1.24 & -1.22 & -1.22 & -1.22 & 12.85 \\
    BAC   & -1.60 & -1.64 & -1.64 & -1.60 & -1.61 & -1.61 & 17.01 \\
    BEN   & -2.12 & -1.85 & -1.85 & -1.99 & -1.91 & -1.88 & 20.24 \\
    BK    & -1.37 & -1.46 & -1.46 & -1.41 & -1.40 & -1.41 & 14.90 \\
    CL    & -1.10 & -1.37 & -1.37 & -1.34 & -1.28 & -1.22 & 13.57 \\
    CVX   & -1.56 & -1.43 & -1.43 & -1.49 & -1.46 & -1.45 & 15.36 \\
    D     & -1.25 & -1.36 & -1.36 & -1.30 & -1.32 & -1.31 & 13.90 \\
    EMR   & -1.46 & -1.47 & -1.47 & -1.46 & -1.45 & -1.45 & 15.34 \\
    FPL   & -1.21 & -1.29 & -1.29 & -1.32 & -1.33 & -1.31 & 14.01 \\
    GE    & -1.74 & -1.78 & -1.78 & -1.77 & -1.75 & -1.75 & 18.52 \\
    ITW   & -1.27 & -1.27 & -1.27 & -1.25 & -1.24 & -1.24 & 13.14 \\
    JNJ   & -1.06 & -1.03 & -1.03 & -1.04 & -1.03 & -1.03 & 10.91 \\
    JPM   & -1.51 & -1.37 & -1.37 & -1.34 & -1.34 & -1.34 & 14.15 \\
    KO    & -0.94 & -0.96 & -0.96 & -0.99 & -1.00 & -0.99 & 10.47 \\
    LOW   & -1.56 & -1.51 & -1.51 & -1.51 & -1.52 & -1.54 & 16.12 \\
    MMM   & -1.18 & -1.25 & -1.25 & -1.24 & -1.23 & -1.21 & 12.92 \\
    MRK   & -1.21 & -1.24 & -1.24 & -1.24 & -1.27 & -1.29 & 13.45 \\
    OXY   & -2.31 & -1.63 & -1.63 & -1.63 & -1.62 & -1.62 & 17.20 \\
    PCAR  & -1.67 & -1.55 & -1.55 & -1.55 & -1.55 & -1.56 & 16.41 \\
    PEP   & -0.98 & -0.94 & -0.94 & -0.95 & -0.93 & -0.96 & 10.03 \\
    PFE   & -1.10 & -1.13 & -1.13 & -1.13 & -1.11 & -1.11 & 11.79 \\
    PG    & -0.97 & -1.02 & -1.02 & -1.04 & -1.04 & -1.03 & 10.90 \\
    SO    & -1.12 & -1.11 & -1.11 & -1.07 & -1.07 & -1.09 & 11.45 \\
    SYY   & -1.17 & -1.25 & -1.25 & -1.24 & -1.23 & -1.24 & 13.16 \\
    TGT   & -1.79 & -1.84 & -1.84 & -1.96 & -1.94 & -1.89 & 20.22 \\
    TMO   & -1.35 & -1.36 & -1.36 & -1.33 & -1.31 & -1.32 & 13.94 \\
    WMT   & -1.24 & -1.27 & -1.27 & -1.33 & -1.32 & -1.32 & 13.94 \\
    XOM   & -1.28 & -1.23 & -1.23 & -1.22 & -1.22 & -1.21 & 12.85 \\
    \hline\hline
    \end{tabular}%
	\caption*{\scriptsize Numbers in the table correspond to out-of-sample average ES forecasts for different stocks (in rows) and liquidity horizons in columns. The last column corresponds to ES coefficient defined in \eqref{eq:ES_bcbs}.}
	\label{tab:cmcs1}%
\end{table}%

\begin{table}[h!]
	\centering \footnotesize
	\caption{Conditional Expected Shortfall forecasts $t+1$ for DFC with sample covariance estimator.}
	\begin{tabular}{lccccccc} \hline\hline
    & UC    & $LH = 10$ & $LH =20$ & $LH = 40$ & $LH = 60$ & $LH = 120$ & $ES_{BCBS}$ \bigstrut[b]\\
    \hline 
    ABT   & -1.51 & -1.51 & -1.51 & -1.51 & -1.51 & -1.51 & 15.92 \bigstrut[t]\\
    ADBE  & -1.64 & -1.64 & -1.64 & -1.64 & -1.64 & -1.64 & 17.26 \\
    AMAT  & -1.90 & -1.90 & -1.90 & -1.90 & -1.90 & -1.90 & 19.99 \\
    APD   & -1.22 & -1.22 & -1.22 & -1.22 & -1.22 & -1.22 & 12.90 \\
    BAC   & -1.64 & -1.64 & -1.64 & -1.64 & -1.64 & -1.64 & 17.32 \\
    BEN   & -2.30 & -2.31 & -2.31 & -2.30 & -2.30 & -2.30 & 24.22 \\
    BK    & -1.35 & -1.35 & -1.35 & -1.35 & -1.35 & -1.35 & 14.19 \\
    CL    & -1.49 & -1.49 & -1.49 & -1.49 & -1.49 & -1.49 & 15.74 \\
    CVX   & -1.33 & -1.33 & -1.33 & -1.33 & -1.33 & -1.33 & 14.04 \\
    D     & -1.06 & -1.06 & -1.06 & -1.06 & -1.06 & -1.06 & 11.12 \\
    EMR   & -1.36 & -1.36 & -1.36 & -1.36 & -1.36 & -1.36 & 14.37 \\
    FPL   & -1.02 & -1.02 & -1.02 & -1.02 & -1.02 & -1.02 & 10.74 \\
    GE    & -1.65 & -1.65 & -1.65 & -1.65 & -1.65 & -1.65 & 17.36 \\
    ITW   & -1.17 & -1.17 & -1.17 & -1.17 & -1.17 & -1.17 & 12.36 \\
    JNJ   & -0.95 & -0.95 & -0.95 & -0.95 & -0.95 & -0.95 & 10.04 \\
    JPM   & -1.32 & -1.32 & -1.32 & -1.32 & -1.32 & -1.32 & 13.87 \\
    KO    & -0.90 & -0.90 & -0.90 & -0.90 & -0.90 & -0.90 & 9.46 \\
    LOW   & -1.59 & -1.59 & -1.59 & -1.59 & -1.59 & -1.59 & 16.76 \\
    MMM   & -1.14 & -1.14 & -1.14 & -1.14 & -1.14 & -1.14 & 11.98 \\
    MRK   & -1.20 & -1.20 & -1.20 & -1.20 & -1.20 & -1.20 & 12.67 \\
    OXY   & -1.54 & -1.54 & -1.54 & -1.54 & -1.54 & -1.54 & 16.27 \\
    PCAR  & -1.49 & -1.49 & -1.49 & -1.49 & -1.49 & -1.49 & 15.67 \\
    PEP   & -0.88 & -0.88 & -0.88 & -0.88 & -0.88 & -0.88 & 9.22 \\
    PFE   & -1.12 & -1.12 & -1.12 & -1.12 & -1.12 & -1.12 & 11.79 \\
    PG    & -0.94 & -0.94 & -0.94 & -0.94 & -0.94 & -0.94 & 9.91 \\
    SO    & -1.00 & -1.00 & -1.00 & -1.00 & -1.00 & -1.00 & 10.56 \\
    SYY   & -1.14 & -1.14 & -1.14 & -1.14 & -1.14 & -1.14 & 11.99 \\
    TGT   & -1.57 & -1.57 & -1.57 & -1.57 & -1.57 & -1.57 & 16.50 \\
    TMO   & -1.30 & -1.30 & -1.30 & -1.30 & -1.30 & -1.30 & 13.70 \\
    WMT   & -1.25 & -1.25 & -1.25 & -1.25 & -1.25 & -1.25 & 13.21 \\
    XOM   & -1.15 & -1.15 & -1.15 & -1.15 & -1.15 & -1.15 & 12.10 \\
    
    \hline\hline
    \end{tabular}%
	\caption*{\scriptsize Numbers in the table correspond to out-of-sample average ES forecasts for different stocks (in rows) and liquidity horizons in columns. The last column corresponds to ES coefficient defined in \eqref{eq:ES_bcbs}.}
	\label{tab:bt_sample1}%
\end{table}%

\begin{table}[h!]
	\centering \footnotesize
	\caption{Conditional Expected Shortfall forecasts $t+1$ for DFC with truncation $T/4$.}
	\begin{tabular}{lccccccc} \hline\hline
    & UC    & $LH = 10$ & $LH =20$ & $LH = 40$ & $LH = 60$ & $LH = 120$ & $ES_{BCBS}$ \bigstrut[b]\\
    \hline 
    ABT   & -2.94 & -2.99 & -2.99 & -2.98 & -2.96 & -2.99 & 31.50 \bigstrut[t]\\
    ADBE  & -4.03 & -4.05 & -4.05 & -4.12 & -4.23 & -4.26 & 44.78 \\
    AMAT  & -4.29 & -4.53 & -4.53 & -4.41 & -4.47 & -4.53 & 47.77 \\
    APD   & -2.84 & -2.61 & -2.61 & -2.88 & -2.89 & -2.86 & 30.29 \\
    BAC   & -3.50 & -3.58 & -3.58 & -3.62 & -3.65 & -3.57 & 37.98 \\
    BEN   & -3.09 & -3.46 & -3.46 & -3.44 & -3.40 & -3.41 & 36.20 \\
    BK    & -3.09 & -3.08 & -3.08 & -2.91 & -3.05 & -3.04 & 32.28 \\
    CL    & -2.25 & -2.21 & -2.21 & -2.24 & -2.22 & -2.22 & 23.41 \\
    CVX   & -2.79 & -2.89 & -2.89 & -2.95 & -2.95 & -2.95 & 31.15 \\
    D     & -1.32 & -1.45 & -1.45 & -1.42 & -1.44 & -1.33 & 36.52 \\
    FPL   & -2.42 & -2.38 & -2.38 & -2.44 & -2.43 & -2.39 & 25.56 \\
    GE    & -3.30 & -3.17 & -3.17 & -3.36 & -3.29 & -3.27 & 34.87 \\
    ITW   & -2.98 & -2.71 & -2.71 & -2.85 & -2.80 & -2.75 & 29.96 \\
    JNJ   & -1.89 & -1.93 & -1.93 & -1.64 & -1.69 & -1.70 & 18.37 \\
    JPM   & -3.07 & -3.02 & -3.02 & -2.93 & -2.99 & -3.11 & 32.50 \\
    KO    & -1.94 & -1.91 & -1.91 & -1.94 & -1.94 & -1.95 & 20.50 \\
    LOW   & -2.61 & -2.58 & -2.58 & -2.51 & -2.51 & -2.50 & 26.52 \\
    MMM   & -2.31 & -2.20 & -2.20 & -2.23 & -2.29 & -2.30 & 24.15 \\
    MRK   & -2.44 & -2.54 & -2.54 & -2.44 & -2.46 & -2.49 & 26.35 \\
    OXY   & -3.26 & -3.25 & -3.25 & -3.27 & -3.18 & -3.17 & 33.97 \\
    PCAR  & -2.79 & -2.60 & -2.60 & -2.65 & -2.60 & -2.58 & 27.79 \\
    PEP   & -1.68 & -1.65 & -1.65 & -1.65 & -1.70 & -1.71 & 17.97 \\
    PG    & -2.02 & -1.91 & -1.91 & -1.96 & -1.95 & -1.98 & 20.83 \\
    SO    & -2.26 & -2.22 & -2.22 & -2.28 & -2.30 & -2.33 & 24.49 \\
    SYY   & -1.80 & -1.82 & -1.82 & -1.83 & -1.84 & -1.83 & 19.35 \\
    TGT   & -3.24 & -3.14 & -3.14 & -3.24 & -3.22 & -3.19 & 33.79 \\
    TMO   & -2.99 & -2.81 & -2.81 & -2.84 & -2.80 & -2.82 & 30.03 \\
    WMT   & -2.22 & -2.15 & -2.15 & -2.30 & -2.26 & -2.26 & 23.81 \\
    XOM   & -2.60 & -2.63 & -2.63 & -2.46 & -2.49 & -2.47 & 26.57 \\

    \hline\hline
    \end{tabular}%
	\caption*{\scriptsize Numbers in the table correspond to out-of-sample average ES forecasts for different stocks (in rows) and liquidity horizons in columns. The last column corresponds to ES coefficient defined in \eqref{eq:ES_bcbs}.}
	\label{tab:bt_tr4_1}%
\end{table}%

\begin{table}[h!]
	\centering \footnotesize
	\caption{Conditional Expected Shortfall forecasts $t+1$ for DFC with truncation $T/2$.}
	\begin{tabular}{lccccccc} \hline\hline
    & UC    & $LH = 10$ & $LH =20$ & $LH = 40$ & $LH = 60$ & $LH = 120$ & $ES_{BCBS}$ \bigstrut[b]\\
    \hline 
    ABT   & -3.12 & -3.00 & -3.00 & -3.04 & -3.06 & -3.02 & 32.05 \bigstrut[t]\\
    ADBE  & -4.06 & -4.09 & -4.09 & -3.99 & -4.04 & -3.98 & 42.65 \\
    AMAT  & -4.50 & -4.23 & -4.23 & -4.60 & -4.62 & -4.66 & 48.88 \\
    APD   & -2.48 & -2.90 & -2.90 & -2.75 & -2.75 & -2.83 & 29.87 \\
    BAC   & -3.70 & -3.70 & -3.70 & -3.59 & -3.63 & -3.56 & 38.03 \\
    BEN   & -3.37 & -3.41 & -3.41 & -3.41 & -3.42 & -3.43 & 36.16 \\
    BK    & -3.01 & -3.06 & -3.06 & -2.96 & -3.02 & -2.97 & 31.90 \\
    CL    & -2.26 & -2.18 & -2.18 & -2.22 & -2.23 & -2.24 & 23.52 \\
    CVX   & -2.92 & -2.94 & -2.94 & -2.95 & -2.96 & -2.95 & 31.24 \\
    D     & -0.39 & -1.02 & -1.02 & -1.42 & -1.45 & -1.35 & 37.94 \\
    FPL   & -2.30 & -2.41 & -2.41 & -2.32 & -2.41 & -2.40 & 25.28 \\
    GE    & -3.15 & -3.24 & -3.24 & -3.27 & -3.30 & -3.30 & 34.91 \\
    ITW   & -2.95 & -3.05 & -3.05 & -2.90 & -2.93 & -2.95 & 31.36 \\
    JNJ   & -1.82 & -1.85 & -1.85 & -1.70 & -1.68 & -1.69 & 18.26 \\
    JPM   & -2.92 & -3.17 & -3.17 & -3.04 & -2.98 & -3.13 & 33.03 \\
    KO    & -1.91 & -1.92 & -1.92 & -1.94 & -1.93 & -1.94 & 20.43 \\
    LOW   & -2.45 & -2.62 & -2.62 & -2.47 & -2.48 & -2.41 & 26.15 \\
    MMM   & -2.37 & -2.29 & -2.29 & -2.20 & -2.15 & -2.05 & 22.64 \\
    MRK   & -2.50 & -2.50 & -2.50 & -2.52 & -2.52 & -2.51 & 26.53 \\
    OXY   & -3.29 & -3.23 & -3.23 & -3.06 & -3.13 & -3.18 & 33.60 \\
    PCAR  & -2.90 & -2.99 & -2.99 & -2.91 & -2.89 & -2.80 & 30.35 \\
    PEP   & -1.75 & -1.72 & -1.72 & -1.72 & -1.78 & -1.78 & 18.68 \\
    PG    & -1.94 & -1.98 & -1.98 & -2.00 & -1.99 & -1.97 & 20.99 \\
    SO    & -2.21 & -2.23 & -2.23 & -1.95 & -2.07 & -2.07 & 22.38 \\
    SYY   & -1.82 & -1.81 & -1.81 & -1.75 & -1.81 & -1.81 & 19.04 \\
    TGT   & -3.26 & -3.22 & -3.22 & -3.17 & -3.20 & -3.18 & 33.70 \\
    TMO   & -2.92 & -2.87 & -2.87 & -2.92 & -2.90 & -2.88 & 30.57 \\
    WMT   & -2.19 & -2.24 & -2.24 & -2.28 & -2.26 & -2.25 & 23.80 \\
    XOM   & -2.42 & -2.57 & -2.57 & -2.54 & -2.56 & -2.51 & 26.97 \\
    
    \hline\hline
    \end{tabular}%
	\caption*{\scriptsize Numbers in the table correspond to out-of-sample average ES forecasts for different stocks (in rows) and liquidity horizons in columns. The last column corresponds to ES coefficient defined in \eqref{eq:ES_bcbs}.}
	\label{tab:bt_tr2_1}%
\end{table}%

\begin{figure}[h!]
	\caption{Method selection for 1 day ahead ES forecasts for different liquidity horizons:  DFC with truncation $T/4$.}
	\centering
	\includegraphics[trim={0cm 0cm 0cm 0cm},clip,width = 0.9\textwidth]{t10/BT_stocks_BT_TR978.eps}
	\caption*{\scriptsize Heatmaps correspond to the liquidity horizons as specified in Table \ref{tab:LH}. On each heatmap the x-axis corresponds to stocks and the y-axis to models. The heatmap cells correspond to the average number of the out-of-sample periods where the method is included in the MCS. The warmer the colour, the more frequently was the method selected for the forecast combination.}
	\label{fig:BT4_stocks_1}
\end{figure}

\begin{figure}[h!]
	\caption{Method selection for 1 day ahead ES forecasts for different liquidity horizons:  DFC with truncation $T/2$.}
	\centering
	\includegraphics[trim={0cm 0cm 0cm 0cm},clip,width = 0.9\textwidth]{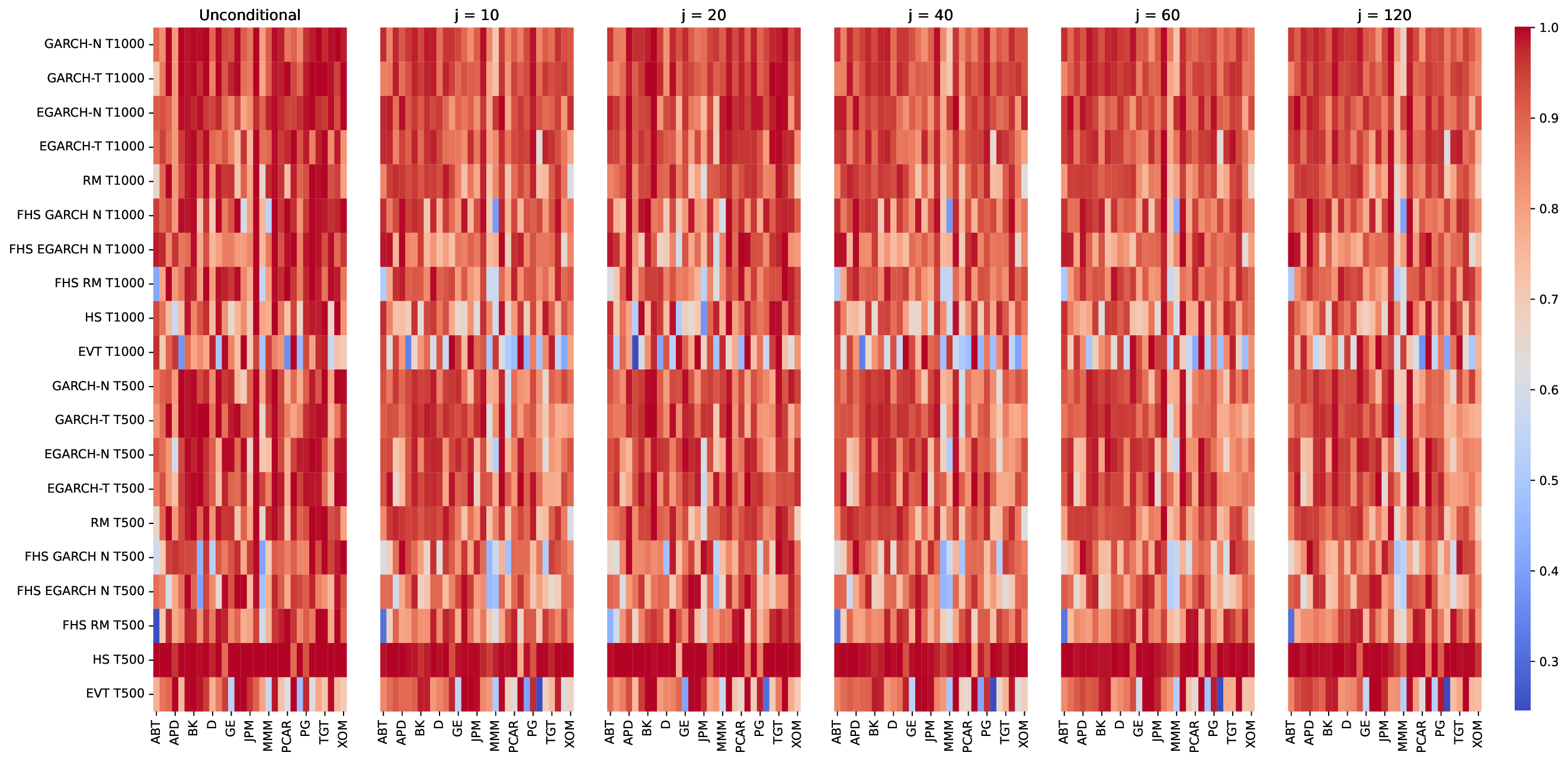}
	\caption*{\scriptsize Heatmaps correspond to the liquidity horizons as specified in Table \ref{tab:LH}. On each heatmap the x-axis corresponds to stocks and the y-axis to models. The heatmap cells correspond to the average number of the out-of-sample periods where the method is included in the MCS. The warmer the colour, the more frequently was the method selected for the forecast combination.}
	\label{fig:BT2_stocks_1}
\end{figure}

\begin{figure}[h!]
	\caption{Method selection for 1 day ahead ES forecasts for different liquidity horizons:  DFC with sample covariance estimator.}
	\centering
	\includegraphics[trim={0cm 0cm 0cm 0cm},clip,width = 0.9\textwidth]{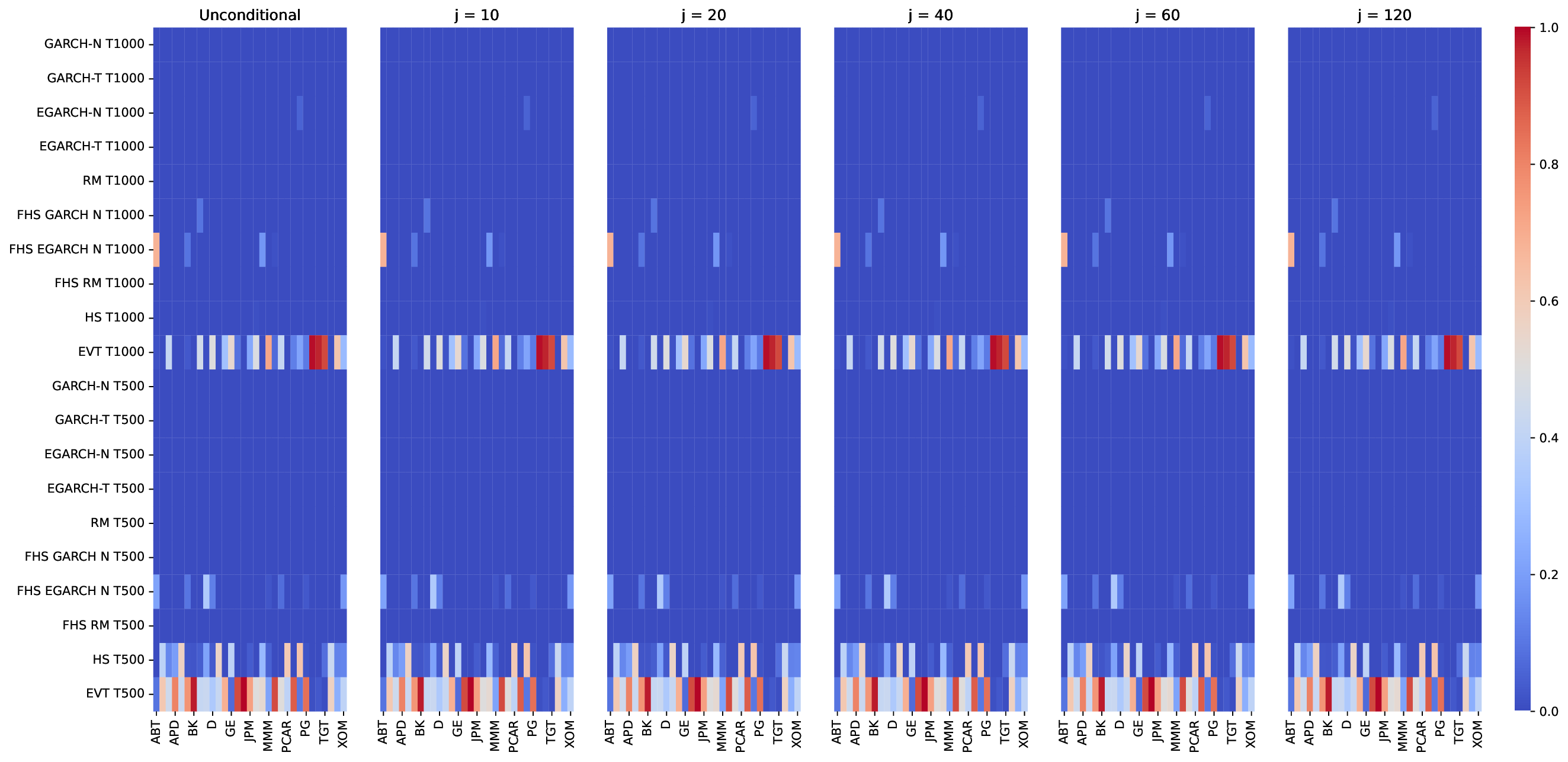}
	\caption*{\scriptsize Heatmaps correspond to the liquidity horizons as specified in Table \ref{tab:LH}. On each heatmap the x-axis corresponds to stocks and the y-axis to models. The heatmap cells correspond to the average number of the out-of-sample periods where the method is included in the MCS. The warmer the colour, the more frequently was the method selected for the forecast combination.}
	\label{fig:BTs_stocks_1}
\end{figure}

\clearpage
\subsubsection{Ten days ahead forecasts}

\begin{table}[h!]
	\centering
	\footnotesize
	\caption{Conditional Expected Shortfall forecasts $t+10$ for CMCS.}
	\begin{tabular}{lccccccc} \hline\hline
    & UC    & $LH = 10$ & $LH =20$ & $LH = 40$ & $LH = 60$ & $LH = 120$ & $ES_{BCBS}$ \bigstrut[b]\\
    \hline 
       ABT   & -1.44 & -1.58 & -1.58 & -1.45 & -1.49 & -1.57 & 5.42 \bigstrut[t]\\
    ADBE  & -1.82 & -1.72 & -1.72 & -1.66 & -1.69 & -1.68 & 5.86 \\
    AMAT  & -1.96 & -2.07 & -2.07 & -2.08 & -2.07 & -2.07 & 7.23 \\
    APD   & -1.36 & -1.24 & -1.24 & -1.22 & -1.21 & -1.21 & 4.23 \\
    BAC   & -1.61 & -1.70 & -1.70 & -1.64 & -1.62 & -1.63 & 5.70 \\
    BEN   & -2.78 & -1.93 & -1.93 & -2.02 & -1.95 & -1.94 & 6.88 \\
    BK    & -1.49 & -1.46 & -1.46 & -1.45 & -1.47 & -1.46 & 5.08 \\
    CL    & -1.29 & -1.25 & -1.25 & -1.31 & -1.28 & -1.30 & 4.54 \\
    CVX   & -1.42 & -1.43 & -1.43 & -1.47 & -1.42 & -1.42 & 4.99 \\
    D     & -1.17 & -1.12 & -1.12 & -1.24 & -1.28 & -1.28 & 4.39 \\
    EMR   & -1.43 & -1.49 & -1.49 & -1.50 & -1.50 & -1.51 & 5.23 \\
    FPL   & -1.25 & -1.14 & -1.14 & -1.35 & -1.39 & -1.43 & 4.90 \\
    GE    & -1.82 & -1.75 & -1.75 & -1.75 & -1.73 & -1.73 & 6.05 \\
    ITW   & -1.38 & -1.29 & -1.29 & -1.32 & -1.30 & -1.30 & 4.53 \\
    JNJ   & -0.98 & -0.99 & -0.99 & -1.02 & -1.05 & -1.04 & 3.58 \\
    JPM   & -1.42 & -1.41 & -1.41 & -1.38 & -1.38 & -1.38 & 4.80 \\
    KO    & -0.93 & -0.94 & -0.94 & -1.00 & -1.00 & -1.01 & 3.47 \\
    LOW   & -1.53 & -1.62 & -1.62 & -1.51 & -1.53 & -1.55 & 5.40 \\
    MMM   & -1.28 & -1.23 & -1.23 & -1.21 & -1.23 & -1.24 & 4.28 \\
    MRK   & -1.30 & -1.24 & -1.24 & -1.26 & -1.26 & -1.25 & 4.34 \\
    OXY   & -1.62 & -1.81 & -1.81 & -1.84 & -1.94 & -1.91 & 6.80 \\
    PCAR  & -1.55 & -1.58 & -1.58 & -1.57 & -1.55 & -1.53 & 5.38 \\
    PEP   & -0.92 & -0.96 & -0.96 & -0.95 & -0.93 & -0.94 & 3.28 \\
    PFE   & -1.14 & -1.13 & -1.13 & -1.14 & -1.08 & -1.09 & 3.95 \\
    PG    & -0.94 & -0.92 & -0.92 & -1.02 & -1.07 & -1.07 & 3.63 \\
    SO    & -1.03 & -1.09 & -1.09 & -1.07 & -1.07 & -1.07 & 3.75 \\
    SYY   & -1.36 & -1.23 & -1.23 & -1.25 & -1.30 & -1.31 & 4.48 \\
    TGT   & -1.71 & -1.89 & -1.89 & -1.92 & -1.85 & -1.85 & 6.54 \\
    TMO   & -1.51 & -1.36 & -1.36 & -1.34 & -1.33 & -1.32 & 4.63 \\
    WMT   & -1.24 & -1.21 & -1.21 & -1.26 & -1.29 & -1.26 & 4.38 \\
    XOM   & -1.25 & -1.23 & -1.23 & -1.26 & -1.24 & -1.24 & 4.32 \\
    
    \hline\hline
    \end{tabular}%
	\caption*{\scriptsize Numbers in the table correspond to out-of-sample average ES forecasts for different stocks (in rows) and liquidity horizons in columns. The last column corresponds to ES coefficient defined in \eqref{eq:ES_bcbs}.}
	\label{tab:cmcs10}%
\end{table}%

\begin{table}[h!]
	\centering \footnotesize
	\caption{Conditional Expected Shortfall forecasts $t+10$ for DFC with sample covariance estimator.}
	\begin{tabular}{lccccccc} \hline\hline
    & UC    & $LH = 10$ & $LH =20$ & $LH = 40$ & $LH = 60$ & $LH = 120$ & $ES_{BCBS}$ \bigstrut[b]\\
    \hline 
    
    ABT   & -1.51 & -1.51 & -1.51 & -1.51 & -1.51 & -1.51 & 5.23 \bigstrut[t]\\
    ADBE  & -1.64 & -1.64 & -1.64 & -1.64 & -1.64 & -1.64 & 5.68 \\
    AMAT  & -1.90 & -1.90 & -1.90 & -1.90 & -1.90 & -1.90 & 6.57 \\
    APD   & -1.22 & -1.22 & -1.22 & -1.22 & -1.22 & -1.22 & 4.24 \\
    BAC   & -1.64 & -1.64 & -1.64 & -1.64 & -1.64 & -1.64 & 5.70 \\
    BEN   & -2.30 & -2.31 & -2.31 & -2.30 & -2.30 & -2.30 & 7.97 \\
    BK    & -1.35 & -1.35 & -1.35 & -1.35 & -1.35 & -1.35 & 4.67 \\
    CL    & -1.49 & -1.49 & -1.49 & -1.49 & -1.49 & -1.49 & 5.17 \\
    CVX   & -1.33 & -1.33 & -1.33 & -1.33 & -1.33 & -1.33 & 4.61 \\
    D     & -1.06 & -1.06 & -1.06 & -1.06 & -1.06 & -1.06 & 3.66 \\
    EMR   & -1.36 & -1.36 & -1.36 & -1.36 & -1.36 & -1.36 & 4.72 \\
    FPL   & -1.02 & -1.02 & -1.02 & -1.02 & -1.02 & -1.02 & 3.53 \\
    GE    & -1.65 & -1.65 & -1.65 & -1.65 & -1.65 & -1.65 & 5.71 \\
    ITW   & -1.17 & -1.17 & -1.17 & -1.17 & -1.17 & -1.17 & 4.06 \\
    JNJ   & -0.95 & -0.95 & -0.95 & -0.95 & -0.95 & -0.95 & 3.30 \\
    JPM   & -1.32 & -1.32 & -1.32 & -1.32 & -1.32 & -1.32 & 4.56 \\
    KO    & -0.90 & -0.90 & -0.90 & -0.90 & -0.90 & -0.90 & 3.11 \\
    LOW   & -1.59 & -1.59 & -1.59 & -1.59 & -1.59 & -1.59 & 5.51 \\
    MMM   & -1.14 & -1.14 & -1.14 & -1.14 & -1.14 & -1.14 & 3.94 \\
    MRK   & -1.20 & -1.20 & -1.20 & -1.20 & -1.20 & -1.20 & 4.17 \\
    OXY   & -1.54 & -1.54 & -1.54 & -1.54 & -1.54 & -1.54 & 5.35 \\
    PCAR  & -1.49 & -1.49 & -1.49 & -1.49 & -1.49 & -1.49 & 5.15 \\
    PEP   & -0.88 & -0.88 & -0.88 & -0.88 & -0.88 & -0.88 & 3.03 \\
    PFE   & -1.12 & -1.12 & -1.12 & -1.12 & -1.12 & -1.12 & 3.88 \\
    PG    & -0.94 & -0.94 & -0.94 & -0.94 & -0.94 & -0.94 & 3.26 \\
    SO    & -0.99 & -0.99 & -0.99 & -0.99 & -0.99 & -0.99 & 3.42 \\
    SYY   & -1.09 & -1.09 & -1.09 & -1.09 & -1.09 & -1.09 & 3.77 \\
    TGT   & -1.49 & -1.49 & -1.49 & -1.49 & -1.49 & -1.49 & 5.18 \\
    TMO   & -1.31 & -1.31 & -1.31 & -1.31 & -1.31 & -1.31 & 4.53 \\
    WMT   & -1.22 & -1.22 & -1.22 & -1.22 & -1.22 & -1.22 & 4.21 \\
    XOM   & -1.15 & -1.15 & -1.15 & -1.15 & -1.15 & -1.15 & 3.97 \\    
    \hline\hline
    \end{tabular}%
	\caption*{\scriptsize Numbers in the table correspond to out-of-sample average ES forecasts for different stocks (in rows) and liquidity horizons in columns. The last column corresponds to ES coefficient defined in \eqref{eq:ES_bcbs}.}
	\label{tab:bt_sample10}%
\end{table}%

\begin{table}[h!]
	\centering \footnotesize
	\caption{Conditional Expected Shortfall forecasts $t+10$ for DFC with truncation $T/4$.}
	\begin{tabular}{lccccccc} \hline\hline
    & UC    & $LH = 10$ & $LH =20$ & $LH = 40$ & $LH = 60$ & $LH = 120$ & $ES_{BCBS}$ \bigstrut[b]\\
    \hline 
    ABT   & -2.94 & -2.99 & -2.99 & -2.98 & -2.96 & -2.99 & 10.36 \bigstrut[t]\\
    ADBE  & -4.03 & -4.05 & -4.05 & -4.12 & -4.23 & -4.26 & 14.69 \\
    AMAT  & -4.29 & -4.53 & -4.53 & -4.41 & -4.47 & -4.53 & 15.73 \\
    APD   & -2.84 & -2.61 & -2.61 & -2.88 & -2.89 & -2.86 & 9.91 \\
    BAC   & -3.50 & -3.58 & -3.58 & -3.62 & -3.65 & -3.57 & 12.49 \\
    BEN   & -3.09 & -3.46 & -3.46 & -3.44 & -3.40 & -3.41 & 11.92 \\
    BK    & -3.09 & -3.08 & -3.08 & -2.91 & -3.05 & -3.04 & 10.63 \\
    CL    & -2.25 & -2.21 & -2.21 & -2.24 & -2.22 & -2.22 & 7.69 \\
    CVX   & -2.79 & -2.89 & -2.89 & -2.95 & -2.95 & -2.95 & 10.23 \\
    D     & -1.32 & -1.45 & -1.45 & -1.42 & -1.44 & -1.33 & 12.00 \\
    FPL   & -2.42 & -2.38 & -2.38 & -2.44 & -2.43 & -2.39 & 8.41 \\
    GE    & -3.30 & -3.17 & -3.17 & -3.36 & -3.29 & -3.27 & 11.45 \\
    ITW   & -2.98 & -2.71 & -2.71 & -2.85 & -2.80 & -2.75 & 9.85 \\
    JNJ   & -1.89 & -1.93 & -1.93 & -1.64 & -1.69 & -1.70 & 6.10 \\
    JPM   & -3.07 & -3.02 & -3.02 & -2.93 & -2.99 & -3.11 & 10.68 \\
    KO    & -1.94 & -1.91 & -1.91 & -1.94 & -1.94 & -1.95 & 6.73 \\
    LOW   & -2.61 & -2.58 & -2.58 & -2.51 & -2.51 & -2.50 & 8.74 \\
    MMM   & -2.31 & -2.20 & -2.20 & -2.23 & -2.29 & -2.30 & 7.93 \\
    MRK   & -2.44 & -2.54 & -2.54 & -2.44 & -2.46 & -2.49 & 8.68 \\
    OXY   & -3.26 & -3.25 & -3.25 & -3.27 & -3.18 & -3.17 & 11.19 \\
    PCAR  & -2.79 & -2.60 & -2.60 & -2.65 & -2.60 & -2.58 & 9.16 \\
    PEP   & -1.68 & -1.65 & -1.65 & -1.65 & -1.70 & -1.71 & 5.90 \\
    PG    & -2.02 & -1.91 & -1.91 & -1.96 & -1.95 & -1.98 & 6.84 \\
    SO    & -2.44 & -2.40 & -2.40 & -2.46 & -2.45 & -2.48 & 8.54 \\
    SYY   & -1.69 & -1.77 & -1.77 & -1.77 & -1.77 & -1.78 & 6.17 \\
    TGT   & -3.23 & -3.30 & -3.30 & -3.24 & -3.27 & -3.28 & 11.42 \\
    TMO   & -2.77 & -2.71 & -2.71 & -2.71 & -2.80 & -2.80 & 9.69 \\
    WMT   & -2.36 & -2.23 & -2.23 & -2.27 & -2.27 & -2.28 & 7.87 \\
    XOM   & -2.42 & -2.57 & -2.57 & -2.50 & -2.50 & -2.54 & 8.86 \\
    \hline\hline
    \end{tabular}%
	\caption*{\scriptsize Numbers in the table correspond to out-of-sample average ES forecasts for different stocks (in rows) and liquidity horizons in columns. The last column corresponds to ES coefficient defined in \eqref{eq:ES_bcbs}.}
	\label{tab:bt_tr4_10}%
\end{table}%

\begin{table}[h!]
	\centering \footnotesize
	\caption{Conditional Expected Shortfall forecasts $t+10$ for DFC with truncation $T/2$.}
	\begin{tabular}{lccccccc} \hline\hline
    & UC    & $LH = 10$ & $LH =20$ & $LH = 40$ & $LH = 60$ & $LH = 120$ & $ES_{BCBS}$ \bigstrut[b]\\
    \hline 
    ABT   & -3.12 & -3.00 & -3.00 & -3.04 & -3.06 & -3.02 & 10.53 \bigstrut[t]\\
    ADBE  & -4.06 & -4.09 & -4.09 & -3.99 & -4.04 & -3.98 & 14.06 \\
    AMAT  & -4.50 & -4.23 & -4.23 & -4.60 & -4.62 & -4.66 & 15.99 \\
    APD   & -2.48 & -2.90 & -2.90 & -2.75 & -2.75 & -2.83 & 9.84 \\
    BAC   & -3.70 & -3.70 & -3.70 & -3.59 & -3.63 & -3.56 & 12.54 \\
    BEN   & -3.37 & -3.41 & -3.41 & -3.41 & -3.42 & -3.43 & 11.89 \\
    BK    & -3.01 & -3.06 & -3.06 & -2.96 & -3.02 & -2.97 & 10.51 \\
    CL    & -2.26 & -2.18 & -2.18 & -2.22 & -2.23 & -2.24 & 7.72 \\
    CVX   & -2.92 & -2.94 & -2.94 & -2.95 & -2.96 & -2.95 & 10.27 \\
    D     & -0.39 & -1.02 & -1.02 & -1.42 & -1.45 & -1.35 & 12.53 \\
    FPL   & -2.30 & -2.41 & -2.41 & -2.32 & -2.41 & -2.40 & 8.32 \\
    GE    & -3.15 & -3.24 & -3.24 & -3.27 & -3.30 & -3.30 & 11.47 \\
    ITW   & -2.95 & -3.05 & -3.05 & -2.90 & -2.93 & -2.95 & 10.34 \\
    JNJ   & -1.82 & -1.85 & -1.85 & -1.70 & -1.68 & -1.69 & 6.05 \\
    JPM   & -2.92 & -3.17 & -3.17 & -3.04 & -2.98 & -3.13 & 10.89 \\
    KO    & -1.91 & -1.92 & -1.92 & -1.94 & -1.93 & -1.94 & 6.72 \\
    LOW   & -2.45 & -2.62 & -2.62 & -2.47 & -2.48 & -2.41 & 8.64 \\
    MMM   & -2.37 & -2.29 & -2.29 & -2.20 & -2.15 & -2.05 & 7.49 \\
    MRK   & -2.50 & -2.50 & -2.50 & -2.52 & -2.52 & -2.51 & 8.72 \\
    OXY   & -3.29 & -3.23 & -3.23 & -3.06 & -3.13 & -3.18 & 11.07 \\
    PCAR  & -2.90 & -2.99 & -2.99 & -2.91 & -2.89 & -2.80 & 10.01 \\
    PEP   & -1.75 & -1.72 & -1.72 & -1.72 & -1.78 & -1.78 & 6.13 \\
    PG    & -1.94 & -1.98 & -1.98 & -2.00 & -1.99 & -1.97 & 6.91 \\
    SO    & -2.42 & -2.45 & -2.45 & -2.39 & -2.43 & -2.40 & 8.49 \\
    SYY   & -1.78 & -1.77 & -1.77 & -1.73 & -1.73 & -1.74 & 6.06 \\
    TGT   & -3.26 & -3.22 & -3.22 & -3.17 & -3.20 & -3.18 & 11.09 \\
    TMO   & -2.84 & -2.80 & -2.80 & -2.84 & -2.76 & -2.75 & 9.68 \\
    WMT   & -2.23 & -2.22 & -2.22 & -2.27 & -2.25 & -2.25 & 7.84 \\
    XOM   & -2.64 & -2.61 & -2.61 & -2.59 & -2.49 & -2.43 & 8.76 \\
    \hline\hline
    \end{tabular}%
	\caption*{\scriptsize Numbers in the table correspond to out-of-sample average ES forecasts for different stocks (in rows) and liquidity horizons in columns. The last column corresponds to ES coefficient defined in \eqref{eq:ES_bcbs}.}
	\label{tab:bt_tr2_10}%
\end{table}%


\begin{figure}[h!]
	\caption{Time series of 10 days ahead ES forecasts for BAC:  DFC with truncation $T/2$.}
	\centering
	\includegraphics[trim={0cm 0cm 0cm 0cm},clip,width = 0.8\textwidth]{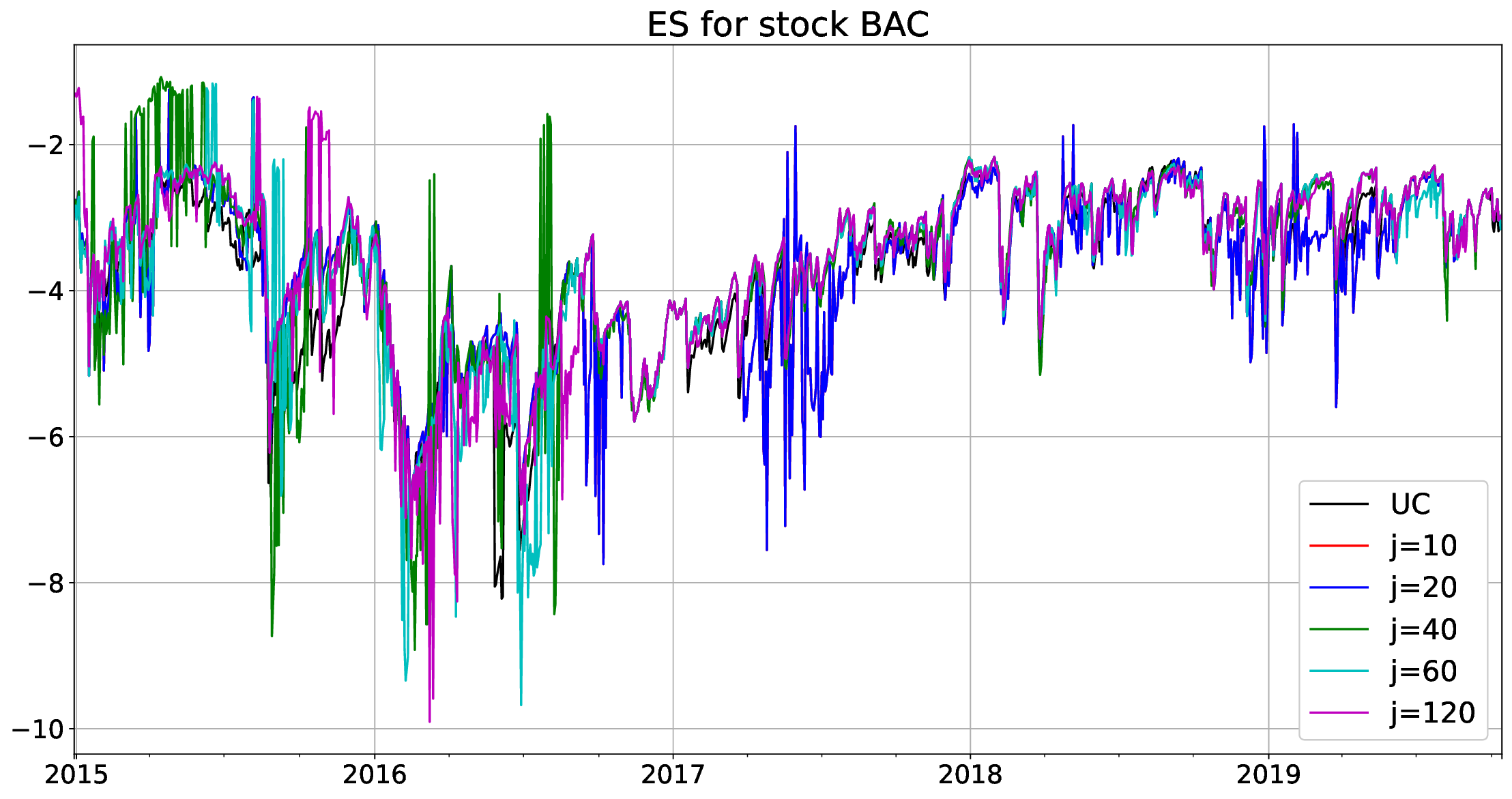}\vspace{-5pt}
	\caption*{\scriptsize 
		Coloured lines correspond to 10-day ahead ES forecasts for different liquidity horizons.}
	\label{fig:BT2_BAC}
\end{figure}

\begin{figure}[h!]
	\caption{Time series of 10 days ahead ES forecasts for BAC: DFC with sample covariance estimator.}
	\centering
	\includegraphics[trim={0cm 0cm 0cm 0cm},clip,width = 0.9\textwidth]{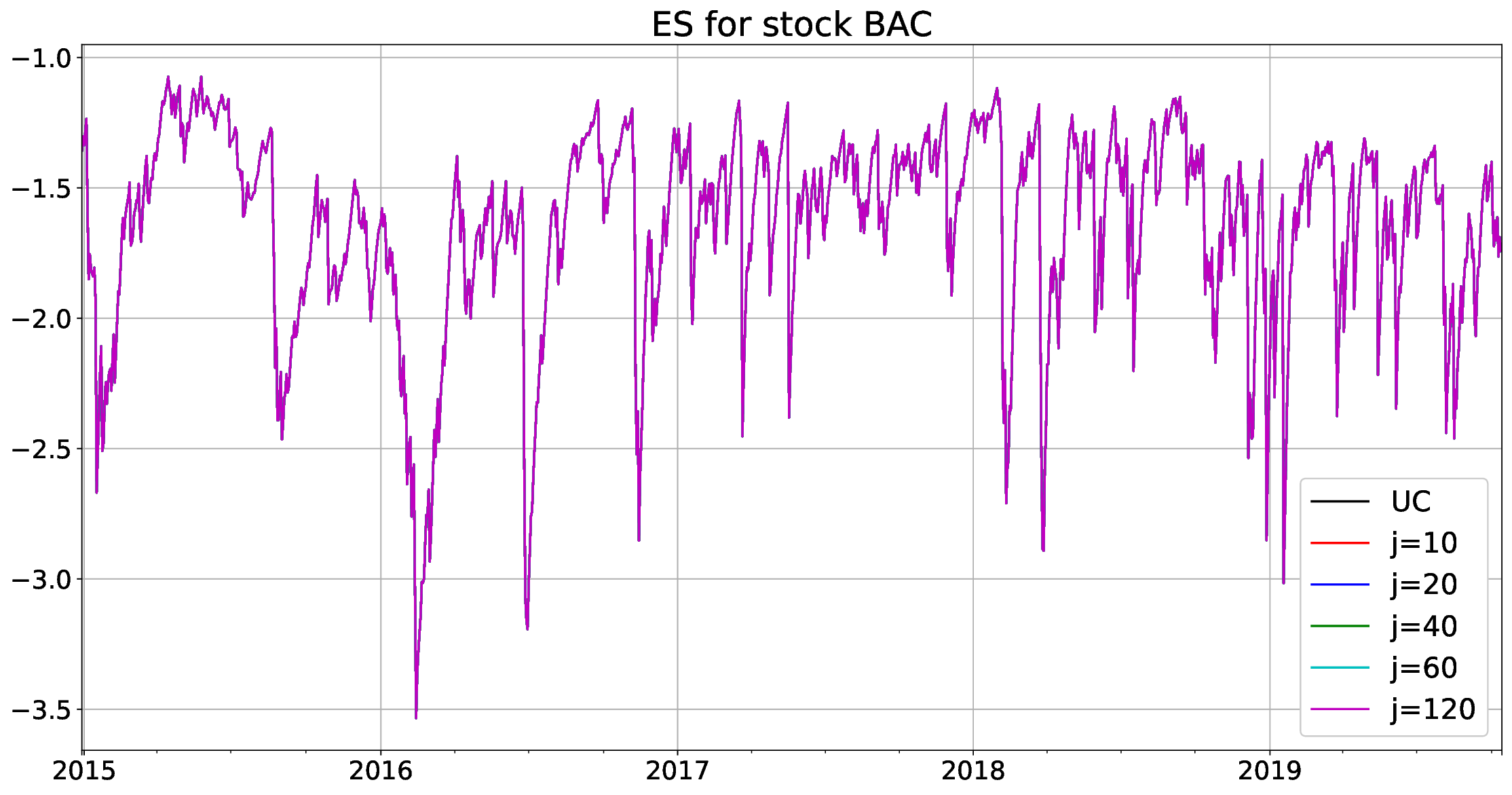}\vspace{-5pt}
	\caption*{\scriptsize 
		Coloured lines correspond to 1-day ahead ES forecasts for different liquidity horizons.}
	\label{fig:BTs_BAC}
\end{figure}

\begin{figure}[h!]
	\caption{Method selection for 10 day ahead ES forecasts for different liquidity horizons:  DFC with truncation $T/2$.}
	\centering
	\includegraphics[trim={0cm 0cm 0cm 0cm},clip,width = 0.9\textwidth]{t10/BT_stocks_BT_TR1955.eps}
	\caption*{\scriptsize Heatmaps correspond to the liquidity horizons as specified in Table \ref{tab:LH}. On each heatmap the x-axis corresponds to stocks and the y-axis to models. The heatmap cells correspond to the average number of the out-of-sample periods where the method is included in the MCS. The warmer the colour, the more frequently was the method selected for the forecast combination.}
	\label{fig:BT2_stocks}
\end{figure}

\begin{figure}[h!]
	\caption{Method selection for 10 day ahead ES forecasts for different liquidity horizons:  DFC with sample covariance estimator.}
	\centering
	\includegraphics[trim={0cm 0cm 0cm 0cm},clip,width = 0.9\textwidth]{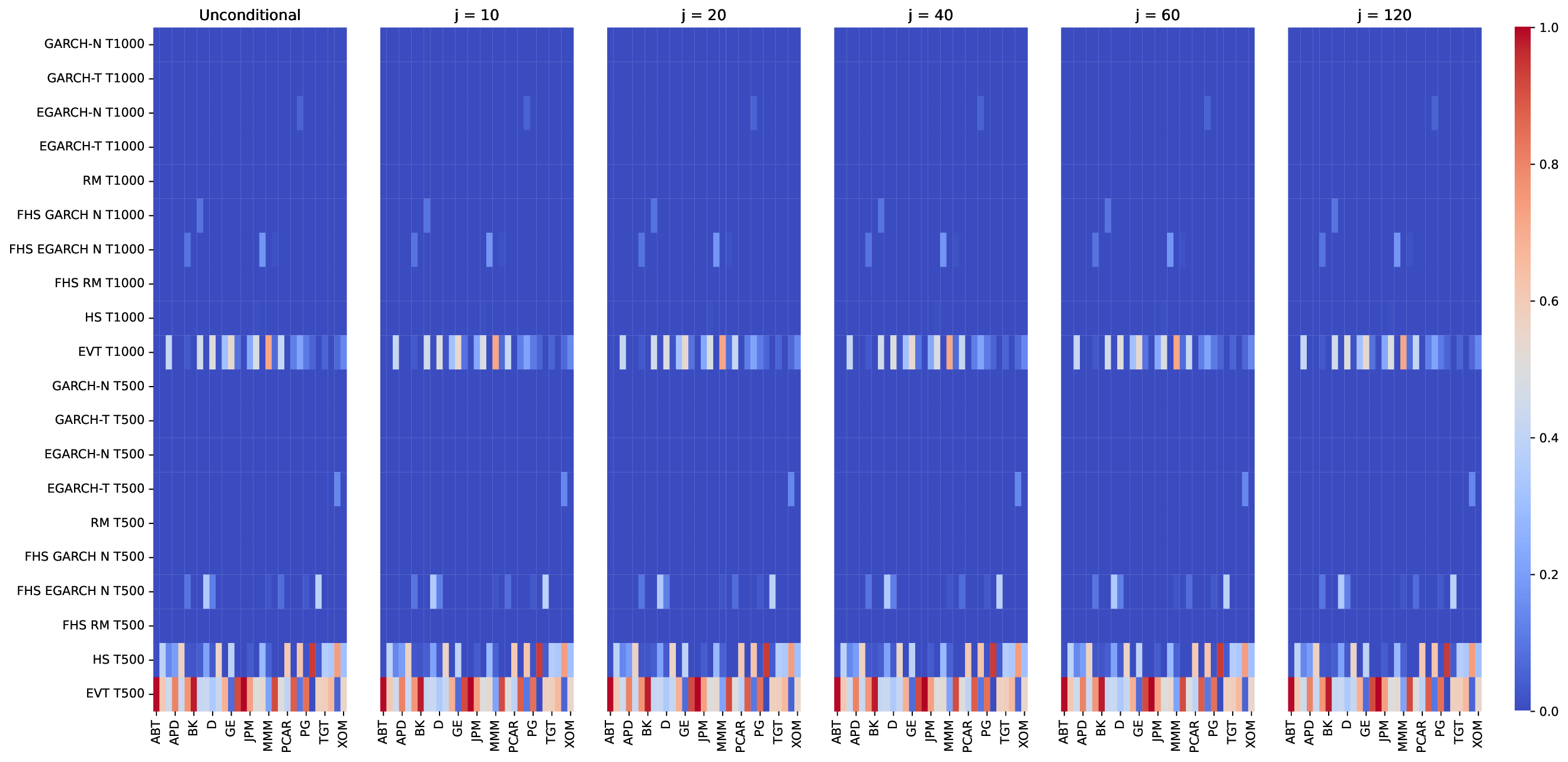}
	\caption*{\scriptsize Heatmaps correspond to the liquidity horizons as specified in Table \ref{tab:LH}. On each heatmap the x-axis corresponds to stocks and the y-axis to models. The heatmap cells correspond to the average number of the out-of-sample periods where the method is included in the MCS. The warmer the colour, the more frequently was the method selected for the forecast combination.}
	\label{fig:BTs_stocks}
\end{figure}

\end{appendix}

\end{document}